\documentclass[]{revtex4-2}
\usepackage{graphicx} 
\usepackage{amsmath}
\usepackage{amssymb}
\usepackage{amsthm}
\usepackage{braket}
\usepackage{comment}
\usepackage{xcolor}
\usepackage{hyperref}

\newtheorem{definition}{Definition}[section]
\newtheorem{example}[definition]{Example}

\newtheorem{lemma}{Lemma}[section]

\theoremstyle{plain}
\newtheorem{theorem}{Theorem}
\newtheorem{proposition}{Proposition}[section]

\begin{document}
\title{No classical particle limit for massless quanta}

\author{Riccardo Falcone}
\thanks{These authors contributed equally to this work.}
\affiliation{Institute for Quantum Optics and Quantum Information (IQOQI), Austrian Academy of Sciences, Boltzmanngasse 3, 1090 Vienna, Austria}

\author{Simon Fuchs}
\thanks{These authors contributed equally to this work.}
\affiliation{Institute for Theoretical Physics, ETH Zurich, Wolfgang-Pauli-Strasse 27, 8093 Zurich, Switzerland}

\begin{abstract}
We investigate whether relativistic massless classical particles may emerge as the classical limit of massless quanta. To address this question independently of any specific dynamics, environment, or pointer basis, we develop an axiomatic and operational framework for the coarse-graining approach. In this formulation, a candidate classical phase space is taken as the outcome space of a POVM subject to minimal classicality and covariance requirements. Applying this framework to the Poincaré group, we prove a no-go theorem for massless particles: the covariance requirement is incompatible with the operational conditions for classicality. Conversely, we show that the obstruction can be avoided for massive particles by constructing coarse-graining schemes satisfying the minimal classicality requirements. The no-go theorem leaves open field-like limits of massless quanta, such as the emergence of the electromagnetic field, while ruling out classical massless particles, such as classical photons.

\end{abstract}
\maketitle

\makeatletter
\renewcommand{\l@subsubsection}[2]{}
\makeatother
\tableofcontents

\section{Introduction}\label{Introduction}

Understanding the classical limit of quantum theory involves two related but distinct questions. The first is \textit{how} an effectively classical description can arise from an underlying quantum theory. The second is \textit{which} classical theory emerges in the first place. In ordinary nonrelativistic quantum mechanics this second question is usually taken for granted rather than treated as an independent problem. For a single nonrelativistic particle with Hilbert space $L^2(\mathbb{R}^3)$, the natural classical target is Hamiltonian mechanics on the canonical phase space $T^*\mathbb{R}^3$, with position and momentum as classical variables \cite{landsman2007between}.

The situation is conceptually different in second-quantized quantum mechanics. In this framework, fields and particles are unified: the same quantum degrees of freedom may be described in terms of field operators or in terms of particle-number sectors \cite{fetter1971quantum}. At the classical level, however, particles and fields are different kinds of theories. A classical particle theory is formulated in terms of positions, momenta, and trajectories or phase-space distributions, whereas a classical field theory is formulated in terms of field configurations. A complete account of the quantum-to-classical transition should therefore explain not only why classical behaviour emerges, but also when the emergent classical description is particle-like and when it is field-like \cite{Hepp1974}. The issue becomes even more central in relativistic quantum theory, where the second-quantized formulation is the natural one \cite{weinberg1995quantum,duncan2012conceptual}.

The general question of \textit{how} quantum systems become effectively classical has been extensively studied \cite{landsman2007between}. Environment-induced decoherence gives a dynamical explanation of this process, showing how interactions with uncontrolled environmental degrees of freedom suppress interference and stabilize classical records \cite{Zeh1970, PhysRevD.24.1516, PhysRevD.26.1862, Joos1985, RevModPhys.75.715}. The coarse-graining approach gives a complementary perspective, based on finite measurement resolution or fuzziness: rather than asking how coherences are dynamically destroyed, it asks which statistics remain operationally accessible when measurements have finite resolution, and whether the corresponding measurements admit an effective classical interpretation \cite{Omnes1988, peres1995quantum, PhysRevLett.99.180403, PhysRevLett.101.090403, kofler2010macroscopic, PhysRevLett.112.010402, bibak2026classicallimitquantummechanics}. In this sense, decoherence provides a dynamical route to classicality, while coarse graining provides an operational perspective on the classical behaviour of finite-resolution quantum measurements.

The question of the target theory enters these approaches in different ways. In decoherence, the classical description is tied to the pointer states selected by the system-environment dynamics. In coarse graining, it is encoded instead in the choice of finite-resolution observables whose outcomes are interpreted as regions of a classical state space. Different choices of finite-resolution observables can select different target classical descriptions. The existence of a classical probability measure on such a space is not sufficient by itself: the measurements must also have the operational properties needed to interpret their outcomes as observations of classical alternatives.

In the nonrelativistic single-particle case, canonical coherent states, also known as Schr\"odinger or Weyl--Heisenberg coherent states, play a central role in both approaches \cite{schrodinger1928collected,perelomov1977generalized}. This family consists of displaced minimum-uncertainty Gaussian wave packets $\ket{\mathbf{x},\mathbf{p}}=e^{-\frac{i}{\hbar}\mathbf{x}\cdot\hat{\mathbf{P}}}e^{\frac{i}{\hbar}\mathbf{p}\cdot\hat{\mathbf{X}}}\ket{\mathbf{0},\mathbf{0}}$, where $\ket{\mathbf{0},\mathbf{0}}$ is a Gaussian wave packet centered at the origin, and $\hat{\mathbf X}$ and $\hat{\mathbf P}$ are the canonical position and momentum operators. Each state is localized around $\mathbf{x}$ in position space and around $\mathbf{p}$ in momentum space, and the family is covariant under phase-space translations. In decoherence, localized Gaussian wave packets arise as natural candidates for pointer states in nonrelativistic particle models \cite{PhysRevLett.70.1187,PhysRevE.50.2538,PhysRevLett.85.3552}. In the coarse-graining approach, the phase-space localization and covariance of these states make them a natural kernel for phase-space POVMs, assigning to each cell $C$ an effect $\hat{\Pi}_{C}=\int_{C}\frac{d^3\mathbf{x}\,d^3\mathbf{p}}{(2\pi\hbar)^3}\ket{\mathbf{x},\mathbf{p}}\bra{\mathbf{x},\mathbf{p}}$. For sufficiently coarse cells, the corresponding measurement statistics admit an effective classical probabilistic description on the canonical position-momentum phase space. In this limit, the outcome $C$ of the POVM can be interpreted as the finite-resolution statement that the particle lies in the classical phase-space cell $C$.

Here, we ask whether the same particle-like target exists for relativistic massless particles. At first sight, one might expect a positive answer: relativistic point particles are standard objects in classical mechanics and relativity, where they are described by spacetime trajectories and four-momenta \cite{landau1975classical,goldstein2002classical}. However, the most familiar classical limits of massless quantum degrees of freedom are often field-like. For instance, in classical electrodynamics, massive charged matter is described in terms of particles or currents, whereas the massless electromagnetic sector is described by a field. The question we address here is not whether massless quantum systems can have a classical limit at all, but whether they can have a classical \emph{particle} phase-space limit satisfying minimal covariance and operational requirements.

We show that this particle-like target is obstructed already at the kinematical level. The obstruction first appears if one tries to construct relativistic analogues of the nonrelativistic coherent states $\ket{\mathbf{x},\mathbf{p}}$. Poincar\'e covariance requires the seed state $\ket{\mathbf{0},\mathbf{p}}$ associated with a null momentum $p^\mu$ to be invariant under the corresponding little group. This forces its momentum-space wavefunction $\psi_{\mathbf{0},\mathbf{p}}(\mathbf{k})$ to depend on $\mathbf{k}$ and $\mathbf{p}$ only through the invariant combination $k^\mu p_\mu = |\mathbf{k}|\,|\mathbf{p}|-\mathbf{k}\cdot\mathbf{p}$. Along the direction $\mathbf{k}\parallel\mathbf{p}$ this invariant is independent of the magnitude $|\mathbf{k}|$. The wavefunction therefore cannot be localized around a definite value of $|\mathbf{p}|$. Thus, the same covariance condition that would make the states transform as classical massless particles prevents them from being localized in the full momentum space. Consequently, coarse-grained measurements cannot resolve different energy scales along the same null direction, even though these correspond to distinct momentum labels in a classical phase space.

The same obstruction is relevant for the decoherence program: unlike in the nonrelativistic massive case, there is no Poincar\'e-covariant analogue of Weyl--Heisenberg coherent states, optimally localized in both position and momentum, that could serve as natural pointer states for massless particles. Hence, the standard route to a particle-like classical phase space, whether formulated in terms of coarse graining or in terms of decoherence, is already obstructed at the kinematical level.

The failure of a coherent-state construction is not, by itself, a general no-go theorem. One might still hope that a more general POVM, not built from rank-one coherent-state projectors, could define a particle-like classical phase space. For this reason, we formulate the problem directly at the level of phase-space effects. A coarse-graining scheme consists of a candidate classical phase space $\Gamma$, a POVM $C\mapsto\hat\Pi_C$, and a family $\mathcal P_{\rm CG}$ of admissible coarse partitions. We distinguish the measurements in this family, whose outcomes are intended to admit a direct classical interpretation, from finer events whose probabilities may only be inferred from repeated coarse measurements and statistical post-processing.

Within this framework, we impose four main requirements: (i) the assignment $C\mapsto\hat{\Pi}_{C}$ must be Poincar\'e covariant; (ii) for any sufficiently coarse partition $\{ C_i \}_{i=1}^N$, the corresponding outcomes should be mutually exclusive and stable under immediate repetition at the coarse-grained level; (iii) the momentum marginal must agree, up to coarse-grained accuracy, with the standard momentum spectral measure, $ \hat{\Pi}_{\mathbb{R}^3\times Q} \approx \int_Q \frac{d^3\mathbf{k}}{2|\mathbf{k}|} \ket{\mathbf{k}}\bra{\mathbf{k}}$; and (iv) the coarse-grained sector must carry a minimal amount of directional information: it must contain at least one yes/no alternative approximately distinguishing into which hemisphere the particle travels. We prove that, for massless particles, these requirements cannot be satisfied simultaneously. More precisely, the two triples of requirements (i),(ii),(iv) and (i),(iii),(iv) are independently inconsistent. Notably, the obstruction survives the inclusion of multiple helicities and holds for any instrument compatible with the POVM $\hat \Pi_C$.

We then show that the obstruction is specific to massless particles. To this end, we construct coarse-graining schemes for massive particles in both the nonrelativistic and relativistic settings. In the nonrelativistic case, Weyl--Heisenberg coherent states have momentum profiles depending on $\mathbf{k}-\mathbf{p}$ and can therefore be localized around $\mathbf{k}=\mathbf{p}$. In the relativistic massive case, the invariant $k^\mu p_\mu=\sqrt{m^2+|\mathbf{k}|^2}\sqrt{m^2+|\mathbf{p}|^2}-\mathbf{k}\cdot\mathbf{p}$ has compact level sets and a unique minimum at $\mathbf{k}=\mathbf{p}$, thus allowing for localized covariant particle coherent states.

The paper is organized as follows. In Sect.~\ref{Phasespace_measurements} we define the candidate classical target phase space for a massless particle and, assuming that this description emerges from an underlying quantum theory, we represent measurements of phase-space events by a POVM. In Sect.~\ref{Minimal_requirements_for_Poincaré_group} we give an operational formulation of Poincar\'e covariance, requiring the emergent classical transformation law on phase space to be induced by the unitary Poincar\'e transformations of the underlying quantum theory. In Sect.~\ref{Coarse_partitions} we define the family of coarse partitions, distinguishing directly classical coarse measurements from finer phase-space events that may still be probed but need not themselves be classical. In Sect.~\ref{requirements} we state the minimal operational requirements for a coarse-grained classical particle interpretation: weak non-invasive repeatability, weak compatibility with the quantum momentum observable, and minimal directional information gain. In Sect.~\ref{nogo_theorem} we prove that these assumptions are mutually inconsistent for a massless, fixed-helicity representation, and Sect.~\ref{sec:multiple-helicities} extends the result to multiple helicities. In Sect.~\ref{massive} we construct coherent-state coarse-graining schemes for massive particles in the nonrelativistic and relativistic settings, showing that the minimal classicality requirements can be satisfied.
Conclusions are drawn in Sect.~\ref{Conclusions}, where we also comment on the prospects for a classical field limit of massless bosons. Appendices~\ref{Regularity_conditions}--\ref{appendix_minimal_information} provide alternative or more detailed operational justifications for some assumptions used in the main text. In Appendix~\ref{stronger_classicality_nonrelativistic}, we discuss stronger classicality requirements for massive nonrelativistic particles. Appendix~\ref{proofs} contains the technical proofs in the order of the corresponding sections.

\section{Phase-space measurements}\label{Phasespace_measurements}

In this section, we introduce the basic setting for a particle-like classical limit. We first specify the candidate classical target phase space for a massless particle, whose points are labelled by position and nonzero momentum. We then assume an underlying quantum theory and formalize measurements of phase-space events by a POVM. This provides the microscopic probabilistic structure on which the later covariance and classicality requirements will be imposed.

We start by specifying the target classical state space. For a massless particle of arbitrary helicity $\lambda \in \mathbb R$, we take this space to be
\begin{equation}\label{Gamma_R_R_}
    \Gamma := \mathbb{R}^3\times\mathbb{R}^3_*, \qquad \mathbb{R}^3_*:=\mathbb{R}^3\setminus\{\mathbf 0\}.
\end{equation}
The variables $\mathbf x, \mathbf p$ have their usual classical interpretation: $\mathbf x$ is the position, $\mathbf p/|\mathbf p|$ the direction of propagation, and $|\mathbf p|$ the energy, in units $c=1$. The classical helicity $\lambda$ is a label rather than a degree of freedom, since it is invariant under translations, rotations and boosts, but as we will see later, it affects how $\mathbf x$ transforms. In contrast to its quantum counterpart, classical helicity is not discrete and can take on any value \cite{andrzejewski2015chiral}.

We model measurements as attempts to locate the particle within this phase-space $\Gamma$. The elementary classical yes/no observable associated with a measurable phase-space cell $C\subseteq\Gamma$ is the effect $E_{C}$, corresponding to the question whether the particle lies in $C$. A complete finite-resolution measurement is then represented by a finite measurable partition of $\Gamma$. We denote the set of such partitions by
\begin{equation}
    \mathcal{P}_\Gamma := \left\{ \{C_i\}_{i=1}^N \,\middle|\, N\in\mathbb N,\  C_i\in\sigma_\Gamma,\  \bigsqcup_{i=1}^NC_i=\Gamma \right\}.
\end{equation}
Here, $\sigma_{\Gamma}$ is the $\sigma$-algebra of measurable subsets of $\Gamma$ with respect to the Lebesgue measure $d\nu(\mathbf x,\mathbf p) =  d^3\mathbf x d^3 \mathbf p $.

We regard the classical phase-space theory described above as an emergent description of an underlying fundamental theory, which we take to be quantum mechanical. The operational content of a particle-like classical limit is therefore encoded in an assignment of quantum effects to classical phase-space events. We assume the existence of a positive operator-valued measure (POVM) $C \mapsto \hat \Pi_{C}$, which assigns to every measurable subset of the classical phase-space $C\in\sigma_\Gamma$ an effect $\hat{\Pi}_{C} \ge 0$, such that
\begin{equation}\label{identity_C_i}
\hat{\Pi}_\Gamma=\hat{\mathbb I},
\end{equation}
and, for every countable family $\{C_n\}_{n\in\mathbb N}\subset\sigma_\Gamma$ of pairwise disjoint sets,
\begin{equation}\label{additivity}
\hat{\Pi}_{\bigsqcup_{n\in\mathbb N}C_n}=\sum_{n\in\mathbb N}\hat{\Pi}_{C_n}.
\end{equation}

The effects $\hat \Pi_{C}$ act on the Hilbert space of the positive-energy, helicity $h$ (possibly different from the classical target-space helicity $\lambda$), massless representation of the Poincar\'e group $\mathcal H = L^2(V_0^+,d\mu_0)$, with $V_0^+:=\{k\in\mathbb R^{1,3}:k^0=|\mathbf k|\}$ and where $d\mu_0(k)=d^3\mathbf k/(2|\mathbf k|)$ is the Lorentz-invariant measure on the forward light cone \cite{wigner1939unitary}. This Hilbert space should be interpreted as the one-particle sector of the free Fock space, not as a standalone theory of relativistic quantum mechanics. Given a density operator $\hat{\rho}$ on $\mathcal H$, the POVM assigns probabilities to an experiment $P \in \mathcal{P}_\Gamma$ according to the Born rule, $p_{\hat \rho}(C) = \text{tr}(\hat \rho \hat \Pi_{C})$ for $C \in P$ \cite{holevo2011probabilistic}.

We assume that the POVM $C\mapsto\hat\Pi_{C}$ admits an operator-valued phase-space density representation in the weak sense: there exist a positive measure $d\mu(\mathbf x,\mathbf p)=\rho(\mathbf x,\mathbf p)\,d\nu(\mathbf x,\mathbf p)$ and a family of positive sesquilinear forms $(\mathbf x,\mathbf p)\mapsto\hat\pi_{(\mathbf x,\mathbf p)}$ such that, for every $\ket\psi,\ket\phi\in\mathcal H$, the function $(\mathbf x,\mathbf p)\mapsto\braket{\psi|\hat\pi_{(\mathbf x,\mathbf p)}|\phi}$ is defined $\mu$-almost everywhere, with  $\braket{\psi|\hat{\Pi}_{C}|\phi}=\int_{C}d\mu(\mathbf x,\mathbf p)\,\braket{\psi|\hat\pi_{(\mathbf x,\mathbf p)}|\phi}$ for every $C\in\sigma_\Gamma$. We write
\begin{equation}\label{Pi_C_density}
\hat{\Pi}_{C}=\int_{C}d\mu(\mathbf x,\mathbf p)\,\hat \pi_{(\mathbf x,\mathbf p)}=\int_{C}d^3\mathbf x\,d^3\mathbf p\,\rho(\mathbf x,\mathbf p)\hat \pi_{(\mathbf x,\mathbf p)}.
\end{equation}

In particular, the normalization condition \eqref{identity_C_i} becomes
\begin{equation}\label{identity}
\int_\Gamma d\mu(\mathbf x,\mathbf p)\,\hat \pi_{(\mathbf x,\mathbf p)}=\hat{\mathbb I}.
\end{equation}
Sufficient regularity conditions on the POVM $C\mapsto\hat\Pi_{C}$ guaranteeing this representation are discussed in Appendix~\ref{Regularity_conditions}.

At this point, one could impose a rank-one, or atomic, form for the operator-valued density by requiring $\hat\pi_{(\mathbf x,\mathbf p)}=\ket{\mathbf x,\mathbf p}\bra{\mathbf x,\mathbf p}$, for some family of normalized states $\ket{\mathbf x,\mathbf p}\in\mathcal H$. In this case, Eq.~\eqref{Pi_C_density} would reduce to a coherent-state POVM of the usual Perelomov type \cite{perelomov1977generalized}. We do not impose this additional assumption. Instead, we allow the density $\hat\pi_{(\mathbf x,\mathbf p)}$ to be an arbitrary positive operator-valued density compatible with the normalization condition \eqref{identity}.

\section{Covariance under Poincar\'e group}\label{Minimal_requirements_for_Poincaré_group}

In this section, we formulate the covariance requirement connecting the emergent classical phase-space description with the underlying quantum theory. The guiding idea is that classical Poincar\'e covariance should arise from the unitary Poincar\'e representation acting on the quantum Hilbert space. Operationally, Poincar\'e covariance means that inertial observers related by a Poincar\'e transformation assign compatible probabilities to the corresponding physical events. Since generic Lorentz transformations do not preserve the reference hypersurface $t=0$, this requirement cannot be imposed as a global transformation law on a single equal-time phase space. We therefore first formulate it locally, for transformations that map one point of the reference hypersurface to another. We then express this local covariance condition in terms of the probabilities assigned to small phase-space neighbourhoods around Poincar\'e-related points.

Before discussing covariance specifically for relativistic massless particles, we briefly recall how covariance is formulated for a general phase-space POVM \cite{holevo2011probabilistic}. Let $\Gamma$ be the classical phase space associated with a candidate emergent description, not necessarily the one given by Eq.~\eqref{Gamma_R_R_}, and let $z\in\Gamma$ denote a generic phase-space point. Let $G$ be a group acting measurably on $\Gamma$, and suppose that the corresponding quantum transformations are implemented on the Hilbert space $\mathcal H$ by a unitary representation
\begin{equation} \label{U_g}
g \mapsto \hat U(g),
\qquad
\hat U(g)\hat U(g')=\hat U(gg'),
\qquad
\hat U(g^{-1})=\hat U^\dagger(g).
\end{equation}
As in Sect.~\ref{Phasespace_measurements}, quantum measurements on $\Gamma$ are represented by a POVM $C\mapsto\hat\Pi_{C}$.
\begin{definition}[$G$-covariant phase-space POVM]\label{def:covariant-phase-space-povm}
Let $G$ be a group acting on $\Gamma$, and let $g\mapsto\hat U(g)$ be a unitary representation of $G$ on $\mathcal H$. A POVM $C\mapsto\hat\Pi_{C}$ on $\Gamma$ is called $G$-covariant if, for every $g\in G$ and every $C\in\sigma_\Gamma$,
\begin{equation}\label{general_covariance}
\hat\Pi_{gC}=\hat U(g)\hat\Pi_{C}\hat U^\dagger(g).
\end{equation}
\end{definition}

This formulation captures the basic idea that the classical symmetry transformations should emerge from the corresponding quantum ones. Operationally, transforming the classical region $C$ into $gC$ should be equivalent to applying the inverse transformation to the quantum state. Indeed, Eq.~\eqref{general_covariance} implies that, for every state $\hat\rho\in S(\mathcal H)$,
\begin{equation}
\operatorname{tr}\!\left(\hat\rho\hat\Pi_{gC}\right)=\operatorname{tr}\!\left(\hat U(g^{-1})\hat\rho\hat U^\dagger(g^{-1})\hat\Pi_{C}\right).
\end{equation}
In this sense, covariance expresses the compatibility between the action of $G$ on the emergent classical phase space and its unitary implementation in the underlying quantum theory.

In the nonrelativistic case, this structure can usually be implemented on a single equal-time phase space. For example, for the Galilei group without time translations, one may take
\begin{equation}
\Gamma=\mathbb R^3\times\mathbb R^3,
\end{equation}
with points labelled by position and momentum, $(\mathbf x,\mathbf p)$. Spatial translations, rotations, and Galilean boosts act on these variables while preserving the equal-time phase space itself. In particular, because Galilean spacetime has an absolute notion of simultaneity, a boost maps configurations at fixed time to configurations that can again be represented within the same phase space, with the usual transformation of position and momentum. Thus all observers may be regarded as assigning probabilities to subsets of one common phase space $\Gamma$.

In the relativistic setting, the relevant symmetry group is the proper orthochronous Poincar\'e group,
\begin{equation}
P^\uparrow_+:=\mathbb R^{1,3}\rtimes SO^+(1,3).
\end{equation}
The subgroup generated by spatial translations and spatial rotations preserves the reference hypersurface $t=0$ and acts globally on the reference phase space $\Gamma$. If $h=(\mathbf a,R)$, with $\mathbf a\in\mathbb R^3$ and $R\in SO(3)$, we write
\begin{equation}\label{Euclidean_action}
h\cdot(\mathbf x,\mathbf p):=(R\mathbf x+\mathbf a,R\mathbf p).
\end{equation}
For this Euclidean subgroup, the general covariance condition \eqref{general_covariance} takes the following form.
\begin{definition}[Euclidean covariance of the phase-space POVM]\label{def:euclidean-covariance}
Let $C\mapsto\hat\Pi_C$ be a phase-space POVM on $(\Gamma,\sigma_\Gamma)$. We say that it is covariant under spatial translations and rotations if, for every $\mathbf a\in\mathbb R^3$, every $R\in SO(3)$, and every measurable region $C\in\sigma_\Gamma$,
\begin{equation}\label{Euclidean_integrated_covariance}
\hat\Pi_{hC}=\hat U(h)\hat\Pi_C\hat U^\dagger(h),
\qquad h\cdot(\mathbf x,\mathbf p)=(R\mathbf x+\mathbf a,R\mathbf p).
\end{equation}
\end{definition}

Unlike spatial translations and rotations, Lorentz boosts change the equal-time hypersurface associated with an inertial observer. Nevertheless, a Poincar\'e transformation may map a given event of the reference hypersurface $t=0$ to another event of the same hypersurface. We will only use Poincar\'e transformations in this local sense. Given $(\mathbf x,\mathbf p)\in\Gamma$, set $x:=(0,\mathbf x)$ and $p:=(|\mathbf p|,\mathbf p)$. For $g=(a,\Lambda)\in P^\uparrow_+$, suppose that $(\Lambda x+a)^0=0$. Then $g$ maps the phase-space point $(\mathbf x,\mathbf p)$ to another point of the reference phase space.

For a nonvanishing helicity $\lambda$, this map is not obtained by simply transporting $\mathbf x$ as a spatial vector, because a boost displaces the position by a helicity-dependent amount, transverse to the direction of propagation \cite{andrzejewski2015chiral}. Accordingly, we denote the transformed point by
\begin{equation}\label{g_x_p_prime}
g\cdot(\mathbf x,\mathbf p):=(\mathbf x',\mathbf p'),\qquad (0,\mathbf x')=\Lambda(0,\mathbf x)+a+\big(0,\boldsymbol\Delta_\lambda(\Lambda,\mathbf p)\big), \qquad \mathbf p'=\Lambda\mathbf p,
\end{equation}
where, by a slight abuse of notation, $\Lambda\mathbf p$ denotes the spatial part of the transformed null momentum $\Lambda(|\mathbf p|,\mathbf p)$, and $\boldsymbol\Delta_\lambda(\Lambda,\mathbf p)$ is the shift or side-jump caused by helicity. Importantly, $\boldsymbol\Delta_\lambda$ vanishes for spatial translations and rotations, so that the Euclidean action \eqref{Euclidean_action} and Definition~\ref{def:euclidean-covariance} are unaffected. Additionally, note that $\boldsymbol\Delta_\lambda \equiv \mathbf 0 \Leftrightarrow \lambda = 0$ \cite{andrzejewski2015chiral}.

The classical covariance of the emergent phase-space description should be implemented by the unitary representation of the same Poincar\'e transformation in the quantum theory. We impose it locally, for those Poincar\'e transformations that map a given phase-space point of the reference hypersurface to another point of the same hypersurface.

\begin{definition}[Microscopic Poincar\'e covariance]\label{def:microscopic-poincare-covariance}
Assume that the POVM $C\mapsto\hat\Pi_{C}$ admits the density representation \eqref{Pi_C_density}. We say that the density $\hat\pi_{(\mathbf x,\mathbf p)}$ satisfies microscopic Poincar\'e covariance if the following condition holds. Let $(\mathbf x,\mathbf p)\in\Gamma$, let $g=(a,\Lambda)\in P^\uparrow_+$, and suppose that $g\cdot(\mathbf x,\mathbf p)=(\mathbf x',\mathbf p')$ is defined as in Eq.~\eqref{g_x_p_prime}. Then
\begin{equation}\label{microscopic_covariance_coordinates}
\hat U(g)\hat\pi_{(\mathbf x,\mathbf p)}\hat U^\dagger(g)=\hat\pi_{(\mathbf x',\mathbf p')}.
\end{equation}
\end{definition}

Definition~\ref{def:microscopic-poincare-covariance} is the density-level form of local Poincar\'e covariance. It states that the local phase-space effect assigned to a point and the one assigned to its Poincar\'e-related image are connected by the same unitary transformation that implements the Poincar\'e transformation in the quantum theory. This is the microscopic compatibility law that will be used in this work.

However, Eq.~\eqref{microscopic_covariance_coordinates} is stated directly at the level of the operator-valued density. To keep the formulation operational, we now give a sufficient condition expressed only in terms of probabilities assigned by the POVM to small but finite phase-space cells. The relevant comparison is not between the total probabilities of shrinking cells, since these probabilities vanish as the cells shrink. Rather, one must compare the corresponding probability densities with respect to the scalar measure $\mu$ used in Eq.~\eqref{Pi_C_density}. Thus the following condition is a condition on the chosen density representation $(\hat\Pi,\mu)$, not on the abstract POVM alone.

\begin{definition}[Operational local Poincar\'e covariance]\label{def:operational-local-poincare-covariance}
Let $C\mapsto\hat\Pi_{C}$ be a POVM on $(\Gamma,\sigma_\Gamma)$, and let $\mu$ be the positive measure appearing in the density representation \eqref{Pi_C_density}. We say that the pair $(\hat\Pi,\mu)$ satisfies operational local Poincar\'e covariance if the following condition holds. Let $(\mathbf x,\mathbf p)\in\Gamma$, let $g=(a,\Lambda)\in P^\uparrow_+$, and suppose that $g\cdot(\mathbf x,\mathbf p)=(\mathbf x',\mathbf p')$ is defined as in Eq.~\eqref{g_x_p_prime}. For any pair of regular shrinking neighbourhoods $B_\epsilon(\mathbf x,\mathbf p)\subset\Gamma$ and $B_\epsilon(\mathbf x',\mathbf p')\subset\Gamma$, centered at $(\mathbf x,\mathbf p)$ and $(\mathbf x',\mathbf p')$, with nonzero $\mu$-measure, and for every state $\hat\rho\in S(\mathcal H)$, one has
\begin{equation}\label{operational_local_poincare_covariance}
\lim_{\epsilon\to0}\left[
\frac{\operatorname{tr}\!\left(\hat\rho\,\hat\Pi_{B_\epsilon(\mathbf x',\mathbf p')}\right)}{\mu(B_\epsilon(\mathbf x',\mathbf p'))}
-
\frac{\operatorname{tr}\!\left(\hat U^\dagger(g)\hat\rho\,\hat U(g)\hat\Pi_{B_\epsilon(\mathbf x,\mathbf p)}\right)}{\mu(B_\epsilon(\mathbf x,\mathbf p))}
\right]=0.
\end{equation}
\end{definition}

The limiting procedure in Eq.~\eqref{operational_local_poincare_covariance} extracts a local covariance law. This should not be interpreted as assigning a direct classical meaning to arbitrary fine-grained phase-space measurements. The local probabilities appearing in Definition~\ref{def:operational-local-poincare-covariance} may be accessed indirectly, in particular by repeated preparations of the same state, by using different coarse or indirect measurements, and by inferring local data through statistical post-processing. A direct measurement of an arbitrarily small phase-space cell may fail the repeatability or non-invasiveness expected of a classical measurement; this issue will be addressed later when the coarse-grained classical sector is defined. For the present purpose, the condition only states that the fine local structure used to describe the POVM must be compatible with the empirically observed classical Poincar\'e covariance and with its unitary implementation in the quantum theory.

The following proposition shows that the operational condition is sufficient to obtain the microscopic density-level covariance law. The proof is given in Appendix~\ref{proof:operational-implies-microscopic-covariance}.

\begin{proposition}[Operational local covariance implies microscopic covariance]\label{prop:operational-implies-microscopic-covariance}
Assume that the POVM admits the density representation \eqref{Pi_C_density} with weakly continuous density. If operational local Poincar\'e covariance holds in the sense of Definition~\ref{def:operational-local-poincare-covariance}, then the density satisfies microscopic Poincar\'e covariance in the sense of Definition~\ref{def:microscopic-poincare-covariance}.
\end{proposition}

In Appendix~\ref{Local_objectivity}, we present an alternative geometric route to the same microscopic covariance law. That route is based on a different set of assumptions: frame-dependent phase spaces, covariance of the integrated frame-dependent POVMs, and a local objectivity condition at common points of intersecting hypersurfaces.

Combining Definition~\ref{def:euclidean-covariance}, Definition~\ref{def:microscopic-poincare-covariance}, and the required regularity of the density representation, we obtain the following definition.

\begin{definition}[Regular Poincar\'e-covariant phase-space POVM]\label{def:regular-poincare-covariant-povm}
Let $C\mapsto\hat\Pi_C$ be a POVM on $(\Gamma,\sigma_\Gamma)$. We say that $C\mapsto\hat\Pi_C$ is a regular Poincar\'e-covariant phase-space POVM if the following conditions hold:
\begin{enumerate}
\item it admits the density representation \eqref{Pi_C_density}, with scalar measure $d\mu(\mathbf x,\mathbf p)=\rho(\mathbf x,\mathbf p)\,d\nu(\mathbf x,\mathbf p)$ and positive operator-valued density $(\mathbf x,\mathbf p)\mapsto\hat\pi_{(\mathbf x,\mathbf p)}$;
\item the density representation is nondegenerate, in the sense that $\hat\pi_{(\mathbf x,\mathbf p)}$ does not vanish identically on any subset of positive $\nu$-measure;
\item the integrated POVM is Euclidean covariant in the sense of Definition~\ref{def:euclidean-covariance};
\item the operator-valued density satisfies microscopic Poincar\'e covariance in the sense of Definition~\ref{def:microscopic-poincare-covariance}.
\end{enumerate}
\end{definition}

Under these conditions, the scalar density $\rho(\mathbf x,\mathbf p)$ is constrained by the following proposition.

\begin{proposition}[Spatial homogeneity and isotropy of the scalar density]\label{prop:rho-mu-depends-on-p-modulus}
Let $C\mapsto\hat\Pi_C$ be a regular Poincar\'e-covariant phase-space POVM in the sense of Definition~\ref{def:regular-poincare-covariant-povm}. Then the scalar density $\rho$ is invariant under spatial translations and rotations:
\begin{equation}\label{rho_mu_translation_rotation_invariance}
\rho(R\mathbf x+\mathbf a,R\mathbf p)=\rho(\mathbf x,\mathbf p)
\end{equation}
for $\nu$-almost every $(\mathbf x,\mathbf p)\in\Gamma$, for every $\mathbf a\in\mathbb R^3$ and every $R\in SO(3)$. Consequently, with a slight abuse of notation, we may write
\begin{equation}\label{rho_mu_mod_p}
\rho(\mathbf x,\mathbf p)=\rho(|\mathbf p|)
\end{equation}
for $\nu$-almost every $(\mathbf x,\mathbf p)\in\Gamma$.
\end{proposition}

The proof of Proposition \ref{prop:rho-mu-depends-on-p-modulus} is given in Appendix~\ref{proof:rho-mu-depends-on-p-modulus}. 
It is important to stress that Poincar\'e covariance alone does not make the class of admissible POVMs empty. The following example gives an explicit family of density representations satisfying Poincar\'e covariance and the POVM normalization condition, for arbitrary quantum helicity $h$ and vanishing side-jump $\boldsymbol\Delta_\lambda \equiv \mathbf 0$, or, equivalently, $\lambda = 0$.

\begin{example}[A Poincar\'e-covariant phase-space POVM]\label{ex:covariant-massless-povm}
Let $h\in\frac{1}{2} \mathbb Z$, $\boldsymbol\Delta_\lambda \equiv \mathbf 0$ $(\Leftrightarrow \lambda =0)$, let $\mathbf q\in\mathbb R^3_*$ be a reference momentum, and let $\phi:[0,\infty)\to\mathbb C$ be a measurable function such that
\begin{equation}\label{N_q_covariant_povm}
0<N_{\mathbf q}:=(2\pi\hbar)^3\int_{\mathbb R^3_*}\frac{d^3\mathbf p}{2|\mathbf p|^4}\,|\phi(q^\mu p_\mu)|^2<\infty.
\end{equation}
With $\eta$ the relative Wigner phase of Lemma~\ref{lem:relative-wigner-phase}, define generalized seed states by
\begin{equation}\label{seed_covariant_povm_helicity}
\braket{\mathbf k|\mathbf 0,\mathbf p}=e^{-2ih\,\eta(k,p)}\,\phi(k^\mu p_\mu),
\end{equation}
and set $\ket{\mathbf x,\mathbf p}=e^{-\frac{i}{\hbar}\mathbf x\cdot\hat{\mathbf P}}\ket{\mathbf 0,\mathbf p}$, $\hat\pi_{(\mathbf x,\mathbf p)}=\ket{\mathbf x,\mathbf p}\bra{\mathbf x,\mathbf p}$, and $\rho(|\mathbf p|)=\frac{|\mathbf q|}{N_{\mathbf q}}\frac{1}{|\mathbf p|^4}$. Inserted in the density representation \eqref{Pi_C_density}, these data define a regular Poincar\'e-covariant phase-space POVM satisfying the normalization condition \eqref{identity}. This is proven in Appendix~\ref{proof:covariant-povm-example}.
\end{example}

The following proposition shows that it is highly non-trivial to construct examples with nonvanishing side-jump.

\begin{proposition}[Covariance nearly forces a vanishing shift]\label{prop:rigidity-shift}
Let $\lambda\in\mathbb R$, $h \in \frac{1}{2} \mathbb Z$ and let $C\mapsto\hat\Pi_C$ be a regular Poincar\'e-covariant phase-space POVM in the sense of Definition~\ref{def:regular-poincare-covariant-povm}, with respect to the transformation law \eqref{g_x_p_prime} with shift $\boldsymbol\Delta_\lambda$. Assume in addition that for $\nu$-almost every $(\mathbf x,\mathbf p)$ the momentum kernel $\braket{\mathbf k|\hat\pi_{(\mathbf x,\mathbf p)}|\mathbf k'}$ is given by a measurable function. Then $\boldsymbol\Delta_\lambda(\Lambda,\mathbf p)=\mathbf 0$ for every $\Lambda\in SO^+(1,3)$ and $\nu$-almost every $\mathbf p\in\mathbb R^3_*$.
\end{proposition}

This is proven in Appendix~\ref{proof:rigidity-shift} and indicates that Poincar\'e covariance, complemented by a regularity condition on $\braket{\mathbf k|\hat\pi_{(\mathbf x,\mathbf p)}|\mathbf k'}$, forces any emergent classical massless particle to have vanishing helicity. This is similar to work by Wightman, who showed that no Euclidean-covariant projection-valued localization exists for massless quanta of helicity $|h|\geq1$ \cite{wightman1962localizability}. However, the obstruction in that result relates to sharpness, and it is lifted by passing to POVMs, where Euclidean-covariant localization observables do exist for every helicity \cite{Castrigiano1981}. Proposition~\ref{prop:rigidity-shift} is of a different nature and cannot be evaded by unsharpness. It states that, among the admissible classical transformation laws \eqref{g_x_p_prime}, unitary implementability selects the one with $\boldsymbol\Delta_\lambda=\mathbf 0$. However, as the proof of the no-go will also hold for $\boldsymbol\Delta_\lambda \neq \mathbf 0$, and we want to avoid any superfluous assumptions on $\hat \pi_{(\mathbf x,\mathbf p)}$, we do not restrict to this case and keep Eq.~\eqref{g_x_p_prime} in its general form, with an unspecified, momentum-dependent side-jump.

\section{Coarse partitions}\label{Coarse_partitions}

In Sect.~\ref{Phasespace_measurements}, we introduced phase-space measurements on the candidate phase space $\Gamma$. For a finite partition $P=\{C_i\}_{i=1}^N$, the underlying quantum measurement is the POVM $\{\hat\Pi_{C_i}\}_{i=1}^N$, while the classical measurement $\{E_{C_i}\}_{i=1}^N$ is recovered only when the POVM outcomes admit an effective classical interpretation. In Sect.~\ref{Minimal_requirements_for_Poincaré_group}, we formulated the compatibility between the emergent classical description and the underlying quantum theory in terms of Poincar\'e covariance. Because generic Lorentz transformations do not preserve the reference hypersurface $t=0$, this covariance condition had to be imposed locally, at the level of phase-space points and probability densities. At first sight, this seems to be in tension with the idea that classical behaviour should emerge only after coarse graining. The resolution is that the fine-grained level is not assumed to be fully classical. It is used as an idealized probabilistic structure on which local Poincar\'e covariance can be formulated and tested operationally, for instance through probabilities inferred from repeated preparations, coarse measurements, indirect measurements, and statistical post-processing. Direct fine-grained measurements, by contrast, need not themselves be classical and may fail properties such as repeatability and non-invasiveness. Thus, we distinguish between two levels of emergence. At the fine-grained level, we assume only the partial classical structure needed to formulate local Poincar\'e covariance. At the coarse-grained level, we require a stronger form of classical behaviour: the outcomes of suitable finite-resolution measurements should be interpretable as classical alternatives and should satisfy the operational classicality requirements which will be introduced in the next section.

In this section, we formalize this separation between the fine-grained and the coarse-grained sector, and clarify how probabilities for events outside the directly classical sector may nevertheless be accessed inferentially. We denote by $\mathcal P_\Gamma$ the class of measurable partitions of $\Gamma$ considered in the phase-space theory, and we introduce a nonempty family $\mathcal P_{\rm CG}\subseteq\mathcal P_\Gamma$. The elements of $\mathcal P_{\rm CG}$ are the partitions whose associated POVM outcomes admit a direct classical interpretation. Thus, if $P=\{C_i\}_{i=1}^N\in\mathcal P_{\rm CG}$, the outcomes of the quantum measurement $\{\hat\Pi_{C_i}\}_{i=1}^N$ are interpreted as classical alternatives.

\begin{definition}[Coarse partitions and coarse cells]\label{def:coarse-partitions}
The elements of $\mathcal P_{\rm CG}$ are called coarse partitions. A measurable set $C\subseteq\Gamma$ is called a coarse cell if there exists a coarse partition $P\in\mathcal P_{\rm CG}$ such that $C\in P$. In this case we write $C\triangleleft\mathcal P_{\rm CG}$.
\end{definition}

We leave the precise notion of coarseness unspecified. In a concrete model, it may for instance be determined by a minimal resolvable phase-space volume or by limitations of an experimental apparatus. Our aim here is not to give a sufficient characterization of when a partition is coarse enough to be classical. The next section will formulate the minimal necessary consistency conditions that any candidate pair $(\mathcal P_{\rm CG},\hat\Pi)$ must satisfy in order to support a classical massless-particle interpretation: non-invasive repeatability, compatibility between the classical and quantum momentum measurements, and a minimal information-gain condition. The role of the present
section is different: we specify structural properties of $\mathcal P_{\rm CG}$ that encode the distinction between direct classical measurability and indirect statistical inferability.

The first structural property is inferential completeness. The partitions in $\mathcal P_{\rm CG}$ are precisely those whose outcomes admit a direct classical interpretation. Direct measurements associated with fine-grained events may fail classical repeatability or non-invasiveness, and are therefore not assumed to belong to the classical sector $\mathcal P_{\rm CG}$; however, their probabilities can be reconstructed indirectly from coarse-grained statistics, through repeated preparations, classical statistical inference, and post-processing. In this sense, the full POVM $C\mapsto\hat\Pi_{C}$ is not interpreted as a collection of directly classical fine-grained measurements, but as an idealized infinite-resource probabilistic extension of the finite-resolution classical theory.

Indirect access to events outside the directly classical sector $\mathcal P_{\rm CG}$ is obtained by combining statistics from admissible coarse measurements through the usual operations of classical probability. In a classical, Kolmogorovian probability theory, statistics obtained from different measurements can be recombined to infer probabilities for other events, provided that all events belong to the same classical sample space \cite{Kolmogorov1956,Billingsley1995}. For example, if $C_1$ and $C_2$ are events, possibly appearing in different partitions, then finite additivity gives
\begin{align}
C_1\cap C_2=\emptyset &\quad\Longrightarrow\quad p(C_1\sqcup C_2)=p(C_1)+p(C_2), \label{p_operation_1}\\
C_1\subseteq C_2 &\quad\Longrightarrow\quad p(C_2\setminus C_1)=p(C_2)-p(C_1). \label{p_operation_2}
\end{align}
Thus, by estimating the probabilities of $C_1$ and $C_2$ from repeated preparations and admissible coarse measurements, the observer can infer the probabilities of the derived events $C_1\sqcup C_2$ and $C_2\setminus C_1$. In the present framework, these identities hold automatically for $p=p_{\hat\rho}$ as a consequence of the additivity of the POVM, Eq.~\eqref{additivity}. What is nontrivial is the operational interpretation: the derived events need not themselves be directly measurable as classical coarse-grained outcomes, even though their probabilities may be inferable from admissible coarse-grained data.

This separation between direct measurability and indirect inferability is particularly useful for the operational local covariance condition introduced in Definition~\ref{def:operational-local-poincare-covariance}. That condition is formulated in terms of local probability densities, obtained as limits of probabilities assigned to shrinking neighbourhoods $B_\epsilon(\mathbf x,\mathbf p)\subset\Gamma$. This locality is unavoidable: even classically, testing Poincar\'e transformations while remaining within the fixed reference hypersurface $t=0$ requires considering transformations that act locally on phase-space points. Such local information need not come from direct fine-grained classical measurements; it may instead be inferred, in principle, from repeated preparations, suitably chosen admissible coarse measurements, and statistical post-processing. Inferential completeness is the structural condition that makes this possible: the coarse sector contains the directly classical measurements, while its inferential closure gives access, in principle, to the full measurable phase-space structure.

We now formalize this ideal inferential access. Given a finite partition $P=\{C_i\}_{i=1}^N$, we define the family of events resolved by $P$ as
\begin{equation}
\operatorname{Ev}(P):=\left\{\bigsqcup_{i\in J}C_i:J\subseteq\{1,\ldots,N\}\right\}.
\end{equation}
Thus, $\operatorname{Ev}(P)$ contains all events whose probabilities can be obtained by classical post-processing of the outcomes of $P$. For a family of partitions $\mathcal A\subseteq\mathcal P_\Gamma$, define
\begin{equation}
\mathcal G(\mathcal A):=\bigcup_{P\in\mathcal A}\operatorname{Ev}(P).
\end{equation}
We denote by $\mathcal D(\mathcal A)$ the Dynkin closure of $\mathcal G(\mathcal A)$, namely the smallest family $\mathcal D\subseteq\sigma_\Gamma$ such that $\mathcal G(\mathcal A)\subseteq\mathcal D$, $\Gamma\in\mathcal D$, and
\begin{equation}
C\in\mathcal D\quad\Longrightarrow\quad \Gamma\setminus C\in\mathcal D,
\end{equation}
and, for every countable family $\{C_n\}_{n\in\mathbb N}\subseteq\mathcal D$ of pairwise disjoint sets,
\begin{equation}
\bigsqcup_{n\in\mathbb N}C_n\in\mathcal D.
\end{equation}
Equivalently, $\mathcal D(\mathcal A)$ is the smallest Dynkin system containing all events resolved by the partitions in $\mathcal A$ \cite{Billingsley1995,Bogachev2007,Kallenberg2021}. 

Since the POVM is insensitive to $\nu$-null modifications of phase-space cells, events that differ only by a set of vanishing $\nu$-measure are operationally indistinguishable. It is therefore natural to impose inferential completeness only up to $\nu$-null sets. We define the $\nu$-completion of the Dynkin closure generated by $\mathcal P_{\rm CG}$ as
\begin{equation}
\overline{\mathcal D}^{\,\nu}(\mathcal P_{\rm CG}) := \left\{ E\in\sigma_\Gamma:\exists B\in\mathcal D(\mathcal P_{\rm CG})\ \text{such that}\ \nu(B\triangle E)=0 \right\}.
\end{equation}
The inferential completeness condition is then
\begin{equation}\label{dynkin_inferential_completeness}
\overline{\mathcal D}^{\,\nu}(\mathcal P_{\rm CG})=\sigma_\Gamma.
\end{equation}
It expresses the assumption that the full measurable phase-space structure can, at least in principle and with idealized infinite resources, be reconstructed at the level of inferred probabilities from the statistics of directly classical coarse measurements, up to operationally irrelevant null sets.

The second structural property is closure under coarsening. If a partition is directly measurable as a coarse classical measurement, then any partition with lower resolution should also be directly admissible as a coarse classical measurement. Indeed, a coarsening distinguishes fewer alternatives and therefore cannot require more resolving power than the original partition. For two partitions $P,P'\in\mathcal P_\Gamma$, we say that $P'$ is a coarsening of $P$, and write $P'\succeq P$, if for every cell $C\in P$ there exists a cell $C'\in P'$ such that $C\subseteq C'$. The closure-under-coarsening condition is
\begin{equation}\label{closure_under_coarsening_definition}
P\in\mathcal P_{\rm CG},\quad P'\succeq P\quad\Longrightarrow\quad P'\in\mathcal P_{\rm CG}.
\end{equation}

\begin{definition}[Family of coarse partitions]\label{def:family-coarse-partitions}
A family of coarse partitions $\mathcal P_{\rm CG}$ is a nonempty subset of $\mathcal P_\Gamma$ satisfying inferential completeness, Eq.~\eqref{dynkin_inferential_completeness}, and closure under coarsening, Eq.~\eqref{closure_under_coarsening_definition}.
\end{definition}

The definition separates direct classical accessibility from ideal inferability. The direct classical sector consists only of the partitions in $\mathcal P_{\rm CG}$, and the classicality requirements imposed below, such as non-invasive repeatability, will apply only to these partitions. By contrast, the equality $\overline{\mathcal D}^\nu(\mathcal P_{\rm CG})=\sigma_\Gamma$ concerns what can be inferred in principle from the statistics of sufficiently many coarse-grained measurements. This is why the use of the full POVM $C\mapsto\hat\Pi_{C}$ in the previous section should not be interpreted as assigning direct classical meaning to arbitrary fine-grained measurements. Rather, the full POVM provides the idealized probabilistic structure from which local densities and Poincar\'e-covariant transformations can be discussed, while direct classical behaviour is required only for the coarse partitions in $\mathcal P_{\rm CG}$.

An equivalent contextual characterization of families of coarse partitions $\mathcal P_{\rm CG}$ is given in Appendix~\ref{app:contexts}. There, $\mathcal P_{\rm CG}$ is described as a union of finite-inferentially closed contexts. This makes explicit that finite statistical inferences are required to remain directly classical only within a context, while the full inferential completeness condition $\overline{\mathcal D}^\nu(\mathcal P_{\rm CG})=\sigma_\Gamma$ applies to the union of contexts.

\section{Coarse-graining schemes}\label{requirements}

In Sects.~\ref{Phasespace_measurements} and \ref{Coarse_partitions}, we introduced the basic structures needed to formulate a particle-like classical phase-space description: a candidate classical phase space $\Gamma$, a phase-space POVM $C\mapsto\hat{\Pi}_{C}$ assigning quantum effects to measurable phase-space regions, and a family of coarse partitions $\mathcal P_{\rm CG}\subseteq\mathcal P_\Gamma$ whose outcomes are intended to admit a classical interpretation. We collect these data in the following definition.

\begin{definition}[Coarse-graining scheme]\label{def:coarse-graining-scheme}
A coarse-graining scheme is a triple $(\Gamma,\hat{\Pi},\mathcal P_{\rm CG})$, where $\Gamma$ is the candidate classical phase space, $\hat{\Pi}:C\mapsto\hat{\Pi}_C$ is a phase-space POVM admitting the density representation \eqref{Pi_C_density}, and $\mathcal P_{\rm CG}$ is a family of coarse partitions in the sense of Definition~\ref{def:family-coarse-partitions}.
\end{definition}

This definition does not by itself guarantee classical behaviour. We now formulate minimal requirements that $(\Gamma,\hat{\Pi},\mathcal P_{\rm CG})$ must satisfy in order to support a classical massless-particle interpretation. These requirements are necessary consistency conditions only: they are not meant to determine uniquely either the family $\mathcal P_{\rm CG}$ or the POVM $C\mapsto\hat{\Pi}_{C}$, nor do they provide a sufficient criterion for classicality.

The first two requirements concern the operational behaviour expected from coarse-grained particle observables. For each coarse partition $P=\{C_i\}_{i=1}^N\in\mathcal P_{\rm CG}$, the effects $\{\hat{\Pi}_{C_i}\}_{i=1}^N$ should behave, at the coarse level, as alternatives stable under immediate repetition. In addition, whenever a coarse cell $C\triangleleft\mathcal P_{\rm CG}$ has the form $C=\mathbb R^3\times Q$, the corresponding effect should agree, within the allowed tolerance, with the standard momentum spectral measure $\hat E_{\hat{\mathbf P}}(Q):=\int_{Q}\frac{d^3\mathbf k}{2|\mathbf k|}\ket{\mathbf k}\bra{\mathbf k}$. The final requirement imposes a minimal informational content on the coarse-grained sector: at the very least, a classical particle description should allow one to ask whether the particle propagates within a sufficiently large angular region or within its complement. Accordingly, $\mathcal P_{\rm CG}$ must contain very coarse yes/no alternatives in the momentum-direction sector.

\subsection{Non-invasive repeatability}\label{sec:weak-noninvasive-repeatability}

For a coarse-graining scheme $(\Gamma,\hat\Pi,\mathcal P_{\rm CG})$ to admit a classical interpretation, each coarse partition $P=\{C_i\}_{i=1}^N\in\mathcal P_{\rm CG}$ should define outcomes that, under the corresponding quantum measurement $\{\hat\Pi_{C_i}\}_{i=1}^N$, behave as mutually exclusive classical alternatives. In particular, an immediate repetition of the same measurement, with no intervening time evolution, should give the same outcome. This is the coarse-grained analogue of the non-invasive measurability assumption used in the Leggett--Garg framework: a classical measurement may reveal which alternative is realized and should not by itself induce a change in the value of the same immediately repeated measurement \cite{PhysRevLett.54.857,Emary_2014}. In the present setting we impose this idea only for the finite-resolution measurements associated with partitions in $\mathcal P_{\rm CG}$.

Classically, if the first measurement gives outcome $i$, the conditional probability $p(j|i)$ of obtaining a different outcome $j\neq i$ in an immediate repetition vanishes. Therefore, the total probability of observing a transition between distinct outcomes, weighted by the probability $p(i)$ of the first outcome, is
\begin{equation}
\sum_{i=1}^N \sum_{j\neq i}p(i)p(j|i)=0.
\end{equation}

In quantum theory, a POVM determines the outcome probabilities but does not uniquely determine the conditional probabilities. This is instead implemented by instruments \cite{DaviesLewis1970,BuschLahti1996}.

\begin{definition}[Instruments]\label{def:compatible-instruments}
Let $P=\{C_i\}_{i=1}^N$ be a finite partition and let
$\{\hat\Pi_{C_i}\}_{i=1}^N$ be the associated POVM. An instrument implementing this POVM is a collection
$\mathcal I^P=\{\mathcal I_i^P\}_{i=1}^N$ of completely positive, trace-nonincreasing maps such that $\sum_i\mathcal I_i^P$ is trace preserving and
\begin{equation}
\operatorname{tr}\!\left[\mathcal I_i^P(\hat \rho)\right]=\operatorname{tr}\!\left[\hat \Pi_{C_i}\hat \rho\right] \qquad \forall i, \hat \rho.
\end{equation}
Consequently, $\operatorname{tr}[\mathcal I_i^P(\hat\rho)]$ is the probability of outcome $i$, while $\mathcal I_i^P(\hat\rho)$ is the corresponding unnormalized conditional state. We denote by $\mathfrak I(P)$ the set of all instruments implementing $\{\hat\Pi_{C_i}\}_{i=1}^N$.
\end{definition}

Let the system initially be prepared in a normalized state $\hat \rho \in S(\mathcal H)$. If the output of the first measurement is subjected to a second measurement of the same POVM, the probability of obtaining $i$ and then $j$ is
\begin{equation}
p_{\hat \rho}(i) p_{\hat \rho,\mathcal I^P}(j|i)
=
\operatorname{tr}\!\left[
\hat\Pi_{C_j}\,
\mathcal I_i^P\bigl(\hat \rho\bigr)
\right].
\end{equation}
The probability of observing a transition between distinct coarse outcomes is consequently
\begin{align}\label{determinism}
\sum_{i=1}^N\sum_{j\neq i}
\operatorname{tr}\!\left[
\hat\Pi_{C_j}\,
\mathcal I_i^P\bigl(\hat \rho\bigr)
\right].
\end{align}

A strong non-invasive repeatability condition would demand Eq.~\eqref{determinism} to be small for a specified measurement implementation and for all, or at least for most, states. For the no-go theorem, however, we impose a considerably weaker condition: for every coarse partition, there need only exist at least one compatible implementation and at least one state in which the repeated measurement is approximately repeatable.

\begin{definition}[Weak non-invasive repeatability]\label{def:weak-noninvasive-repeatability}
Let $\epsilon>0$. A coarse-graining scheme $(\Gamma,\hat\Pi,\mathcal P_{\rm CG})$ satisfies weak non-invasive repeatability with tolerance $\epsilon$ if, for every coarse partition $P=\{C_i\}_{i=1}^N\in\mathcal P_{\rm CG}$,
\begin{equation}\label{PVMness}
\inf_{\mathcal I^P\in\mathfrak I(P)}
\inf_{\hat \rho \in S(\mathcal H)}
\sum_{i=1}^N\sum_{j\neq i}
\operatorname{tr}\!\left[
\hat\Pi_{C_j}\,
\mathcal I_i^P\bigl(\hat \rho\bigr)
\right]
<\epsilon.
\end{equation}
\end{definition}

Thus, for each coarse partition $P=\{C_i\}_{i=1}^N\in\mathcal P_{\rm CG}$, there must exist a measurement implementation and a state in which the probability of observing the classically forbidden transition between disjoint coarse outcomes in an immediate repetition is smaller than the tolerance $\epsilon$. Such deviations from classicality are admissible only insofar as they remain so unlikely that they are not operationally notable; we therefore expect to work in a regime where $\epsilon\ll1$. 

For the purposes of the no-go theorem, however, it is not enough to treat $\epsilon$ as merely small in an asymptotic or qualitative sense. We need an explicit non-asymptotic threshold, sufficiently weak to be physically reasonable but sharp enough for the argument. We therefore assume
\begin{equation}
\epsilon<\frac{1}{6}.
\end{equation}
This means that, even in the best-case implementation and state selected by the infima in Definition~\ref{def:weak-noninvasive-repeatability}, the probability of observing a non-classical transition must be less than one sixth. The numerical value $1/6$ is not meant to define a universal boundary of classicality; it is only a quantitative threshold sufficient for the no-go argument.

\subsection{Compatibility between classical and quantum momentum}\label{sec:momentum-compatibility}

The classical phase-space coordinate $\mathbf p$ is meant to represent the particle momentum. In the quantum theory, the corresponding observable is the momentum operator $\hat{\mathbf P}$, equivalently the generator of spatial translations. The effect associated with the question whether the momentum lies in a measurable region $Q\subseteq\mathbb R^3_*$ is the corresponding element of the momentum spectral measure,
\begin{equation}
\hat E_{\hat{\mathbf P}}(Q):=\int_Q\frac{d^3\mathbf k}{2|\mathbf k|}\ket{\mathbf k}\bra{\mathbf k},
\end{equation}
where $\ket{\mathbf k}$ are the eigenstates of $\hat{\mathbf P}$ normalized as $\braket{\mathbf k|\mathbf k'}=2|\mathbf k|\delta^3(\mathbf k-\mathbf k')$. To recover the classical momentum coordinate $\mathbf p$ as the classical limit of the quantum momentum observable $\hat{\mathbf P}$, the momentum marginal of the phase-space POVM $C\mapsto\hat{\Pi}_{C}$ must be compatible, at the coarse level, with the standard quantum momentum measurement. Thus, whenever $\mathbb R^3\times Q$ is a coarse cell, one should have
\begin{equation}\label{Pi_P_approx}
\hat{\Pi}_{\mathbb R^3\times Q}\approx\hat E_{\hat{\mathbf P}}(Q).
\end{equation}

We formulate this compatibility operationally, in analogy with Definition~\ref{def:weak-noninvasive-repeatability}. Let $\{Q_i\}_{i=1}^N$ be a measurable partition of $\mathbb R^3_*$ such that the corresponding momentum-marginal partition $P_Q:=\{\mathbb R^3\times Q_i\}_{i=1}^N$ belongs to $\mathcal P_{\rm CG}$. Consider two immediate measurements: first the momentum PVM $\{\hat E_{\hat{\mathbf P}}(Q_i)\}_{i=1}^N$, and then the momentum marginal of the phase-space POVM, $\{\hat{\Pi}_{\mathbb R^3\times Q_i}\}_{i=1}^N$. Compatibility requires that the probability of obtaining different outcomes in these two measurements be small. We take the first measurement to be the ideal Lüders implementation of the momentum PVM, as this is the canonical, repeatable instrument for which we expect to recover the outcome of the first measurement \cite{luders1950zustandsanderung}. As in Definition~\ref{def:weak-noninvasive-repeatability}, we impose only a weak version of this condition, requiring it to hold in at least one normalized state.

\begin{definition}[Weak classical--quantum momentum compatibility]\label{def:weak-momentum-compatibility}
Let $\epsilon'>0$. A coarse-graining scheme $(\Gamma,\hat\Pi,\mathcal P_{\rm CG})$ satisfies weak classical--quantum momentum compatibility with tolerance $\epsilon'$ if, for every measurable partition $\{Q_i\}_{i=1}^N$ of $\mathbb R^3_*$ such that $\{\mathbb R^3\times Q_i\}_{i=1}^N\in\mathcal P_{\rm CG}$,
\begin{equation}\label{inf_compatibility}
\inf_{\hat \rho \in S(\mathcal H)}\sum_{i=1}^N\sum_{j\neq i} \operatorname{tr}\!\left[ \hat{\Pi}_{\mathbb R^3\times Q_j}\hat E_{\hat{\mathbf P}}(Q_i) \hat \rho \hat E_{\hat{\mathbf P}}(Q_i)\right]<\epsilon'.
\end{equation}
\end{definition}

The tolerance $\epsilon'$ here has the same status as the tolerance $\epsilon$ in Definition~\ref{def:weak-noninvasive-repeatability}: it quantifies the allowed deviation from ideal compatibility in the best-case scenario selected by the infimum. We expect $\epsilon' \ll 1$, and for the purposes of the no-go theorem we impose the explicit threshold $\epsilon'<1/6$.

We do not impose any compatibility between the classical phase-space coordinate $\mathbf{x}$ and position eigenvalues since sharp localization in relativistic quantum mechanics faces well-known tensions with causality and covariance. For example, the Newton-Wigner operator $\hat{\mathbf{X}}$, whose eigenvectors are defined as the Fourier transform of the momentum states $\ket{\mathbf{k}}$ \cite{RevModPhys.21.400, fulling_1989}, analogously to nonrelativistic quantum mechanics, suffers from causality and covariance issues \cite{PhysRev.139.B963, PhysRevD.10.3320, PhysRevD.22.377, Malament1996-MALIDO, 87c29b34-f4a4-31da-ad1b-ae81d42452e1}. In any case, compatibility with a position operator is not needed for the no-go theorem.

\subsection{Minimal directional information gain}\label{section:information}

A completely trivial coarse-grained sector $\mathcal P_{\rm CG}$, containing only the partition $\{\Gamma\}$, would always give a formally classical outcome, but it would tell us nothing about the state of the particle. We are instead interested in emergent classical descriptions with nontrivial empirical content, in which an observer can extract a non-negligible amount of information about the system. At a minimum, the coarse partitions of the scheme $(\Gamma,\hat\Pi,\mathcal P_{\rm CG})$ should allow one to obtain coarse information about the direction of propagation. In particular, it should be possible to ask whether the particle propagates within a sufficiently coarse angular region or within its complement. Such yes/no alternatives are represented by partitions of the form $\{\mathbb R^3\times Q,\Gamma\setminus(\mathbb R^3\times Q)\}$, where $Q\subseteq\mathbb R^3_*$ is determined by a region of directions on $S^2$.

\begin{definition}[Angular regions and associated momentum cones]\label{def:angular-region-momentum-cone}
For every measurable angular region $\Omega\subseteq S^2$, we denote by
\begin{equation}
\sigma(\Omega):=\frac{1}{4\pi}\int_\Omega d\Omega(\mathbf n)\label{sigma_Omega}
\end{equation}
its normalized surface area, where $d\Omega(\mathbf n)$ is the standard area element on the unit sphere. We also define the associated momentum cone
\begin{equation}
Q_\Omega:=\left\{\mathbf p\in\mathbb R^3_*:\frac{\mathbf p}{|\mathbf p|}\in\Omega\right\}.
\end{equation}
\end{definition}

\begin{definition}[Minimal directional information gain]\label{def:minimal-directional-information-gain}
Let $0\leq\delta<1$. A coarse-graining scheme $(\Gamma,\hat\Pi,\mathcal P_{\rm CG})$ satisfies minimal directional information gain with tolerance $\delta$ if there exists a measurable angular region $\Omega\subseteq S^2$ whose normalized area is close to that of a hemisphere,
\begin{equation}\label{delta_sigma}
\frac{1}{2}(1-\delta)\leq\sigma(\Omega)\leq\frac{1}{2},
\end{equation}
such that the corresponding dichotomic phase-space partition belongs to the coarse-grained sector:
\begin{equation}\label{Omega}
\left\{\mathbb R^3\times Q_\Omega,\Gamma\setminus(\mathbb R^3\times Q_\Omega)\right\}\in\mathcal P_{\rm CG}.
\end{equation}
\end{definition}

Definition~\ref{def:minimal-directional-information-gain} imposes a minimal threshold on the information about the direction of the momentum that is directly available in the coarse-grained sector. It does not require fine angular resolution. It only requires the existence of one sufficiently balanced yes/no question about the propagation direction. Moreover, by closure under coarsening, Eq.~\eqref{closure_under_coarsening_definition}, the definition is automatically satisfied if $\mathcal P_{\rm CG}$ contains any refinement of the dichotomic partition $\{\mathbb R^3\times Q_\Omega,\Gamma\setminus(\mathbb R^3\times Q_\Omega)\}$. Such a refinement may also resolve the position $\mathbf x$ or the energy $|\mathbf p|$, and would therefore give the observer more information than the original yes/no angular alternative.

Definition~\ref{def:minimal-directional-information-gain} can also be understood as a minimal information-gain condition in the sense of Shannon information \cite{shannon1948mathematical,cover1999elements}. A more detailed operational formulation of this point is given in Appendix~\ref{appendix_minimal_information}. Since we are concerned here only with the orientation of the momentum, $\mathbf p/|\mathbf p|$, we restrict attention to the classical state space $S^2$. In the absence of prior knowledge about the direction of propagation, the natural classical prior is the uniform probability measure $\sigma$ on $S^2$, defined by Eq.~\eqref{sigma_Omega}. The average information gain expected from a partition $\{\Omega_i\}_{i=1}^N$ of the sphere $S^2$ is given by the Shannon entropy of the induced uniform distribution,
\begin{equation}
H\!\left(\{\Omega_i\}_{i=1}^N\right)=-\sum_{i=1}^N \sigma(\Omega_i)\log\sigma(\Omega_i).
\end{equation}

For a dichotomic measurement with $N=2$, this quantity is maximized when the two alternatives have equal normalized area, $\sigma(\Omega)=\sigma(\Omega^\mathrm{c})=1/2$. Operationally, this expresses the fact that a binary question is most informative when neither answer is almost certain in advance. If $\sigma(\Omega)$ is close to $1$, the outcome $\Omega$ is very likely but gives little information; the complementary outcome is more informative but occurs with very small probability. The average information gain is therefore maximized by a balanced yes/no alternative.

The parameter $\delta>0$ in Eq.~\eqref{delta_sigma} quantifies how far the available dichotomic angular alternative is allowed to deviate from this Shannon-optimal balance. For the purposes of the no-go theorem, we impose the explicit threshold
\begin{equation}\label{delta}
\delta<\frac{1}{6}.
\end{equation}
This numerical value is not intended to define a universal boundary for meaningful information gain; it is a non-asymptotic tolerance, weak enough to allow deviations from the ideal Shannon-optimal binary partition, but strong enough for the no-go argument.

We allow a nonzero deviation $\delta>0$ rather than imposing the existence of an exactly Shannon-optimal dichotomic partition. Requiring an exact half-area partition would impose an unnecessarily rigid constraint on the coarse-grained sector. Allowing a tolerance $\delta$ keeps the condition minimal: the coarse-graining scheme is only required to contain a yes/no angular alternative with non-negligible and approximately balanced information about the momentum direction. The no-go theorem will therefore rule out the existence of any coarse-graining scheme with this minimal directional information gain.

\section{No-go theorem}\label{nogo_theorem}

In this section, we present the no-go theorem for a particle-like classical phase-space limit of massless quantum degrees of freedom. The proof follows the same mechanism already suggested by the coherent-state obstruction discussed in the Introduction, but now at the level of a general phase-space POVM. The key object is the local operator-valued density $\hat\pi_{(\mathbf0,\mathbf p)}$, or more precisely its momentum-space diagonal response. Poincar\'e covariance forces this response to depend on the quantum momentum $\mathbf k$ and the classical momentum label $\mathbf p$ only through the invariant $k^\mu p_\mu =|\mathbf k|\,|\mathbf p|-\mathbf k\cdot\mathbf p$. For massless momenta, this invariant is insensitive to changes of $|\mathbf k|$ along the direction $\mathbf k\parallel\mathbf p$. Thus the same covariance condition that makes the POVM transform as a massless particle phase-space measurement also limits its ability to localize momentum in the full classical momentum space.

The propositions below make this statement quantitative. First, covariance fixes the invariant form of the detection kernel associated with $\hat\pi_{(\mathbf0,\mathbf p)}$. Second, integrating over the position variable gives a diagonal expression for the momentum-marginal effects $\hat\Pi_{\mathbb R^3\times Q}$. Third, when $Q$ is a momentum cone determined by an angular region on $S^2$, one obtains bounds depending only on the angular area. The last two propositions combine Poincar\'e covariance with the two classicality conditions, namely weak non-invasive repeatability and weak classical--quantum momentum compatibility. The former forces the coarse momentum-marginal effect $\hat\Pi_{\mathbb R^3\times Q}$ to behave quasi-projectively in a weak essential sense, while the latter forces it to behave quasi-locally with respect to the classical momentum region $Q$, i.e. close to the corresponding sharp momentum projector. The theorem then applies these two consequences to the angular cones selected by minimal directional information gain and combines them with the area bounds of the third proposition. This leads to the no-go result: for sufficiently small tolerances, no regular Poincar\'e-covariant phase-space POVM can support a coarse-grained classical particle description with even minimal directional information. All proofs of the propositions and of the theorem are collected in Appendix~\ref{proof_nogo_theorem}.

\begin{proposition}[Invariant form of the detection kernel]\label{prop:detection-kernel-invariant-form}
Let $C\mapsto\hat\Pi_C$ be a regular Poincar\'e-covariant phase-space POVM in the sense of Definition~\ref{def:regular-poincare-covariant-povm}. Then the detection kernel $\mathcal K(\mathbf k,\mathbf p):=\braket{\mathbf k|\hat\pi_{(\mathbf0,\mathbf p)}|\mathbf k}$ depends on $\mathbf k$ and $\mathbf p$ only through the Lorentz-invariant quantity $k^\mu p_\mu:=|\mathbf k|\,|\mathbf p|-\mathbf k\cdot\mathbf p$. That is, there exists a nonnegative function $\Phi$ such that
\begin{equation}\label{Phi}
\mathcal K(\mathbf k,\mathbf p)=\braket{\mathbf k|\hat\pi_{(\mathbf0,\mathbf p)}|\mathbf k}=\Phi(k^\mu p_\mu)
\end{equation}
for almost every $\mathbf k,\mathbf p\in\mathbb R^3_*$.
\end{proposition}

\begin{proposition}[Diagonal form of momentum-marginal effects]\label{prop:diagonal-momentum-marginal}
Let $C\mapsto\hat\Pi_C$ be a regular Poincar\'e-covariant phase-space POVM in the sense of Definition~\ref{def:regular-poincare-covariant-povm}. Then, for every measurable $Q\subseteq\mathbb R^3_*$, the momentum-marginal effect $\hat\Pi_{\mathbb R^3\times Q}$ is diagonal in the momentum representation:
\begin{equation}\label{k_Pi_Q_k_prime}
\braket{\mathbf k|\hat\Pi_{\mathbb R^3\times Q}|\mathbf k'}=2|\mathbf k|\,\delta^3(\mathbf k-\mathbf k')\,\Pi_{\mathbb R^3\times Q}(\mathbf k),
\end{equation}
where
\begin{equation}\label{Pi_Q_kernel}
\Pi_{\mathbb R^3\times Q}(\mathbf k)=\frac{(2\pi\hbar)^3}{2|\mathbf k|}\int_Qd^3\mathbf p\,\rho(|\mathbf p|)\Phi(k^\mu p_\mu).
\end{equation}
Equivalently,
\begin{equation}\label{Pi_Q_spectral_form}
\hat\Pi_{\mathbb R^3\times Q}=\int_{\mathbb R^3_*}\frac{d^3\mathbf k}{2|\mathbf k|}\,\Pi_{\mathbb R^3\times Q}(\mathbf k)\ket{\mathbf k}\bra{\mathbf k}.
\end{equation}
\end{proposition}

\begin{proposition}[Angular response and area bounds for momentum cones]\label{prop:angular-response-area-bounds}
Let $C\mapsto\hat\Pi_C$ be a regular Poincar\'e-covariant phase-space POVM in the sense of Definition~\ref{def:regular-poincare-covariant-povm}. Let $\Omega\subseteq S^2$ be a measurable angular region, and let $Q_\Omega$ be the associated momentum cone of Definition~\ref{def:angular-region-momentum-cone}. Then the diagonal multiplier of the momentum-marginal effect $\hat\Pi_{\mathbb R^3\times Q_\Omega}$ is
\begin{equation}\label{Pi_QOmega_angular_response}
\Pi_{\mathbb R^3\times Q_\Omega}(\mathbf k)=\frac{1}{4\pi}\int_\Omega d\Omega(\mathbf n)\,\left(1-\mathbf n_{\mathbf k}\cdot\mathbf n\right)
\end{equation}
for almost every $\mathbf k\in\mathbb R^3_*$, where $\mathbf n_{\mathbf k}:=\mathbf k/|\mathbf k|$. Moreover, it satisfies the area bounds
\begin{equation}\label{Pi_QOmega_area_bounds}
\sigma(\Omega)^2
\leq
\Pi_{\mathbb R^3\times Q_\Omega}(\mathbf k)
\leq
1-[1-\sigma(\Omega)]^2
\end{equation}
for almost every $\mathbf k\in\mathbb R^3_*$.
\end{proposition}

The bounds in Eq.~\eqref{Pi_QOmega_area_bounds} are the quantitative core of the obstruction. If the angular alternative is approximately balanced, then $\sigma(\Omega)$ is close to $1/2$, and the multiplier $\Pi_{\mathbb R^3\times Q_\Omega}(\mathbf k)$ is forced to remain uniformly away from both $0$ and $1$. This clashes with the classicality requirements below, which demand, in different ways, that the same multiplier behave almost like a characteristic function.

We first combine Poincar\'e covariance with weak non-invasive repeatability. For a dichotomic coarse measurement, repeatability means that the two effects should behave, at least in one state, approximately like orthogonal alternatives. In the momentum representation this forces the multiplier of a coarse momentum-marginal effect to become almost projective somewhere in the essential sense: it must approach either $0$ or $1$.

\begin{proposition}[Weak repeatability implies extremal momentum-marginal values]
\label{prop:weak-repeatability-almost-projective}
Let $(\Gamma,\hat\Pi,\mathcal P_{\rm CG})$ be a coarse-graining scheme such that $C\mapsto\hat\Pi_C$ is a regular Poincar\'e-covariant phase-space POVM. Assume that the scheme satisfies weak non-invasive repeatability with tolerance $\epsilon$, in the sense of Definition~\ref{def:weak-noninvasive-repeatability}. Let $Q\subseteq\mathbb R^3_*$ be measurable and suppose that
$\mathbb R^3\times Q \, \triangleleft\, \mathcal P_{\rm CG}$. Then
\begin{equation}\label{projective}
\operatorname*{ess\,inf}_{\mathbf k\in\mathbb R^3_*}
\Pi_{\mathbb R^3\times Q}(\mathbf k)
<\epsilon
\quad\text{or}\quad
\operatorname*{ess\,sup}_{\mathbf k\in\mathbb R^3_*}
\Pi_{\mathbb R^3\times Q}(\mathbf k)
>1-\epsilon,
\end{equation}
where $\Pi_{\mathbb R^3\times Q}(\mathbf k)$ is the diagonal multiplier defined in Proposition~\ref{prop:diagonal-momentum-marginal}.
\end{proposition}

Thus repeatability tries to push the momentum-cone multiplier toward the extremal values expected of a sharp classical alternative. The area bounds in Eq.~\eqref{Pi_QOmega_area_bounds} show that a sufficiently balanced angular question, guaranteed by the minimal directional information gain condition, prevents this from happening.

The second obstruction uses the compatibility between the classical momentum coordinate and the quantum momentum observable. If the cell $\mathbb R^3\times Q$ is meant to represent the classical statement that the momentum lies in $Q$, then its effect should agree, at the coarse level, with the standard quantum momentum PVM. In the momentum representation this means that the multiplier $\Pi_{\mathbb R^3\times Q}(\mathbf k)$ should behave approximately like the characteristic function $\chi_Q(\mathbf k)$, at least in the weak best-case sense of Definition~\ref{def:weak-momentum-compatibility}.

\begin{proposition}[Weak momentum compatibility implies quasi-localized momentum marginals]\label{prop:weak-momentum-compatibility-quasi-localization}
Let $(\Gamma,\hat\Pi,\mathcal P_{\rm CG})$ be a coarse-graining scheme such that $C\mapsto\hat\Pi_C$ is a regular Poincar\'e-covariant phase-space POVM. Assume that the scheme satisfies weak classical--quantum momentum compatibility with tolerance $\epsilon'$, in the sense of Definition~\ref{def:weak-momentum-compatibility}. Let $Q\subseteq\mathbb R^3_*$ be measurable and suppose that $\mathbb R^3\times Q\triangleleft\mathcal P_{\rm CG}$. Then,
\begin{equation}\label{compatibility_double}
\operatorname*{ess\,sup}_{\mathbf k\in Q}\Pi_{\mathbb R^3\times Q}(\mathbf k)>1-\epsilon'
\quad\text{or}\quad
\operatorname*{ess\,inf}_{\mathbf k\in Q^{\rm c}}\Pi_{\mathbb R^3\times Q}(\mathbf k)<\epsilon',
\end{equation}

where $\Pi_{\mathbb R^3\times Q}(\mathbf k)$ is the diagonal multiplier defined in Proposition~\ref{prop:diagonal-momentum-marginal}.
\end{proposition}

Momentum compatibility requires the multiplier associated with $Q$ to become large somewhere inside $Q$, or small somewhere outside $Q$. For the angular cones selected by minimal directional information gain, the universal area bounds again forbid such behaviour when the tolerances are sufficiently small.

We can now combine the two classicality consequences with the angular alternative guaranteed by minimal directional information gain. Let $\Omega$ be the angular region supplied by Definition~\ref{def:minimal-directional-information-gain}. Since its normalized area satisfies $(1-\delta)/2\leq\sigma(\Omega)\leq1/2$, the corresponding momentum-cone multiplier is bounded away from the extremal values required by repeatability or by momentum compatibility, unless the tolerances are large enough. This gives the following two independent no-go constraints.

\begin{theorem}[No-go theorem for the classical particle limit of massless quanta]\label{thm:no-go-constraints}
Let $(\Gamma,\hat\Pi,\mathcal P_{\rm CG})$ be a coarse-graining scheme such that $C\mapsto\hat\Pi_C$ is a regular Poincar\'e-covariant phase-space POVM in the sense of Definition~\ref{def:regular-poincare-covariant-povm}. Assume that the scheme satisfies minimal directional information gain with tolerance $\delta$, in the sense of Definition~\ref{def:minimal-directional-information-gain}, with $0\leq\delta<1$.

\begin{enumerate}
\item If the scheme also satisfies weak non-invasive repeatability with tolerance $\epsilon$, in the sense of Definition~\ref{def:weak-noninvasive-repeatability}, with $0<\epsilon<1/4$, then necessarily
\begin{equation}\label{no_go_repeatability_constraint}
\delta+2\sqrt{\epsilon}>1.
\end{equation}
In particular, no such scheme can satisfy simultaneously
\begin{equation}
\epsilon<\frac{1}{6},
\qquad
\delta<\frac{1}{6}.
\end{equation}

\item If the scheme also satisfies weak classical--quantum momentum compatibility with tolerance $\epsilon'$, in the sense of Definition~\ref{def:weak-momentum-compatibility}, with $0<\epsilon'<1/4$, then necessarily
\begin{equation}\label{no_go_momentum_constraint}
\delta+2\sqrt{\epsilon'}>1.
\end{equation}
In particular, no such scheme can satisfy simultaneously
\begin{equation}
\epsilon'<\frac{1}{6},
\qquad
\delta<\frac{1}{6}.
\end{equation}
\end{enumerate}
\end{theorem}

The theorem shows that the obstruction follows from the invariant structure forced by Poincar\'e covariance on the local detection kernel, together with the existence of even a single sufficiently balanced directional question in the coarse-grained sector. Weak repeatability and weak classical--quantum momentum compatibility fail for different reasons, but both failures originate in the same massless invariant $k^\mu p_\mu$: it is too poor to localize the full momentum vector along a null direction while preserving Poincar\'e covariance.

\section{Multiple helicities}\label{sec:multiple-helicities}

Up to this point, the quantum theory we considered was the irreducible massless representation of a single helicity $h$. However, most physically relevant particles exist in reducible representations spanning multiple helicities, like the photon or graviton that carry two helicities, $h = \pm 1$ and $h=\pm2$, respectively \cite{weinberg1995quantum}. In this section we extend the framework of Sects.~\ref{Phasespace_measurements}--\ref{requirements} to this situation and show that the no-go theorem survives. The extension of the no-go theorem is not automatic, as a reducible representation admits coherences between different helicity sectors. We show that Poincar\'e covariance forces the momentum marginals to be blind to these coherences.

On the quantum side, we consider a positive-energy massless representation of the proper orthochronous Poincar\'e group
\begin{equation}\label{H_reducible}
\mathcal H=\bigoplus_{h\in H}\mathcal H_h\otimes\mathbb C^{n_h},
\end{equation}
where $H\subset\frac12\mathbb Z$ is a finite set of helicities and $\mathcal H_h=L^2(V_0^+,d\mu_0)$ carries the irreducible helicity-$h$ representation of Sect.~\ref{Phasespace_measurements}. This is the general form of a massless representation with discrete helicities and finite multiplicities \cite{wigner1939unitary,weinberg1995quantum}. We use generalized momentum-helicity eigenstates $\ket{\mathbf k,h,\alpha}$, with $h\in H$ and $\alpha\in\{1,\ldots,n_h\}$, normalized as $\braket{\mathbf k,h,\alpha|\mathbf k',h',\alpha'}=2|\mathbf k|\delta^3(\mathbf k-\mathbf k')\delta_{hh'}\delta_{\alpha\alpha'}$. Spatial translations act on $\ket{\mathbf k,h,\alpha}$ exactly as in the irreducible case, while Lorentz transformations act as
\begin{equation}\label{U_reducible}
\hat U(\Lambda)\ket{\mathbf k,h,\alpha}=e^{-ih\theta_W(\Lambda,k)}\ket{\Lambda\mathbf k,h,\alpha}.
\end{equation}
Here, $\theta_W$ is the Wigner angle \cite{weinberg1995quantum, wigner1939unitary}. It is detailed in Appendix~\ref{proof:covariant-povm-example} how it can be computed from the group structure. The momentum spectral projectors are
\begin{equation}\label{E_P_reducible}
\hat E_{\hat{\mathbf P}}(Q)=\int_Q\frac{d^3\mathbf k}{2|\mathbf k|}\,\hat{\mathbb I}_{\mathbf k},\qquad \hat{\mathbb I}_{\mathbf k}:=\sum_{h\in H}\sum_{\alpha=1}^{n_h}\ket{\mathbf k,h,\alpha}\bra{\mathbf k,h,\alpha}.
\end{equation}

On the classical side, we enlarge the phase space \eqref{Gamma_R_R_} by a helicity label,
\begin{equation}\label{Gamma_Lambda}
\Gamma_\Lambda:=\mathbb R^3\times\mathbb R^3_*\times\Lambda,
\end{equation}
where $(\Lambda,\kappa)$ is a measure space of classical helicities. The two cases of interest are a finite set $\Lambda\subset\mathbb R$ with counting measure, and an interval of $\mathbb R$ with Lebesgue measure, as classical helicity is not necessarily discrete \cite{andrzejewski2015chiral}. The reference measure on $\Gamma_\Lambda$ is $d\nu_\Lambda(\mathbf x,\mathbf p,\lambda)=d^3\mathbf x\,d^3\mathbf p\,d\kappa(\lambda)$. Since helicity is Poincar\'e invariant, the group acts as $g\cdot(\mathbf x,\mathbf p,\lambda):=(\mathbf x',\mathbf p',\lambda)$ with $(\mathbf x',\mathbf p')$ given by Eq.~\eqref{g_x_p_prime} with side-jump $\boldsymbol\Delta_\lambda$.

Nearly all definitions of Sects.~\ref{Phasespace_measurements}--\ref{requirements} can be extended by replacing $\Gamma$ with $\Gamma_\Lambda$ and $\nu$ with $\nu_\Lambda$. The classical statement that the momentum lies in $Q\subseteq\mathbb R^3_*$, irrespective of the other degrees of freedom, corresponds to the cell $\mathbb R^3\times Q\times\Lambda$, and Definitions~\ref{def:weak-momentum-compatibility} and \ref{def:minimal-directional-information-gain} are to be read with these cells. In particular, minimal directional information gain now requires
\begin{equation}\label{Omega_Lambda}
\left\{\mathbb R^3\times Q_\Omega\times\Lambda,\ \Gamma_\Lambda\setminus(\mathbb R^3\times Q_\Omega\times\Lambda)\right\}\in\mathcal P_{\rm CG}
\end{equation}
for some angular region $\Omega$ satisfying Eq.~\eqref{delta_sigma}.

The proof of Theorem~\ref{thm:no-go-constraints} relies on Propositions~\ref{prop:detection-kernel-invariant-form}--\ref{prop:angular-response-area-bounds}, reducing the momentum marginal of a cone to a diagonal operator with entries bounded by the area of $\Omega$. In the reducible case, the same chain of arguments produces, for each $\mathbf k$, a matrix acting on the internal space $\bigoplus_h\mathbb{C}^{n_h}$. Two effects could spoil the argument of the irreducible case. First, the off-diagonal elements between different helicity sectors need not be little-group invariant, since by Eq.~\eqref{U_reducible} they pick up the relative phase $e^{i(h-h')\theta_W}$. Second, the matrix-valued entries could have eigenvalues that are not controlled by the area bounds. The following proposition shows that Poincar\'e covariance removes the first effect. Its proof can be found in Appendix~\ref{proof:helicity-structure-kernel}. 

\begin{proposition}[Helicity structure of the detection kernel]\label{prop:helicity-structure-kernel}
Let $C\mapsto\hat\Pi_C$ be a regular Poincar\'e-covariant phase-space POVM on $\Gamma_\Lambda$. Then, for $\kappa$-almost every $\lambda\in\Lambda$ and almost every $\mathbf k,\mathbf p\in\mathbb R^3_*$, the following hold.
\begin{enumerate}
\item The scalar density is of the form $\rho(\mathbf x,\mathbf p,\lambda)=\rho_\lambda(|\mathbf p|)$.
\item The momentum-diagonal kernel is helicity diagonal,
\begin{equation}\label{kernel_helicity_diagonal}
\braket{\mathbf k,h,\alpha|\hat\pi_{(\mathbf 0,\mathbf p,\lambda)}|\mathbf k,h',\alpha'}=0\qquad\text{for }h\neq h'.
\end{equation}
\item For every $h\in H$ there exists a measurable positive-semidefinite $n_h\times n_h$ matrix-valued function $\Phi^{(h)}_\lambda$ such that
\begin{equation}\label{kernel_multiplicity_block}
\braket{\mathbf k,h,\alpha|\hat\pi_{(\mathbf 0,\mathbf p,\lambda)}|\mathbf k,h,\alpha'}=\big[\Phi^{(h)}_\lambda(k^\mu p_\mu)\big]_{\alpha\alpha'}.
\end{equation}
\end{enumerate}
\end{proposition}

The next proposition shows that, in the case of momentum cones, the matrices on the helicity-sectors are multiples of the identity. This resolves the second potential effect discussed before Proposition~\ref{prop:helicity-structure-kernel}. The proof is presented in Appendix~\ref{proof:helicity-blind-marginal}.

\begin{proposition}[Momentum marginals of cones are helicity blind]\label{prop:helicity-blind-marginal}
Let $C\mapsto\hat\Pi_C$ be a regular Poincar\'e-covariant phase-space POVM on $\Gamma_\Lambda$. Then, for every measurable $Q\subseteq\mathbb R^3_*$, the momentum-marginal effect is diagonal in momentum and block diagonal in helicity,
\begin{equation}\label{Pi_Q_reducible}
\hat\Pi_{\mathbb R^3\times Q\times\Lambda}=\int_{\mathbb R^3_*}\frac{d^3\mathbf k}{2|\mathbf k|}\sum_{h\in H}\sum_{\alpha,\alpha'=1}^{n_h}\big[\Pi^{(h)}_Q(\mathbf k)\big]_{\alpha\alpha'}\ket{\mathbf k,h,\alpha}\bra{\mathbf k,h,\alpha'},
\end{equation}
with matrix-valued multipliers
\begin{equation}\label{Pi_Q_reducible_multiplier}
\Pi^{(h)}_Q(\mathbf k)=\frac{(2\pi\hbar)^3}{2|\mathbf k|}\int_Qd^3\mathbf p\int_\Lambda d\kappa(\lambda)\,\rho_\lambda(|\mathbf p|)\,\Phi^{(h)}_\lambda(k^\mu p_\mu).
\end{equation}
Moreover, for every measurable angular region $\Omega\subseteq S^2$ and every $h\in H$,
\begin{equation}\label{Pi_QOmega_reducible}
\Pi^{(h)}_{Q_\Omega}(\mathbf k)=\Pi_{\mathbb R^3\times Q_\Omega}(\mathbf k)\,\mathbb I_{n_h},\qquad \Pi_{\mathbb R^3\times Q_\Omega}(\mathbf k)=\frac{1}{4\pi}\int_\Omega d\Omega(\mathbf n)\,\left(1-\mathbf n_{\mathbf k}\cdot\mathbf n\right),
\end{equation}
for almost every $\mathbf k\in\mathbb R^3_*$.
\end{proposition}

Consequently $\hat\Pi_{\mathbb R^3\times Q_\Omega\times\Lambda}=\int_{\mathbb R^3_*}\frac{d^3\mathbf k}{2|\mathbf k|}\Pi_{\mathbb R^3\times Q_\Omega}(\mathbf k)\hat{\mathbb I}_{\mathbf k}$ is a function of the momentum operator alone, and its multiplier satisfies the area bounds \eqref{Pi_QOmega_area_bounds}. The no-go theorem then follows exactly as before.

\begin{theorem}[No-go theorem for multiple helicities]\label{thm:no-go-multiple-helicities}
Let $\mathcal H$ be a massless representation of the form \eqref{H_reducible}, let $(\Lambda,\kappa)$ be an arbitrary $\sigma$-finite measure space of classical helicities, and let $(\Gamma_\Lambda,\hat\Pi,\mathcal P_{\rm CG})$ be a coarse-graining scheme such that $C\mapsto\hat\Pi_C$ is a regular Poincar\'e-covariant phase-space POVM on $\Gamma_\Lambda$. Assume that the scheme satisfies minimal directional information gain with tolerance $0\leq\delta<1$, in the sense of Eq.~\eqref{Omega_Lambda}. Then the conclusions of Theorem~\ref{thm:no-go-constraints} hold: weak non-invasive repeatability with tolerance $0<\epsilon<1/4$ implies $\delta+2\sqrt{\epsilon}>1$, and weak classical--quantum momentum compatibility with tolerance $0<\epsilon'<1/4$ implies $\delta+2\sqrt{\epsilon'}>1$. In particular, neither requirement can be satisfied together with $\delta<1/6$ if $\epsilon<1/6$, respectively $\epsilon'<1/6$.
\end{theorem}

Note that the theorem does not require the coarse-grained sector to resolve helicity at all. A scheme whose classical outcomes are blind to the classical helicity label is obstructed in the same way as a scheme that resolves it, as the obstruction concerns outcomes that span all available classical helicities.

\section{Lifting the massless obstruction} \label{massive}

The no-go results of Theorems~\ref{thm:no-go-constraints} and~\ref{thm:no-go-multiple-helicities} raise the question whether the classical particle requirements can be satisfied at all. We now show that they can be fulfilled for arbitrary tolerances $\epsilon,\epsilon'$ and $\delta$ in the massive constructions considered below. We retain the position--momentum phase space and begin with the familiar nonrelativistic coherent-state construction, before treating the massive relativistic setting. These examples establish that the obstruction of the previous sections originates in the massless nature of the particles we considered.

\begin{figure}[h]
    \centering
    \includegraphics[width=\textwidth]{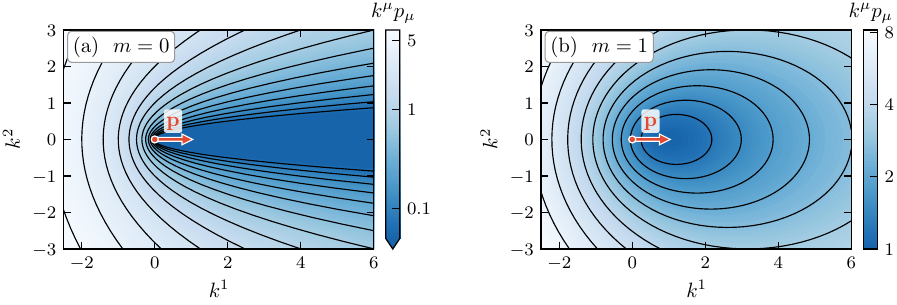}
    \caption{
        Origin of the massless obstruction. The figure illustrates why the construction of momentum-localized coherent states fails for massless particles. Little-group invariant wavefunctions can only depend on the scalar quantity $k^\mu p_\mu$, however, in the massless case, the sets on which this quantity is constant are non-compact, so such wavefunctions cannot be normalized.  Concretely, the plot shows $k^\mu p_\mu =\sqrt{m^2+|\mathbf{k}|^2}\sqrt{m^2+|\mathbf{p}|^2}-\mathbf{k}\cdot\mathbf{p}$
in the $(k^1,k^2)$-plane, for $c=1$. The black curves indicate level sets of constant $k^\mu p_\mu$. Panel (a) shows the massless case, $m=0$, where these level sets are non-compact parabolas. Panel (b) shows the massive case, $m=1$, where the same level sets become compact ellipses. Thus, while the massive invariant permits localized wave packets peaked around $\mathbf p$, the massless invariant does not.
    }
    \label{fig:level_sets}
\end{figure}

Both constructions use coherent states, corresponding to a rank-$1$ choice of seed $\hat \pi_{(\mathbf{0},\mathbf{p})} = \ket{\psi_{\mathbf{0},\mathbf{p}}}\bra{\psi_{\mathbf{0},\mathbf{p}}}$. Their localization allows a state to be concentrated inside a sufficiently coarse cell, so that repeatability and momentum compatibility follow from bounds on the probability outside that cell. For the relativistic construction, the key distinction from the massless case is the geometry of the invariant $k^\mu p_\mu$. A wavefunction depending only on this invariant cannot be normalized in the massless case, where its level sets are non-compact, whereas the massive level sets are compact, as shown in Fig.~\ref{fig:level_sets}.

\subsection{Massive nonrelativistic particle}\label{Massive_nonrelativistic_particles}

In the nonrelativistic quantum theory, the Poincaré group is replaced by the Heisenberg group, which is generated by the momentum operator $\hat{\mathbf{P}}$, implementing translations in position $\hat{U}(T_\mathbf{a}) = \exp (- i \mathbf{a} \cdot \hat{\mathbf{P}}/\hbar)$, and the position operator $\hat{\mathbf{X}}$, implementing translations in momentum (or Galilean boosts) $\hat{U}(G_\mathbf{q}) = \exp (i \mathbf{q} \cdot \hat{\mathbf{X}}/\hbar)$. The eigenstates of the generators are normalized so that $\braket{\mathbf{k} | \mathbf{k}'} = \delta^3(\mathbf{k}-\mathbf{k}')$ and  $\braket{\mathbf{x} | \mathbf{x}'} = \delta^3(\mathbf{x}-\mathbf{x}')$.

At the classical level, nonrelativistic massive particles are described by the phase space $\Gamma = \mathbb{R}^3 \times \mathbb{R}^3$ containing states composed of position and momentum $(\mathbf{x},\mathbf{p})\in\Gamma$. Reference frame transformations are given by the Galilean group, consisting of space translations, time translations, rotations, and Galilean boosts. Here, we focus on space translations $T_{\mathbf{a}}$ and Galilean boosts $G_{\mathbf{q}}$ as they already allow us to explore the entire phase-space and are expected to emerge from their quantum counterpart in the Heisenberg group.

Following \cite{zhang1990coherent} and based on the original work of Schr\"odinger in \cite{schrodinger1926stetige}, we define the non-relativistic coherent states as
\begin{equation}\label{x_p_0_0}
    \ket{\mathbf{x},\mathbf{p}} := \hat{U}(T_\mathbf{x}) \hat{U}(G_\mathbf{p})\ket{\mathbf{0},\mathbf{0}},
\end{equation}
where the seed $\ket{\mathbf{0},\mathbf{0}}$ has a Gaussian wavefunction,
\begin{equation}
    \psi_{\mathbf{0},\mathbf{0}}(\mathbf{x}) := \braket{\mathbf{x} | \mathbf{0},\mathbf{0}} =  \frac{1}{(\pi \sigma^2)^{3/4}} \exp\!\left( - \frac{|\mathbf{x}|^2}{2 \sigma^2} \right),
\end{equation}
for some fixed spread $\sigma > 0$. These states have been shown to be good candidates to prove the quantum-to-classical transition via coarse-graining in the nonrelativistic regime \cite{bibak2026classicallimitquantummechanics}. 

To construct a coarse-graining scheme according to our axiomatic approach, we identify the operator-valued density and measure in \eqref{Pi_C_density} with 
\begin{align}\label{coherent_states}
    &\hat{\pi}_{(\mathbf{x},\mathbf{p})} = \ket{\mathbf{x},\mathbf{p}} \bra{\mathbf{x},\mathbf{p}},\\
    &d\mu(\mathbf{x}, \mathbf{p}) = \frac{d^3 \mathbf{x} d^3 \mathbf{p}}{(2\pi \hbar)^3},
\end{align}
yielding the POVM
\begin{equation} \label{Pi_C_nonrelativistic}
    \hat \Pi_C =   \int_C \frac{d^3 \mathbf{x} d^3 \mathbf{p}}{(2\pi \hbar)^3} \ket{\mathbf{x}, \mathbf p} \bra{\mathbf{x}, \mathbf p}.
\end{equation}

This construction defines a regular Heisenberg covariant phase-space POVM, which we define analogously as for the Poincar\'e group in Def.~\ref{def:regular-poincare-covariant-povm}. Indeed, the coherent states resolve the identity,
\begin{equation}
    \hat \Pi_\Gamma =  \frac{1}{(2\pi \hbar)^3} \int_\Gamma d^3 \mathbf{x} d^3 \mathbf{p} \ket{\mathbf{x}, \mathbf p} \bra{\mathbf{x}, \mathbf p} = \mathbb{\hat I},
\end{equation}
transform covariantly,
\begin{equation}
    \hat U(T_{\mathbf{a}}) \hat U(G_{\mathbf{q}}) \ket{\mathbf{x},\mathbf{p}} \bra{\mathbf{x},\mathbf{p}} \hat U^\dagger(G_{\mathbf{q}}) \hat U^\dagger(T_{\mathbf{a}})  = \ket{\mathbf{x} + \mathbf{a},\mathbf{p} + \mathbf{q}} \bra{\mathbf{x} + \mathbf{a},\mathbf{p} + \mathbf{q}},
\end{equation} 
and the measure $d^3 \mathbf{x} d^3 \mathbf{p}/(2\pi \hbar)^3$ is invariant under the Galilei group. Moreover, the coherent states are approximately orthogonal,
\begin{equation} \label{overlap}
    |\braket{\mathbf{x},\mathbf{p}|\mathbf{x'},\mathbf{p'}}|^2 = \exp \left( -\frac{|\mathbf{x}-\mathbf{x'}|^2}{2\sigma^2} - \frac{\sigma^2 |\mathbf{p}-\mathbf{p'}|^2}{2\hbar^2} \right).
\end{equation}

This approximate orthogonality ensures that the POVM elements $\hat \Pi_C$ become nearly projective for coarse phase-space cells. Here, coarseness is naturally measured with respect to the dimensionless phase-space distance
\begin{equation}
    d_{\sigma} ((\mathbf{x},\mathbf{p}),(\mathbf{x'},\mathbf{p'}))^2 = \frac{|\mathbf{x}-\mathbf{x}'|^2}{2\sigma^2} + \frac{\sigma^2|\mathbf{p}-\mathbf{p}'|^2}{2\hbar^2}.
\end{equation}
For $R>0$ and $(\mathbf{x}_0,\mathbf{p}_0)\in\Gamma$, we denote the corresponding open ball (or rather ellipsoid in Euclidean coordinates) by
\begin{equation}
    B^\sigma_R(\mathbf{x_0},\mathbf{p_0}) := \{(\mathbf{x},\mathbf{p}) \in \Gamma \, : \, d_\sigma((\mathbf{x_0},\mathbf{p_0}), (\mathbf{x},\mathbf{p})) < R \}.
\end{equation}

Then, one can prove that Eq.~\eqref{PVMness} holds for partitions that contain a sufficiently coarse cell.  The concrete statement is the following.
\begin{proposition} [Weak non-invasive repeatability for massive nonrelativistic particles] \label{inf_repeatability_bound_nonrelativistic}
    Let $P = \{C_i\}_{i=1}^N \in \mathcal P_\Gamma$ be a partition so that there exists $C \in P$, $(\mathbf{x_0},\mathbf{p_0}) \in C$ and $R>0$ with $B^\sigma_R(\mathbf{x_0},\mathbf{p_0}) \subseteq C$. Then,

    \begin{equation}
        \inf_{\mathcal I^P\in\mathfrak I(P)}
\inf_{\hat \rho \in S(\mathcal H)}
\sum_{i=1}^N\sum_{j\neq i}
\operatorname{tr}\!\left[
\hat\Pi_{C_j}\,
\mathcal I_i^P\bigl(\hat \rho\bigr)
\right] \le 1 - \left [1-e^{-R^2} \left(1 + R^2 + \frac{R^4}{2}\right) \right]^2.
    \end{equation}
\end{proposition}

The proof is shown in Appendix~\ref{inf_repeatability_bound_nonrelativistic_proof}. A similar argument can be made to show that Eq.~\eqref{inf_compatibility} is fulfilled for partitions that contain a coarse enough cell. This is based on the coherent states being approximate eigenvectors of the momentum operator,

\begin{align} \label{overlap_momentum_position}
\left|\braket{\mathbf{x},\mathbf{p}|\mathbf{k}}\right|^2
&=
\left(\frac{\sigma^2}{\pi\hbar^2}\right)^{3/2}
\exp\!\left[
-\frac{\sigma^2}{\hbar^2}
|\mathbf{k}-\mathbf{p}|^2
\right].
\end{align}

\begin{proposition} [Weak compatibility between classical and quantum momentum for massive nonrelativistic particles] \label{inf_momentum_compatibility_nonrelativistic}
    Let $P = \{\mathbb R^3 \times Q_i\}_{i=1}^N \in \mathcal P_\Gamma$ be a partition so that there exists $C \in P$, $(\mathbf{x_0},\mathbf{p_0}) \in C$ and $R'>0$ with $B^\sigma_{R'}(\mathbf{x_0},\mathbf{p_0}) \subseteq C$. Then,

    \begin{equation}
        \inf_{\hat \rho \in S(\mathcal H)}\sum_{i=1}^N\sum_{j\neq i} \operatorname{tr}\!\left[ \hat{\Pi}_{\mathbb R^3\times Q_j}\hat E_{\hat{\mathbf P}}(Q_i) \hat \rho \hat E_{\hat{\mathbf P}}(Q_i)\right] \le 1 - \left[\operatorname{erf}\left(\frac{R'}{\sqrt{2}}\right)-\sqrt{\frac{2}{\pi}}R'e^{-R'^2/2}\right]^2.
    \end{equation}
\end{proposition}

The proof is shown in Appendix~\ref{inf_momentum_compatibility_nonrelativistic_proof}. An equivalent proposition could be formulated for $\mathbf{\hat X}$ instead of $\mathbf{\hat P}$, but we do not need it here, as the original relativistic requirement \eqref{inf_compatibility} can only be formulated about momentum.

Thus, both classicality conditions in Eqs.~\eqref{PVMness} and \eqref{inf_compatibility} are fulfilled as long as at least one cell of a partition is coarse enough, meaning it contains a $d_\sigma$-ball of sufficient radius. This is a property that is preserved under coarsening, so the definition
\begin{equation} \label{P_CG_nonrelativistic}
    \mathcal P_{\mathrm{CG}}(R_{\mathrm{min}}) = \{P \in \mathcal P_\Gamma \, : \, \exists C \in P, (\mathbf{x_0},\mathbf{p_0}) \in C \text{ so that } B^\sigma_{R_{\mathrm{min}}}(\mathbf{x_0},\mathbf{p_0}) \subseteq C\} 
\end{equation}
is consistent with Eq.~\eqref{closure_under_coarsening_definition}. In light of Def.~\ref{def:minimal-directional-information-gain}, note that the hemispheres $C^\pm = \{(\mathbf x, \mathbf p) \in \Gamma \, : \, \pm p^1 > 0 \}$ contain ellipses of arbitrary radius and imply that this definition of $\mathcal P_{\mathrm{CG}}$ fulfills the conditions of Def.~\ref{def:minimal-directional-information-gain} for any value of $\delta \ge 0$.

Rectangles of some appropriate minimum side length $\ell(R_{\mathrm{min}}) > 0$ also contain $d_\sigma$-balls of a given size. We now show that this allows us to fulfill Eq.~\eqref{dynkin_inferential_completeness}.

\begin{proposition} [Inferential completeness from coarse rectangles] \label{inferential_completeness_rectangles}
    Consider any $\ell \ge 0$ and denote as $\mathcal R_\ell$ the family of all rectangles with side-length $\ge \ell$. Denote as $\mathcal D(\mathcal R_\ell)$ the minimal Dynkin system that contains $\mathcal R_\ell$. Then,
    \begin{equation}
        \overline{\mathcal D(\mathcal R_\ell)}^\nu = \sigma_\Gamma.
    \end{equation}
Accordingly, $\mathcal R_\ell \subseteq \mathcal G(\mathcal P_{\mathrm{CG}})$
    is a sufficient condition for Eq.~\eqref{dynkin_inferential_completeness}.
\end{proposition}

The proof is shown in Appendix~\ref{inferential_completeness_rectangles_proof}. The above discussion, in particular Propositions~\ref{inf_repeatability_bound_nonrelativistic}, \ref{inf_momentum_compatibility_nonrelativistic} and \ref{inferential_completeness_rectangles} imply the following theorem, in contrast to Theorems~\ref{thm:no-go-constraints} and~\ref{thm:no-go-multiple-helicities}.

\begin{theorem}[Existence of a minimal classical particle limit for massive nonrelativistic quanta]
\label{coarse_graining_nonrelativistic}

Let
\begin{equation}
    \hat \Pi_C
    =
    \int_C
    \frac{d^3\mathbf{x}\,d^3\mathbf{p}}
         {(2\pi\hbar)^3}
    \ket{\mathbf{x},\mathbf{p}}
    \bra{\mathbf{x},\mathbf{p}},
\end{equation}
with the massive nonrelativistic coherent states of Eq.~\eqref{x_p_0_0}, and consider the family of admissible coarse-grained partitions of Eq.~\eqref{P_CG_nonrelativistic}.
Then, the triple $(\Gamma, \hat\Pi, \mathcal P_{\mathrm{CG}}(R_{\mathrm{min}}))$ defines a coarse-graining scheme with a regular Heisenberg-covariant phase-space POVM $C \mapsto \hat\Pi_C$. 

Moreover, for any choice of tolerances $\epsilon,\epsilon',\delta \in (0,1)$,
there exists a radius,
\begin{equation}\label{particle_resolution_threshold}
R_{\mathrm{min}}(\epsilon,\epsilon')
:=
\max\!\left\{
3,\,
\sqrt{4\log\!\frac{8\sqrt2}{\min\{\epsilon,\epsilon'\}}}
\right\},
\end{equation}
    such that the scheme satisfies
\begin{enumerate}    
    \item weak non-invasive repeatability with tolerance $\epsilon$ in the sense of Def.~\ref{def:weak-noninvasive-repeatability};
    \item weak classical–quantum momentum compatibility with tolerance $\epsilon'$ in the sense of Def.~\ref{def:weak-momentum-compatibility};
    \item minimal directional information gain with tolerance $\delta$ in the sense of Def.~\ref{def:minimal-directional-information-gain}.
\end{enumerate}
\end{theorem}

For $R\ge3$, the error bounds in Propositions~\ref{inf_repeatability_bound_nonrelativistic} and~\ref{inf_momentum_compatibility_nonrelativistic} are both bounded from above by $4\sqrt2\,e^{-R^2/4}$. Thus, choosing $R\ge R_{\mathrm{min}}(\epsilon,\epsilon')$ makes each at most $\frac12\min\{\epsilon,\epsilon'\}$, ensuring both required inequalities. In Appendix~\ref{stronger_classicality_nonrelativistic}, we go a step further. We motivate that, by demanding that all cells of a coarse partition are coarse, instead of only one, it is possible to satisfy more than just the minimal classicality conditions discussed in this work. Indeed, for a truly emergent classical theory, we would expect that Defs. \ref{def:weak-noninvasive-repeatability} and \ref{def:weak-momentum-compatibility} do not only hold in the best case, represented by the inf, but rather for most states that could be prepared in an experiment. 

\subsection{Massive relativistic particle} \label{massive_relativistic_particle}

Having established the nonrelativistic construction, we now return to the relativistic setting and consider particles with non-zero mass, satisfying the on-shell condition $p^\mu p_\mu = m^2 > 0$. Momentum eigenstates are normalized according to $\braket{\mathbf{k}|\mathbf{k}'} = 2k^0\delta^3(\mathbf{k}-\mathbf{k}')$, where $k^0=\sqrt{m^2+|\mathbf{k}|^2}$. As massive particles have a nonvanishing 4-momentum at rest, we can take $\Gamma=\mathbb R^6$. 

For simplicity, we will present the construction and derive the minimal classicality bounds for spin zero. The POVM construction and the resulting minimal classicality bounds we will derive can be extended to arbitrary reducible spin representations while retaining the classical target $\Gamma=\mathbb R^6$. For a single spin $s$, one replaces the scalar seed $\hat \pi_{(\mathbf 0,\mathbf 0)} = \ket{\mathbf 0,\mathbf 0}\bra{\mathbf 0,\mathbf 0}$ by $\ket{\mathbf 0,\mathbf 0}\bra{\mathbf 0,\mathbf 0}\otimes\hat{\mathbb I}_{2s+1}$. For a reducible representation $\mathcal H=\bigoplus_s(\mathcal H_{s}\otimes \mathbb{C}^{n_s})$, the POVMs from the irreducible representations can be combined as
\begin{equation}
\hat\Pi_C = \bigoplus_s \left(\hat\Pi_C^{(s)}\otimes\hat{\mathbb I}_{n_s}\right).
\end{equation}

Unlike the massless case, a massive relativistic particle admits a classical limit. This is because, analogously to the non-relativistic case, we can construct a coarse-graining scheme based on a family of coherent states. For a fixed inverse width parameter $\alpha > 0$, this family is given by a standard Perelomov construction \cite{perelomov1977generalized,zhang1990coherent},
\begin{equation}\label{x_p_0_0_relativistic}
    \ket{\mathbf{x},\mathbf{p}} := \hat{U}(T_\mathbf{x}) \hat{U}(\Lambda_{\mathbf{0}\mapsto\mathbf{p}})\ket{\mathbf{0},\mathbf{0}},
\end{equation} 
where $\Lambda_{\mathbf{0}\mapsto\mathbf{p}}$ is any Lorentz transformation mapping the zero momentum $(m,\mathbf{0})$ into $(p^0,\mathbf{p})$ and 
\begin{equation}
    \psi_{\mathbf{0},\mathbf{0}}(\mathbf{k}) := \braket{\mathbf{k}|\mathbf{0},\mathbf{0}} =\left(\frac{\alpha}{\pi m K_1(2\alpha m)}\right)^{1/2}\exp\!\left(-\alpha \sqrt{m^2+|\mathbf k|^2}\right).
\end{equation}
Here, $K_\nu$, for $\nu \in \mathbb C$, denotes the modified Bessel function of the second kind, with integral representation
\begin{equation}
    K_\nu(z)=\int_0^\infty dt \, e^{-z\cosh t}\cosh(\nu t),
\qquad \operatorname{Re}(z)>0.
\end{equation}

This choice of seed $\ket{\mathbf{0},\mathbf{0}}$ is motivated by the fact that it saturates the relativistic uncertainty relation between momenta $\mathbf{\hat P}$ and boost generators $\mathbf{\hat K}$,
\begin{equation}
    \Delta_{\ket \psi} \hat P^i \Delta_{\ket \psi} \hat K^i \ge \frac{\hbar}{2}\braket{\psi|\hat P^0|\psi}.
\end{equation}

These coherent states, together with alternative proposals for relativistic coherent states, have been studied in \cite{kowalski2018coherent}. Just as in the nonrelativistic case, we use the coherent states to construct an operator-valued density and a measure.
\begin{align}
    &\hat{\pi}_{(\mathbf{x},\mathbf{p})} = \ket{\mathbf{x},\mathbf{p}} \bra{\mathbf{x},\mathbf{p}}, \label{Pi_x_p_massive_relativistic}\\
    &d\mu(\mathbf{x},\mathbf{p}) = \frac{K_1(2\alpha m)}{K_2(2\alpha m)} \frac{d^3 \mathbf{x} d^3 \mathbf{p}}{(2\pi \hbar)^3}, \label{dmu_massive_relativistic}
\end{align}
leading to the map 
\begin{equation} \label{massive_relativistic_Pi_C}
    C \mapsto \hat \Pi_C = \frac{K_1(2\alpha m)}{K_2(2\alpha m)} \int_C \frac{d^3 \mathbf{x} d^3 \mathbf{p}}{(2\pi \hbar)^3} \ket{\mathbf{x},\mathbf{p}} \bra{\mathbf{x},\mathbf{p}}.
\end{equation}

This is a phase-space measurement as defined in section \ref{Phasespace_measurements}, because, just as in the nonrelativistic case, the coherent states resolve the identity,
\begin{equation} \label{res_identity_massive_rel}
\frac{K_1(2\alpha m)}{K_2(2\alpha m)} \int_\Gamma \frac{d^3 \mathbf{x} d^3 \mathbf{p}}{(2\pi \hbar)^3} \ket{\mathbf{x},\mathbf{p}} \bra{\mathbf{x},\mathbf{p}} = \mathbb{\hat I}.
\end{equation}
The proof is shown in Appendix \ref{res_identity_massive_rel_proof}. 

Importantly, $\psi_{\mathbf{0},\mathbf{0}}(\mathbf{k})$ only depends on $|\mathbf{k}|$ and is therefore invariant under the massive little group $\mathrm{LG}(\mathbf{0}) = SO(3)$. This implies that the coherent states transform covariantly under the Poincar\'e group. More precisely, if $(a,\Lambda) \in P_+^\uparrow$ maps $(\mathbf x,\mathbf p)$ to $(\mathbf{x'},\mathbf{p'})$, then Appendix \ref{covariance_massive_relativistic_proof} shows that
\begin{equation} \label{covariance_massive_relativistic}
    \hat{U}(T_a) \hat{U}(\Lambda) \ket{\mathbf{x},\mathbf{p}} = \ket{\mathbf{x'},\mathbf{p'}}.
\end{equation}
Since the measure \eqref{dmu_massive_relativistic} is also invariant under Euclidean transformations, this implies that $C \mapsto \hat \Pi_C$ in Eq.~\eqref{massive_relativistic_Pi_C} is a regular Poincar\'e covariant POVM.

The massive relativistic coherent states will allow us to define a coarse-graining scheme that satisfies weak non-invasive repeatability and weak compatibility between quantum and classical momentum. Just as in the non-relativistic case, this relies on two fundamental properties. First of all, the coherent states are approximately orthogonal. As shown in Appendix \ref{overlap_relativistic_proof}, 
\begin{equation} \label{overlap_relativistic}
    |\braket{\mathbf{x},\mathbf p| \mathbf{x'},\mathbf{p'}}|^2 =  \left| \frac{2\alpha m}{K_1(2\alpha m)} \frac{K_1\!\left(m\sqrt{\xi^2}\right)}{m\sqrt{\xi^2}}\right|^2,
\end{equation}
with
\begin{equation}
    \xi^\mu = \frac{\alpha}{m}(p+p')^\mu+ \frac{i}{\hbar}(x-x')^\mu,
\end{equation}
and where $x^\mu = (0,\mathbf x)$ and $p^\mu = (p^0,\mathbf p)$. Although considerably more intricate than the Gaussian overlap \eqref{overlap} of the non-relativistic coherent states, this expression exhibits the same qualitative behavior. Indeed, using the asymptotic form
\begin{equation}
K_1(z)
\sim
\sqrt{\frac{\pi}{2z}}\,e^{-z},
\end{equation}
one finds that the overlap decays rapidly. Secondly, the coherent states remain sharply peaked in momentum space,
\begin{equation}
    \braket{\mathbf{k}|\mathbf{x},\mathbf{p}}
    =
    \sqrt{\frac{\alpha}{\pi mK_1(2\alpha m)}}
    \exp\!\left(-\frac{i}{\hbar}\mathbf{k}\cdot\mathbf{x}\right)
    \exp\!\left(-\alpha\frac{p \cdot k}{m}\right),
\end{equation}
and are therefore concentrated around $\mathbf k=\mathbf p$, since $p\cdot k\ge m^2$, with equality if and only if $\mathbf k=\mathbf p$.

Approximate orthogonality and momentum localization provide the relativistic counterparts of the properties underlying Theorem \ref{coarse_graining_nonrelativistic}. Before stating the resulting relativistic theorem, we introduce the hyperbolic balls, more suited to the hyperbolic geometry of Minkowski space,
\begin{equation}
    H_R(\mathbf p) := \left \{\mathbf k \in \mathbb R^3 \, : \, \mathrm{arccosh}\left(\frac{p_\mu k^\mu}{m^2}\right) \le R \right\}.
\end{equation}

\begin{theorem}[Existence of a minimal classical particle limit for massive relativistic quanta]
\label{coarse_graining_relativistic}

Let
\begin{equation}
    \hat \Pi_C
    =
    \frac{K_1(2\alpha m)}{K_2(2\alpha m)}
    \int_C
    \frac{d^3\mathbf{x}\,d^3\mathbf{p}}
         {(2\pi\hbar)^3}
    \ket{\mathbf{x},\mathbf{p}}
    \bra{\mathbf{x},\mathbf{p}},
\end{equation}
with the massive relativistic coherent states of Eq.~\eqref{x_p_0_0_relativistic}, and define the family of admissible coarse-grainings by
\begin{equation}
\label{P_CG_relativistic}
    \mathcal P_{\mathrm{CG}}(R_{\mathrm{min}})
    :=
    \Bigl\{
        P\in\mathcal P_\Gamma :
        \exists\, C\in P,\;
        \exists\, \mathbf p_0\in\mathbb R^3
        \text{ such that }
        \mathbb R^3\times
        H_{R_{\mathrm{min}}}(\mathbf p_0)
        \subseteq C
    \Bigr\}.
\end{equation}
Then the triple $(\Gamma, \hat\Pi, \mathcal P_{\mathrm{CG}}(R_{\mathrm{min}}))$ defines a coarse-graining scheme with a regular Poincar\'e-covariant phase-space POVM $C \mapsto \hat\Pi_C$. Moreover, for any choice of tolerances $\epsilon,\epsilon',\delta \in (0,1)$ there exists a radius $R_{\mathrm{min}} = R_{\mathrm{min}}(\epsilon,\epsilon') > 0$ such that the scheme satisfies
\begin{enumerate}
    \item weak non-invasive repeatability with tolerance $\epsilon$ in the sense of Def.~\ref{def:weak-noninvasive-repeatability};
    \item weak classical–quantum momentum compatibility with tolerance $\epsilon'$ in the sense of Def.~\ref{def:weak-momentum-compatibility};
    \item minimal directional information gain with tolerance $\delta$ in the sense of Def.~\ref{def:minimal-directional-information-gain}.
\end{enumerate}

\end{theorem}

The proof is shown in Appendix \ref{coarse_graining_relativistic_proof}. While the class of partitions \eqref{P_CG_relativistic} is sufficient to meet the minimal requirements of Theorem \ref{thm:no-go-constraints}, it is too restrictive to serve as a complete description of classical relativistic measurements since it requires one cell to span all positions. This restriction is merely a technical convenience of the present proof. Because the coherent states become approximately orthogonal in both position and momentum labels, one expects analogous results to hold for substantially more general partitions. Establishing this rigorously, as well as extending the analysis to stronger classicality requirements beyond the minimal Defs. \ref{def:weak-noninvasive-repeatability} and \ref{def:weak-momentum-compatibility}, is left for future work.

\section{Conclusions}\label{Conclusions}

We have proved a no-go theorem showing that massless quanta cannot admit a classical particle limit satisfying Poincar\'e covariance and minimal classicality requirements. The obstruction is not tied to a particular coherent-state ansatz or to a specific dynamics, environment, or pointer basis; it follows from the structure imposed by Poincar\'e covariance on massless particle phase-space measurements. The theorem should therefore be read as a restriction on the possible target theories of the classical limit, ruling out a Poincar\'e-covariant particle phase-space description of massless quanta.

For massive particles, on the other hand, these operational requirements can be met: in Sect.~\ref{massive}, we have shown that coherent-state POVMs on position--momentum phase space satisfy the minimal classicality conditions, both in the nonrelativistic and relativistic regime. The obstruction of Theorem~\ref{thm:no-go-constraints} is therefore specific to massless particles. It can be traced back to the geometry of the invariant $k^\mu p_\mu$, whose level sets are compact in the massive case and non-compact in the massless one (Fig.~\ref{fig:level_sets}).

We note, however, that these existence results only establish that the massive particle target escapes the obstruction of Theorem~\ref{thm:no-go-constraints}; they do not provide a complete account of its emergence as a classical theory. The weak classicality conditions only require the deviations from classical behaviour to be small for at least one suitably chosen preparation, whereas a full description of the classical limit would require them to be small for a broad class of physically relevant states and measurements. The bounds for massive nonrelativistic particles in Appendix~\ref{stronger_classicality_nonrelativistic} provide a first step in this direction. Extending this analysis to the relativistic particle construction is left for future work.

Our results offer an operational perspective on why, in classical electrodynamics, the electromagnetic sector is described in terms of fields rather than particles, whereas the matter sector may be described in terms of particles or currents. Possible explanations of this distinction have already been proposed in the literature. Duncan \cite{duncan2012conceptual} emphasizes the role of large occupation numbers in the emergence of classical fields, a route naturally available for bosons but restricted for fermions by the exclusion principle. Decoherence-based arguments also identify fermionic particle or current states and bosonic coherent field states as candidates for robust pointer states in appropriate models \cite{KUBLER1973405,PhysRevD.53.7327}. Our results complement these approaches with an argument based on operational requirements imposed on finite-resolution measurements, rather than on the dynamics: under these requirements, the massless electromagnetic sector admits no classical particle description, whereas massive charged matter does.

If a massless sector has a classical limit, the particle-like target is excluded by Theorem~\ref{thm:no-go-constraints}, and the field-like target remains the natural candidate. We expect the obstruction of Theorem~\ref{thm:no-go-constraints} to disappear for massless bosons once the particle phase space is replaced by a space of classical field configurations. This obstruction arises because Poincar\'e covariance prevents quantum states from being localized around a point of the massless-particle phase space, in particular around a definite energy along a fixed null direction. For a field target, one instead needs quantum states localized around classical field configurations, and such states are provided by Glauber coherent states \cite{PhysRev.131.2766}. In a Glauber coherent state, the field is peaked around a classical configuration, with the same fluctuations as in the vacuum, and a Poincar\'e transformation maps it into the Glauber coherent state peaked around the transformed configuration. Unlike in the massless-particle case, localization and covariance are therefore compatible for a field target. Showing that this compatibility leads to coarse-grained field measurements satisfying minimal classicality requirements, however, requires a formalism different from the one developed here: the classical outcome space becomes an infinite-dimensional space of field configurations, the POVM acts on the full Fock space rather than on the one-particle sector, and the operational requirements must be formulated in terms of measurements of smeared fields rather than of particle momenta. This problem goes beyond the scope of the present paper and will be addressed in a forthcoming work. The analogous question for fermions, for which ordinary complex-amplitude coherent states are unavailable, also remains open.

\section*{Acknowledgment}

We thank Č.~Brukner and A.~Soulas for inspiring and helpful discussions. R.F.~acknowledges support as a recipient of a Seal of Excellence Postdoctoral Fellowship of the Austrian Academy of Sciences (OeAW). S.F.~thanks Č. Brukner's group at the University of Vienna for their hospitality during his visit, where part of this work was carried out.

\section*{Use of artificial intelligence}
The authors used AI tools (ChatGPT and Claude) to assist with drafting and revising the manuscript, finding proofs, particularly for Theorem~2, and identifying relevant references. All arguments and references were checked by the authors, who take full responsibility for the final content.

\appendix

\section{Regularity conditions on $C\mapsto\hat{\Pi}_{C}$}\label{Regularity_conditions}

The density representation used in Eq.~\eqref{Pi_C_density} can be justified by imposing suitable regularity conditions on the phase-space POVM $C\mapsto\hat{\Pi}_{C}$. The first condition is insensitivity to measure-zero modifications of phase-space cells. Since finite-resolution cells represent events only up to details of vanishing phase-space measure, two measurable regions that differ by a $\nu$-null set should be assigned the same quantum effect. This is the operator-valued analogue of continuity with respect to the measure $\nu$.
\begin{definition}[$\nu$-continuous POVM]\label{def:null-set-insensitivity}
We say that the POVM $C\mapsto\hat{\Pi}_{C}$ is absolutely continuous with respect to $\nu$ if, for any two measurable regions $C,C'\in\sigma_\Gamma$,
\begin{equation}\label{eq:null-set-insensitivity}
\nu(C\triangle C')=0\quad\Longrightarrow\quad\hat{\Pi}_{C}=\hat{\Pi}_{C'},
\end{equation}
where $C\triangle C':=(C\setminus C')\cup(C'\setminus C)$ denotes the symmetric difference. Equivalently,
\begin{equation}\label{eq:absolute-continuity}
\nu(C)=0\quad\Longrightarrow\quad\hat{\Pi}_{C}=0.
\end{equation}
\end{definition}

The second condition is a smoothness condition on the POVM. Consider the rectangular phase-space region $R_{\mathbf x,\mathbf p}:=\big((-\infty,\mathbf x)\times(-\infty,\mathbf p)\big)\cap\Gamma$, where $(-\infty,\mathbf x):=(-\infty,x^1)\times(-\infty,x^2)\times(-\infty,x^3)$, and analogously for $(-\infty,\mathbf p)$. Operationally, small changes of the endpoints $(\mathbf x,\mathbf p)$ should not produce singular or discontinuous changes in the probabilities assigned to $R_{\mathbf x,\mathbf p}$.
\begin{definition}[Smooth POVM]\label{def:cumulative-smoothness}
We say that the POVM $C\mapsto\hat{\Pi}_{C}$ is cumulatively smooth if, for every density operator $\hat\rho\in S(\mathcal H)$, the map $(\mathbf x,\mathbf p)\mapsto\operatorname{tr}(\hat\rho\hat{\Pi}_{R_{\mathbf x,\mathbf p}})$ admits a continuous mixed derivative $\partial_{\mathbf x}\partial_{\mathbf p}$, where $\partial_{\mathbf x}\partial_{\mathbf p}:=\partial_{x^1}\partial_{x^2}\partial_{x^3}\partial_{p^1}\partial_{p^2}\partial_{p^3}
$. We denote the corresponding operator derivative by $\partial_{\mathbf x}\partial_{\mathbf p}\hat{\Pi}_{R_{\mathbf x,\mathbf p}}$, defined by $\partial_{\mathbf x}\partial_{\mathbf p}\operatorname{tr}(\hat\rho\hat{\Pi}_{R_{\mathbf x,\mathbf p}})=\operatorname{tr}(\hat\rho\,\partial_{\mathbf x}\partial_{\mathbf p}\hat{\Pi}_{R_{\mathbf x,\mathbf p}})$ for every $\hat\rho\in S(\mathcal H)$.
\end{definition}

\begin{proposition}[Density representation of the phase-space POVM]\label{prop:density-representation}
Assume that the POVM $C\mapsto\hat{\Pi}_{C}$ is continuous in the sense of Definition~\ref{def:null-set-insensitivity} and smooth in the sense of Definition~\ref{def:cumulative-smoothness}. Then, for every measurable region $C\in\sigma_\Gamma$, the phase-space effect $\hat{\Pi}_{C}$ admits the integral representation
\begin{equation}\label{Pi_C_int_C_rho_1}
\hat{\Pi}_{C}=\int_{C}d^3\mathbf x\,d^3\mathbf p\,\partial_{\mathbf x}\partial_{\mathbf p}\hat{\Pi}_{R_{\mathbf x,\mathbf p}},
\end{equation}
\end{proposition}
\begin{proof}
The countable additivity and normalization of the POVM imply that, for every density operator $\hat{\rho}\in S(\mathcal H)$, the prescription $\nu_{\hat{\rho}}(C):=\operatorname{tr}(\hat{\rho}\hat{\Pi}_{C})$ defines a probability measure on $(\Gamma,\sigma_\Gamma)$. By Eq.~\eqref{eq:absolute-continuity}, this measure is absolutely continuous with respect to the Lebesgue measure $\nu$, i.e., $\nu(C)=0 \implies \nu_{\hat{\rho}}(C)=0$. Therefore, by the Radon--Nikodym theorem \cite[Theorem~4.2.2]{cohn2013measure}, there exists a function $f_{\hat{\rho}}\in L^1(\Gamma,\nu)$, unique up to measure zero sets, such that
\begin{equation}\label{tr_Pi_rho}
\operatorname{tr}(\hat{\rho}\hat{\Pi}_{C})=\int_{C}d^3\mathbf x\,d^3\mathbf p\,f_{\hat{\rho}}(\mathbf x,\mathbf p)
\end{equation}
for every $C\in\sigma_\Gamma$.

By the smooth condition, the probability $F_{\hat{\rho}}(\mathbf x,\mathbf p):=\operatorname{tr}(\hat{\rho}\hat{\Pi}_{R_{\mathbf x,\mathbf p}})$ admits the mixed derivative
\begin{equation}
\partial_{\mathbf x}\partial_{\mathbf p}F_{\hat{\rho}}(\mathbf x,\mathbf p)=\operatorname{tr}\!\left(\hat{\rho}\,\partial_{\mathbf x}\partial_{\mathbf p}\hat{\Pi}_{R_{\mathbf x,\mathbf p}}\right).
\end{equation}
Hence, the function
\begin{equation}
g_{\hat{\rho}}(\mathbf x,\mathbf p):=\operatorname{tr}\!\left(\hat{\rho}\,\partial_{\mathbf x}\partial_{\mathbf p}\hat{\Pi}_{R_{\mathbf x,\mathbf p}}\right)
\end{equation}
reproduces the measure $\nu_{\hat{\rho}}$ on all rectangles $R_{\mathbf x,\mathbf p}$:
\begin{equation}
\nu_{\hat{\rho}}(R_{\mathbf x,\mathbf p})=\int_{R_{\mathbf x,\mathbf p}}d^3\mathbf x'\,d^3\mathbf p'\,g_{\hat{\rho}}(\mathbf x',\mathbf p').
\end{equation}
Since these rectangles generate $\sigma_\Gamma$, uniqueness of the Radon--Nikodym derivative implies $f_{\hat \rho} = g_{\hat \rho}$ and we have
\begin{equation}
\operatorname{tr}(\hat{\rho}\hat{\Pi}_{C})=\int_{C}d^3\mathbf x\,d^3\mathbf p\,\operatorname{tr}\!\left(\hat{\rho}\,\partial_{\mathbf x}\partial_{\mathbf p}\hat{\Pi}_{R_{\mathbf x,\mathbf p}}\right).
\end{equation}
Since this identity holds for every density operator $\hat{\rho}\in S(\mathcal H)$, the equality must hold for the POVM element alone, leading to Eq.~\eqref{Pi_C_int_C_rho_1}.
\end{proof}

We now want to reinterpret the right-hand side of Eq.~\eqref{Pi_C_int_C_rho_1} as the integral of an operator-valued density. However, there is an ambiguity in how the operator-valued measure element $d^3\mathbf x\,d^3\mathbf p\,\partial_{\mathbf x}\partial_{\mathbf p}\hat{\Pi}_{R_{\mathbf x,\mathbf p}}$ can be decomposed into a scalar phase-space measure and an operator-valued density. For any measurable positive function $\rho(\mathbf x,\mathbf p)>0$, we may define
\begin{equation}\label{dmu}
d\mu(\mathbf x,\mathbf p):=\rho(\mathbf x,\mathbf p)\,d^3\mathbf x\,d^3\mathbf p,\qquad \hat \pi_{(\mathbf x,\mathbf p)}:=\frac{1}{\rho(\mathbf x,\mathbf p)}\,\partial_{\mathbf x}\partial_{\mathbf p}\hat{\Pi}_{R_{\mathbf x,\mathbf p}},
\end{equation}
so that
\begin{equation}\label{Pi_x_p}
d\mu(\mathbf x,\mathbf p)\,\hat \pi_{(\mathbf x,\mathbf p)}=d^3\mathbf x\,d^3\mathbf p\,\partial_{\mathbf x}\partial_{\mathbf p}\hat{\Pi}_{R_{\mathbf x,\mathbf p}}.
\end{equation}
To keep the no-go theorem independent of this choice, we leave the positive function $\rho(\mathbf x,\mathbf p)$ unspecified. With this convention, the POVM element associated with any measurable region $C\subseteq\Gamma$ can be written as in Eq.~\eqref{Pi_C_density}.

\section{Relativistic frame dependence and local objectivity}\label{Local_objectivity}

In this appendix, we consider the unshifted transformation law, corresponding to $\lambda=0$ in Eq.~\eqref{g_x_p_prime}. Extending this geometric formulation to nonzero helicity should be possible and requires incorporating the side jump into the identification of phase-space points between frames. 

We introduce the space $\mathcal F$ of classical inertial reference frames, modeled as a torsor for $P^\uparrow_+$: for any two frames $F,F'\in\mathcal F$, there is a unique $g\in P^\uparrow_+$ such that $F'=gF$ \cite{FedidaGlowacki2026}. Thus, after choosing a reference frame $F_0\in\mathcal F$, every frame can be written uniquely as $F=gF_0$ for some $g\in P^\uparrow_+$. Elements $F\in\mathcal F$ specify an origin $o_F$ and a Lorentz frame $\lambda_F$. For each frame $F\in\mathcal F$, let $\Sigma_F$ denote the equal-time hypersurface associated with $F$. The corresponding massless-particle phase space is
\begin{equation}
\Gamma_F:=\Sigma_F\times V_0^+,
\end{equation}
where $V_0^+:=\{p\in\mathbb R^{1,3}:p^2=0,\ p^0>0\}$ is the future light cone. Thus a point of $\Gamma_F$ consists of a spacetime point on the simultaneity hypersurface of the observer $F$, together with a future-directed null four-momentum.

Once a frame $F$ is fixed, $\Gamma_F$ can be coordinatized by $(\mathbf x,\mathbf p)\in\mathbb R^3\times\mathbb R^3_*$. More explicitly, if $F$ has origin $o_F$ and Lorentz frame $\lambda_F$, the coordinate map is
\begin{equation}\label{iota_F}
\iota_F:\mathbb R^3\times\mathbb R^3_*\to\Gamma_F,\qquad \iota_F(\mathbf x,\mathbf p):=\left(o_F+\lambda_F(0,\mathbf x),\,\lambda_F(|\mathbf p|,\mathbf p)\right).
\end{equation}
In particular, the phase space used in Sect.~\ref{Phasespace_measurements} is recovered by fixing the reference frame $F_0$ and identifying $\Gamma_{F_0}$ with $\Gamma:=\mathbb R^3\times\mathbb R^3_*$ through Eq.~\eqref{iota_F}.

A relativistic phase-space measurement is therefore not a single frame-independent POVM on $(\Gamma,\sigma_\Gamma)$, but a frame-dependent assignment
\begin{equation}
\mathcal F\ni F\mapsto \hat\Pi^F,
\end{equation}
where, for each $F\in\mathcal F$,
\begin{equation}
\hat\Pi^F:\sigma_{\Gamma_F}\to\mathcal E(\mathcal H),\qquad C\mapsto\hat\Pi^F_{C},
\end{equation}
is a POVM on the measurable phase space $(\Gamma_F,\sigma_{\Gamma_F})$. The POVM introduced in Sect.~\ref{Phasespace_measurements} is recovered by fixing $F_0$ and writing
\begin{equation}
\hat\Pi_{C}:=\hat\Pi^{F_0}_{\iota_{F_0}(C)},\qquad C\in\sigma_\Gamma.
\end{equation}

The Poincar\'e group maps phase-space regions between the frame-dependent phase spaces. If $g=(a,\Lambda)\in P^\uparrow_+$ and $z=(x,p)\in\Gamma_F$, we write $gz:=(\Lambda x+a,\Lambda p)\in\Gamma_{gF}$. Accordingly, for every measurable region $C\subseteq\Gamma_F$, we denote by $gC:=\{gz:z\in C\}\subseteq\Gamma_{gF}$ its transformed phase-space region. The following definition expresses the compatibility between the classical Poincar\'e action on frame-dependent phase-space regions and its unitary implementation in the quantum theory.

\begin{definition}[Poincar\'e-covariant frame-dependent POVM]\label{def:poincare-covariant-frame-povm}
A frame-dependent family of POVMs $F\mapsto\hat\Pi^F$ is called Poincar\'e covariant if, for every $F\in\mathcal F$, every $g\in P^\uparrow_+$, and every $C\in\sigma_{\Gamma_F}$,
\begin{equation}\label{frame_covariance}
\hat\Pi^{gF}_{gC}=\hat U(g)\hat\Pi^F_{C}\hat U^\dagger(g).
\end{equation}
\end{definition}

Similarly to the reference-frame POVM $\hat\Pi$, whose density representation was introduced in Eq.~\eqref{Pi_C_density}, we assume that each frame-dependent POVM $\hat\Pi^F$ admits a density representation,
\begin{equation}\label{frame_density_POVM}
\hat\Pi^F_{C}=\int_{C}d\mu_F(z)\,\hat\pi^F_z,\qquad C\in\sigma_{\Gamma_F}.
\end{equation}
Here $z=(x,p)\in\Gamma_F$ denotes a spacetime point $x\in\Sigma_F$ together with a future-directed null momentum $p\in V_0^+$. For the reference frame $F_0$, this representation is required to agree with Eq.~\eqref{Pi_C_density}, namely
\begin{equation}\label{reference_frame_density_coordinates}
\hat\pi^{F_0}_{\iota_{F_0}(\mathbf x,\mathbf p)}=\hat\pi_{(\mathbf x,\mathbf p)},\qquad \mu_{F_0}=(\iota_{F_0})_*\mu.
\end{equation}
The corresponding normalization condition is
\begin{equation}\label{frame_density_identity}
\int_{\Gamma_F}d\mu_F(z)\,\hat\pi^F_z=\hat{\mathbb I}.
\end{equation}

The scalar measures $\mu_F$ are fixed by a covariant normalization convention. In the reference frame $F_0$, with coordinates $(\mathbf x,\mathbf p)$ on $\Gamma_{F_0}$, we take
\begin{equation}
d\mu_{F_0}=d^3\mathbf x\,d^3\mathbf p\,\rho(\mathbf x,\mathbf p).
\end{equation}
For a generic frame $F=gF_0$, we then define
\begin{equation}\label{measure_pushforward}
\mu_F:=g_*\mu_{F_0},
\end{equation}
where $g_*\mu_{F_0}$ denotes the push-forward of $\mu_{F_0}$ under the map $g:\Gamma_{F_0}\to\Gamma_F$, namely
\begin{equation}
(g_*\mu_{F_0})(C):=\mu_{F_0}(g^{-1}C),\qquad C\in\sigma_{\Gamma_F}.
\end{equation}
This convention fixes the otherwise arbitrary separation between the scalar measure and the operator-valued density: any positive rescaling of the measure can be absorbed into the inverse rescaling of the density, leaving the integrated POVM unchanged.

For a generic regular frame-dependent POVM, Eq.~\eqref{frame_density_POVM} only gives densities $\hat\pi^F_z$ defined separately on each $\Gamma_F$. To identify these densities across different hypersurfaces, we introduce a local consistency condition. In the Heisenberg picture, the state $\hat\rho$ represents the same preparation of the quantum system, while the dependence on the observer's hypersurface is carried by the effects $\hat\Pi^F_{C}$. Thus, for different frames $F$ and $F'$, the quantities $\operatorname{tr}(\hat\rho\,\hat\Pi^F_{C})$ and $\operatorname{tr}(\hat\rho\,\hat\Pi^{F'}_{C'})$ are probabilities for two different measurement effects evaluated on the same state, and are therefore operationally comparable. However, arbitrary finite cells $C\subseteq\Gamma_F$ and $C'\subseteq\Gamma_{F'}$ do not in general represent the same classical alternative. Therefore, local objectivity should not require equality of probabilities for arbitrary finite cells on different hypersurfaces. The meaningful comparison is local. When two hypersurfaces intersect, they contain common physical events. At a common phase-space point $z=(x,p)\in\Gamma_F\cap\Gamma_{F'}$, the two observers describe the same spacetime event $x$ and the same future-directed null momentum $p$. Local objectivity is the condition that, in the limit in which the measured neighbourhoods shrink to this common point, both observers assign the same probability per unit phase-space measure. This gives the following definition.

\begin{definition}[Local objectivity]\label{def:local-objectivity}
Let $\{(\hat\Pi^F,\mu_F)\}_{F\in\mathcal F}$ be a frame-dependent family of POVM--measure pairs, where each $\mu_F$ is the scalar measure used in the density representation of $\hat\Pi^F$. We say that this family satisfies local objectivity if the following condition holds. Let $F,F'\in\mathcal F$, and let $z=(x,p)\in\Gamma_F\cap\Gamma_{F'}$, i.e. $x\in\Sigma_F\cap\Sigma_{F'}$ and $p\in V_0^+$. For any pair of regular shrinking neighbourhoods $B_\epsilon^F(z)\subset\Gamma_F$ and $B_\epsilon^{F'}(z)\subset\Gamma_{F'}$ centered at $z$, with nonzero $\mu_F$- and $\mu_{F'}$-measure, and for every state $\hat\rho\in S(\mathcal H)$, one has
\begin{equation}\label{local_objectivity}
\lim_{\epsilon\to0}\left[
\frac{\operatorname{tr}\!\left(\hat\rho\,\hat\Pi^F_{B_\epsilon^F(z)}\right)}{\mu_F(B_\epsilon^F(z))}
-
\frac{\operatorname{tr}\!\left(\hat\rho\,\hat\Pi^{F'}_{B_\epsilon^{F'}(z)}\right)}{\mu_{F'}(B_\epsilon^{F'}(z))}
\right]=0.
\end{equation}
\end{definition}

This formulation is operational: it compares probabilities assigned to small but finite phase-space neighbourhoods, normalized by their reference measures, and evaluated on the same quantum state. Only under additional regularity assumptions is it translated into a pointwise statement about densities.

\begin{proposition}[Local objectivity and hypersurface-independent densities]\label{prop:local-objectivity-density}
Assume that the densities in Eq.~\eqref{frame_density_POVM} are continuous in the weak operator sense. If local objectivity holds, then whenever two frame-dependent phase spaces intersect at a point $z\in\Gamma_F\cap\Gamma_{F'}$, the corresponding densities agree:
\begin{equation}\label{density_hypersurface_independence}
\hat\pi^F_z=\hat\pi^{F'}_z.
\end{equation}
\end{proposition}

\begin{proof}
By the density representation, for every state $\hat\rho\in S(\mathcal H)$,
\begin{equation}
\frac{\operatorname{tr}\!\left(\hat\rho\,\hat\Pi^F_{B_\epsilon^F(z)}\right)}{\mu_F(B_\epsilon^F(z))}=\frac{1}{\mu_F(B_\epsilon^F(z))}\int_{B_\epsilon^F(z)}d\mu_F(z')\,\operatorname{tr}\!\left(\hat\rho\,\hat\pi^F_{z'}\right).
\end{equation}
By weak continuity of the density, the right-hand side converges to $\operatorname{tr}(\hat\rho\,\hat\pi^F_z)$ as $\epsilon\to0$. The same argument for $F'$ gives the limit $\operatorname{tr}(\hat\rho\,\hat\pi^{F'}_z)$. Local objectivity therefore implies
\begin{equation}
\operatorname{tr}\!\left(\hat\rho\,\hat\pi^F_z\right)=\operatorname{tr}\!\left(\hat\rho\,\hat\pi^{F'}_z\right)
\end{equation}
for every state $\hat\rho\in S(\mathcal H)$, and hence $\hat\pi^F_z=\hat\pi^{F'}_z$ in the weak operator sense.
\end{proof}

When the condition of Proposition~\ref{prop:local-objectivity-density} holds, the frame-dependent densities can be consistently identified at common physical events. We may then define a single operator-valued density $\hat\pi_z$ on $\mathbb R^{1,3}\times V_0^+$ by choosing any frame $F$ such that $z\in\Gamma_F$ and setting
\begin{equation}
\hat\pi_z:=\hat\pi^F_z.
\end{equation}
Equation~\eqref{density_hypersurface_independence} guarantees that this definition is independent of the chosen hypersurface. In this case, Eq.~\eqref{frame_density_POVM} becomes
\begin{equation}\label{frame_density_POVM_independent}
\hat\Pi^F_{C}=\int_{C}d\mu_F(z)\,\hat\pi_z,\qquad C\in\sigma_{\Gamma_F}.
\end{equation}

\begin{proposition}[Microscopic covariance on the reference hypersurface]\label{prop:microscopic-covariance-reference-hypersurface}
Assume that the frame-dependent family of POVMs $F\mapsto\hat\Pi^F$ is Poincar\'e covariant in the sense of Definition~\ref{def:poincare-covariant-frame-povm}. Assume also that each $\hat\Pi^F$ admits a density representation as in Eq.~\eqref{frame_density_POVM}, with covariantly normalized measures $\mu_F$, and that the corresponding densities are continuous in the weak operator sense. If local objectivity holds in the sense of Definition~\ref{def:local-objectivity}, then the reference-frame density transforms covariantly between any two points of the reference phase space related by a Poincar\'e transformation. More precisely, let $z=\iota_{F_0}(\mathbf x,\mathbf p)\in\Gamma_{F_0}$ and let $g\in P^\uparrow_+$ be such that $gz\in\Gamma_{F_0}$. If $gz=\iota_{F_0}(\mathbf x',\mathbf p')$, then
\begin{equation}\label{microscopic_covariance_coordinates_appendix}
\hat U(g)\hat\pi_{(\mathbf x,\mathbf p)}\hat U^\dagger(g)=\hat\pi_{(\mathbf x',\mathbf p')}.
\end{equation}
\end{proposition}

\begin{proof}
By Poincar\'e covariance of the integrated POVM, for every measurable $C\in\sigma_{\Gamma_{F_0}}$,
\begin{equation}
\hat\Pi^{gF_0}_{gC}=\hat U(g)\hat\Pi^{F_0}_{C}\hat U^\dagger(g).
\end{equation}
Using the density representations and the covariant normalization of the measures, $\mu_{gF_0}=g_*\mu_{F_0}$, this becomes
\begin{equation}
\int_{gC}d\mu_{gF_0}(z')\,\hat\pi^{gF_0}_{z'}=\int_{C}d\mu_{F_0}(z)\,\hat U(g)\hat\pi^{F_0}_z\hat U^\dagger(g).
\end{equation}
Changing variables $z'=gz$ on the left-hand side gives
\begin{equation}
\int_{C}d\mu_{F_0}(z)\,\hat\pi^{gF_0}_{gz}=\int_{C}d\mu_{F_0}(z)\,\hat U(g)\hat\pi^{F_0}_z\hat U^\dagger(g).
\end{equation}
Since this holds for every measurable $C$, uniqueness of the weak Radon--Nikodym density implies
\begin{equation}\label{density_covariance_between_frames}
\hat\pi^{gF_0}_{gz}=\hat U(g)\hat\pi^{F_0}_z\hat U^\dagger(g)
\end{equation}
for $\mu_{F_0}$-almost every $z\in\Gamma_{F_0}$. By weak continuity of the densities, the equality holds pointwise.

Now suppose that $gz\in\Gamma_{F_0}$. Since $gz\in\Gamma_{gF_0}$ by construction, the point $gz$ lies in the intersection $\Gamma_{F_0}\cap\Gamma_{gF_0}$. Local objectivity, together with Proposition~\ref{prop:local-objectivity-density}, gives
\begin{equation}
\hat\pi^{gF_0}_{gz}=\hat\pi^{F_0}_{gz}.
\end{equation}
Combining this with Eq.~\eqref{density_covariance_between_frames}, we obtain
\begin{equation}
\hat U(g)\hat\pi^{F_0}_z\hat U^\dagger(g)=\hat\pi^{F_0}_{gz}.
\end{equation}
Finally, the reference-frame coordinates $z=\iota_{F_0}(\mathbf x,\mathbf p)$ and $gz=\iota_{F_0}(\mathbf x',\mathbf p')$, together with Eq.~\eqref{density_covariance_between_frames} and local objectivity give Eq.~\eqref{microscopic_covariance_coordinates_appendix}.
\end{proof}

\section{Contexts and coarse partitions}\label{app:contexts}

In the main text, the set of coarse partitions $\mathcal P_{\rm CG}$ was defined by inferential completeness, Eq.~\eqref{dynkin_inferential_completeness}, and closure under coarsening, Eq.~\eqref{closure_under_coarsening_definition}. In this appendix we give an equivalent contextual definition. The contextual formulation makes explicit the role of finite statistical inference within a single classical domain, while the ordinary Dynkin closure of the union of all contexts accounts for the ideal inferential completeness of the full coarse-grained sector.

We use the notation $\operatorname{Ev}(P)$ and $\mathcal G(\mathcal A)$ introduced in the main text. The only new ingredient is the finite additive analogue of the Dynkin closure. For a nonempty finite family of partitions $\mathfrak F=\{P^{(1)},\ldots,P^{(m)}\}$, let $\mathcal D_{\rm fin}(\mathfrak F)$ be the smallest family of measurable subsets of $\Gamma$ containing $\mathcal G(\mathfrak F)$ and closed under the following finite operations:
\begin{equation}
A,B\in\mathcal D_{\rm fin}(\mathfrak F),\qquad A\cap B=\emptyset\quad\Longrightarrow\quad A\sqcup B\in\mathcal D_{\rm fin}(\mathfrak F),
\end{equation}
and
\begin{equation}
A,B\in\mathcal D_{\rm fin}(\mathfrak F),\qquad A\subseteq B\quad\Longrightarrow\quad B\setminus A\in\mathcal D_{\rm fin}(\mathfrak F).
\end{equation}
Unlike the Dynkin closure $\mathcal D(\mathcal P_{\rm CG})$ used in the main text, $\mathcal D_{\rm fin}(\mathfrak F)$ is generated only from finitely many partitions and only by finite additive operations. It therefore captures the finite statistical inferences available within a single classical context.

\begin{definition}[Context]\label{def:context}
A context is a nonempty subset $\mathcal C\subseteq\mathcal P_\Gamma$ satisfying the following finite inferential closure condition. For every nonempty finite family $\mathfrak F\subseteq\mathcal C$, and for every partition $P=\{C_i\}_{i=1}^N\in\mathcal P_\Gamma$,
\begin{equation}\label{context_closure}
\left[C_i\in\mathcal D_{\rm fin}(\mathfrak F)\quad\text{for all }i=1,\ldots,N\right]\quad\Longrightarrow\quad P\in\mathcal C.
\end{equation}
\end{definition}

\begin{definition}[Contextual set of coarse partitions]\label{def:contextual-coarse-partitions}
A contextual set of coarse partitions is a nonempty subset $\mathcal P_{\rm CG}\subseteq\mathcal P_\Gamma$ for which there exists a nonempty collection $\mathfrak C$ of contexts such that
\begin{equation}\label{contextual_union}
\mathcal P_{\rm CG}=\bigcup_{\mathcal C\in\mathfrak C}\mathcal C,
\end{equation}
and such that the $\nu$-completed Dynkin closure of $\mathcal P_{\rm CG}$ is the full measurable phase-space event space,
\begin{equation}\label{contextual_inferential_completeness}
\overline{\mathcal D}^{\,\nu}(\mathcal P_{\rm CG})=\sigma_\Gamma.
\end{equation}
\end{definition}

This definition separates two levels of inference. Inside a context $\mathcal C$, finite statistical inferences are required to remain directly classical. Across the full set $\mathcal P_{\rm CG}$, the $\nu$-completed Dynkin closure condition expresses ideal inferential completeness: arbitrary measurable events may be assigned probabilities in principle, up to operationally irrelevant $\nu$-null modifications, but need not correspond to directly classical measurements.

\begin{proposition}[Equivalence of the direct and contextual definitions]\label{prop:contextual-main-equivalence}
A nonempty subset $\mathcal P_{\rm CG}\subseteq\mathcal P_\Gamma$ is a set of coarse partitions in the sense of Definition~\ref{def:family-coarse-partitions} if and only if it is a contextual set of coarse partitions in the sense of Definition~\ref{def:contextual-coarse-partitions}.
\end{proposition}

\begin{proof}
Assume first that $\mathcal P_{\rm CG}$ satisfies the direct definition in the main text. Thus $\overline{\mathcal D}^{\,\nu}(\mathcal P_{\rm CG})=\sigma_\Gamma$ and $\mathcal P_{\rm CG}$ is closed under coarsening. For each $P\in\mathcal P_{\rm CG}$, define
\begin{equation}
\mathcal C_P:=\{P'\in\mathcal P_\Gamma:P'\succeq P\},
\end{equation}
the set of all coarsenings of $P$. By closure under coarsening, $\mathcal C_P\subseteq\mathcal P_{\rm CG}$. We now show that $\mathcal C_P$ is a context. Let $\mathfrak F\subseteq\mathcal C_P$ be a nonempty finite family, and let $Q=\{Q_a\}_{a=1}^M\in\mathcal P_\Gamma$ be such that $Q_a\in\mathcal D_{\rm fin}(\mathfrak F)$ for all $a=1,\ldots,M$. Since every partition in $\mathfrak F$ is a coarsening of $P$, every event resolved by any partition in $\mathfrak F$ is a union of cells of $P$. Therefore
\begin{equation}
\mathcal G(\mathfrak F)\subseteq\operatorname{Ev}(P).
\end{equation}
Moreover, $\operatorname{Ev}(P)$ is closed under finite disjoint unions and relative complements. Hence
\begin{equation}
\mathcal D_{\rm fin}(\mathfrak F)\subseteq\operatorname{Ev}(P).
\end{equation}
It follows that each cell $Q_a$ is a union of cells of $P$, and therefore $Q\succeq P$. Hence $Q\in\mathcal C_P$. This proves that $\mathcal C_P$ is a context.

Since $P\in\mathcal C_P$ for every $P\in\mathcal P_{\rm CG}$, and since each $\mathcal C_P$ is contained in $\mathcal P_{\rm CG}$, we have
\begin{equation}
\mathcal P_{\rm CG}=\bigcup_{P\in\mathcal P_{\rm CG}}\mathcal C_P.
\end{equation}
Together with $\overline{\mathcal D}^{\,\nu}(\mathcal P_{\rm CG})=\sigma_\Gamma$, this proves that $\mathcal P_{\rm CG}$ is a contextual set of coarse partitions.

Conversely, assume that $\mathcal P_{\rm CG}$ is a contextual set of coarse partitions. By definition, $\overline{\mathcal D}^{\,\nu}(\mathcal P_{\rm CG})=\sigma_\Gamma$. It remains to prove closure under coarsening. Let $P\in\mathcal P_{\rm CG}$ and let $P'\succeq P$. Since $\mathcal P_{\rm CG}=\bigcup_{\mathcal C\in\mathfrak C}\mathcal C$, there exists a context $\mathcal C\in\mathfrak C$ such that $P\in\mathcal C$. Every cell of $P'$ is a union of cells of $P$, hence every cell of $P'$ belongs to $\operatorname{Ev}(P)\subseteq\mathcal D_{\rm fin}(\{P\})$. Applying the context closure condition \eqref{context_closure} with $\mathfrak F=\{P\}$ gives $P'\in\mathcal C$. Since $\mathcal C\subseteq\mathcal P_{\rm CG}$, we conclude that $P'\in\mathcal P_{\rm CG}$. Thus $\mathcal P_{\rm CG}$ is closed under coarsening, completing the proof.
\end{proof}

Definition \ref{def:contextual-coarse-partitions} should not be confused with a global noncontextuality assumption \cite{KochenSpecker1967,Spekkens2005,AbramskyBrandenburger2011}. We do not require all possible coarse-grained measurements, nor all possible POVM effects, to be embedded into a single Boolean algebra or a single underlying classical probability space. A context is only an internally closed domain of finite statistical inference. Different contexts may be mutually incompatible, even though their partitions all belong to $\mathcal P_{\rm CG}$.

\begin{figure}[h]
    \centering
    \includegraphics[width=0.85\textwidth]{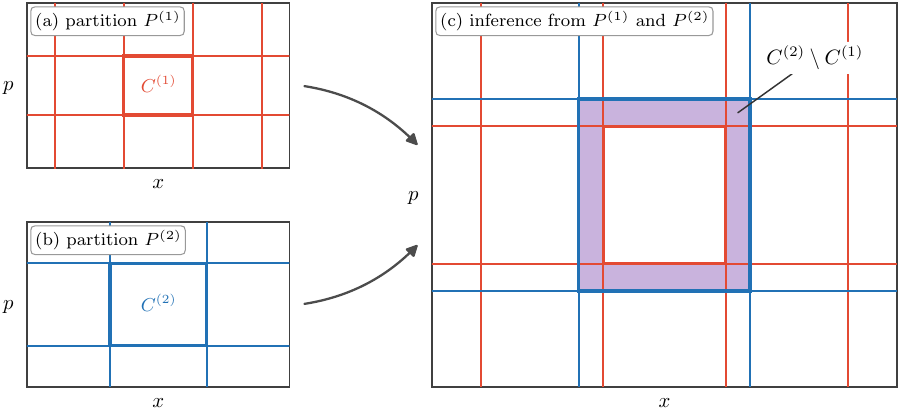}
    \caption{
        A partition may be coarse enough for its cells to behave approximately classically, but combining cells taken from different coarse partitions can generate regions that are much finer than either partition individually. Such regions can be assigned probabilities by formal statistical inference, but they need not correspond to outcomes of any admissible classical measurement. Thus, outside a fixed context, inference can lead to predictions that are not themselves classically testable. The partitions $P^{(1)}$, shown in panel (a), and $P^{(2)}$, shown in panel (b), may each satisfy Eq.~\eqref{PVMness} when considered separately. However, by recombining their cells $C^{(1)}\in P^{(1)}$ and $C^{(2)}\in P^{(2)}$ using Eq.~\eqref{p_operation_2}, one can form the difference $C^{(2)}\setminus C^{(1)}$, shaded in panel (c). This derived region is fine-grained and may fail to satisfy Eq.~\eqref{PVMness}. Within a single coarse context this problem does not arise: all statistical predictions obtained by recombining cells remain associated with coarse regions that can again be tested by classical measurements.
    }
    \label{fig:different_contexts}
\end{figure}

This distinction is important. Two partitions $P^{(1)}$ and $P^{(2)}$ may each be directly classical, while still failing to belong to one common context. If they were placed in the same context, the finite additive operations in Eqs.~\eqref{p_operation_1} and~\eqref{p_operation_2} could generate events that are too fine-grained to admit a direct classical interpretation. An example of this is shown in Fig.~\ref{fig:different_contexts}. In such a case, the two partitions should be treated as belonging to distinct contexts, say $\mathcal C^{(1)}$ and $\mathcal C^{(2)}$. This does not mean that an experimenter is forbidden from implementing measurements associated with different contexts. Rather, it means that two individually admissible coarse-grained measurements need not admit a joint classical event structure. What fails is the interpretation of all their set-theoretic recombinations as directly measurable classical events. 

If the recombination of measurements from different contexts produces a fine-grained event, the corresponding probability may still be inferable within the idealized infinite-resource extension encoded by $\overline{\mathcal D}^{\,\nu}(\mathcal P_{\rm CG})=\sigma_\Gamma$. However, the resulting event need not correspond to a directly implementable classical measurement. By contrast, within a fixed context $\mathcal C$, every partition generated by the finite inferential closure condition is required to remain directly classical. Thus each context provides an internally consistent Kolmogorovian event structure, while $\mathcal P_{\rm CG}$ as a whole provides the broader collection of directly classical partitions whose statistics may generate the full measurable phase-space structure by ideal inference.

The notion of contextuality discussed here should be distinguished from that associated with violations of Bell inequalities. When each party in a bipartite experiment chooses among coarse partitions of a fixed local phase-space POVM, the phase-space coordinate serves as a local hidden variable, thereby excluding Bell violations in any case. The contextuality at issue here instead concerns repeatability: combining individually admissible coarse-grained partitions may generate finer events for which the demanded non-invasive repeatability fails.

\section{Minimal informational content}\label{appendix_minimal_information}

Definition~\ref{def:minimal-directional-information-gain} requires the availability of a sufficiently balanced binary question about the direction of propagation. We now show a weaker and more operational formulation that allows for an arbitrary finite angular partition and bounds its normalized Shannon entropy. As in Sect.~\ref{section:information}, we assume that, a priori, every direction in $S^2$ is equally likely as a classical outcome.

\begin{definition}[Weak minimal directional information gain]\label{def:weak-minimal-directional-information-gain}
A coarse-graining scheme $(\Gamma,\hat\Pi,\mathcal P_{\rm CG})$ satisfies weak minimal directional information gain with threshold $M\in[0,1]$ if there exists a finite measurable partition $\{\Omega_i\}_{i=1}^N$ of $S^2$, with $N\ge2$, such that
\begin{equation}\label{H_log_N_M}
\frac{H(\{\Omega_i\}_{i=1}^N)}{\log N}
:=
\frac{-\sum_{i=1}^N\sigma(\Omega_i)\log\sigma(\Omega_i)}{\log N}
\ge M,
\end{equation}
and the associated phase-space partition belongs to the coarse-grained sector:
\begin{equation}
\{\mathbb R^3\times Q_{\Omega_i}\}_{i=1}^N\in\mathcal P_{\rm CG}.
\end{equation}
\end{definition}

We denote the normalized binary entropy as
\begin{equation}
h_2(s):=\frac{-s\log s-(1-s)\log(1-s)}{\log2},
\qquad 0\le s\le1.
\end{equation}

As shown in Appendix \ref{M_delta_proof}, Def.~\ref{def:minimal-directional-information-gain} implies Def.~\ref{def:weak-minimal-directional-information-gain}.

\begin{proposition}\label{M_delta_proposition}
If a coarse-graining scheme satisfies minimal directional information gain as in Def.~\ref{def:minimal-directional-information-gain} with tolerance $0\le\delta<1$, it satisfies weak minimal directional information gain as in Def.~\ref{def:weak-minimal-directional-information-gain} with threshold
\begin{equation}\label{M_delta}
M(\delta):=h_2\!\left(\frac{1-\delta}{2}\right).
\end{equation}
\end{proposition}

Thus, the entropy formulation allows more general angular partitions while retaining a quantitative requirement on their information content. For the reference tolerance $\delta=1/6$, one obtains $M(1/6)\approx0.98$. Although this value is close to one, the binary entropy is flat near its maximum, so a high normalized entropy still permits appreciable deviations from an exactly balanced partition.

The classicality bounds give two independent obstructions also under this weaker formulation, generalizing Theorem~\ref{thm:no-go-constraints}.

\begin{proposition}[Entropy formulation of the massless obstruction]\label{prop:entropy-no-go}
Let $(\Gamma,\hat\Pi,\mathcal P_{\rm CG})$ be a coarse-graining scheme whose POVM is regular and Poincar\'e covariant in the sense of Definition~\ref{def:regular-poincare-covariant-povm}. Assume that the scheme satisfies weak minimal directional information gain as in Def.~\ref{def:weak-minimal-directional-information-gain} with threshold $M$.
\begin{enumerate}
\item If it satisfies weak non-invasive repeatability as in Def.~\ref{def:weak-noninvasive-repeatability} with tolerance $0<\epsilon<1/4$, then necessarily
\begin{equation}
M<h_2(\sqrt{\epsilon}).
\end{equation}
\item If it satisfies weak classical--quantum momentum compatibility as in Def.~\ref{def:weak-momentum-compatibility} with tolerance $0<\epsilon'<1/4$, then necessarily
\begin{equation}
M<h_2(\sqrt{\epsilon'}).
\end{equation}
\end{enumerate}
Consequently, the condition
\begin{equation}\label{M_epsilon}
M\ge M(\epsilon,\epsilon')
:=
\max\{h_2(\sqrt{\epsilon}),h_2(\sqrt{\epsilon'})\}
\end{equation}
excludes each of these classicality requirements separately.
\end{proposition}

In particular, $M\ge h_2(\sqrt{1/6})$ is incompatible with weak repeatability for $\epsilon<1/6$, and with weak momentum compatibility for $\epsilon'<1/6$. These conclusions therefore apply in the same tolerance regime used in the main text. For binary partitions, inserting Eq.~\eqref{M_delta} recovers the constraints $\delta+2\sqrt{\epsilon}>1$ and $\delta+2\sqrt{\epsilon'}>1$ of Theorem~\ref{thm:no-go-constraints}.

The origin of this statement is that the combination of covariance and either classicality requirement forces one angular cell to occupy almost the entire sphere. Such a partition supplies little average information under the uniform prior. More precisely, the following result applies to the angular partition and is the central result needed to prove Proposition~\ref{prop:entropy-no-go}.

\begin{proposition}[Classicality forces a dominant angular cell]\label{M_sigma_Omega}
Let $(\Gamma,\hat\Pi,\mathcal P_{\rm CG})$ be a coarse-graining scheme whose POVM is regular and Poincar\'e covariant. Let $\{\Omega_i\}_{i=1}^N$ be a finite measurable partition of $S^2$ whose associated phase-space partition belongs to $\mathcal P_{\rm CG}$, and set
\begin{equation}
s_i:=\sigma(\Omega_i),
\qquad
r:=1-\max_i s_i.
\end{equation}
Weak non-invasive repeatability with tolerance $\epsilon$ implies $r^2<\epsilon$, while weak classical--quantum momentum compatibility with tolerance $\epsilon'$ implies $r^2<\epsilon'$. Moreover, whenever $0\le r\le1/2$,
\begin{equation}\label{entropy_dominant_cell}
\frac{H(\{\Omega_i\}_{i=1}^N)}{\log N}\le h_2(r).
\end{equation}
\end{proposition}

The proofs are given in Appendices~\ref{M_sigma_Omega_proof} and~\ref{proof:entropy-no-go}. This formulation shows that the obstruction does not depend on requiring a particular binary angular measurement: a sufficiently informative finite angular partition always conflicts with the minimal classicality requirements.

\section{Outlook towards stronger classicality requirements for massive nonrelativistic particles} \label{stronger_classicality_nonrelativistic}

We motivate that, if further constraining $\mathcal P_{\mathrm{CG}}$ to only contain partitions where all cells are coarse with respect to $d_\sigma$, the coarse-graining scheme of Theorem \ref{coarse_graining_nonrelativistic} fulfills Defs. \ref{def:weak-noninvasive-repeatability} and \ref{def:weak-momentum-compatibility} for a large class of states $\ket \psi \in \mathcal H$. Indeed, we will show that the non-classical deviations related to both definitions become small when considering states that are concentrated far enough from the boundary of the partition. Thus, if a partition is coarse enough, meaning the border region of each cell is relatively small compared to its volume, non-invasive repeatability and compatibility with quantum momentum will hold for nearly all states that could be prepared. We make the notion of concentration more concrete.

\begin{definition} [Probability mass at the border]
    The probability mass of a state $\ket{\psi} \in \mathcal H$ within distance $R>0$ of the border of a partition $P \in \mathcal P_\Gamma$ is

\begin{equation} \label{eta_R}
    \eta_R(\psi,P) := \int_{\partial_RP} \frac{d^3 \mathbf{x} d^3 \mathbf{p}}{(2\pi \hbar)^3} |\braket{\mathbf{x},\mathbf{p}|\psi}|^2.
\end{equation}

Here, for $P = \{C_i\}_{i=1}^N \in \mathcal P_\Gamma$,

\begin{equation}
    \partial_R P := \bigsqcup_{i=1}^N \{z \in C_i \, : \, d_\sigma(z,\Gamma \setminus C_i ) \le R\}.
\end{equation}
    
\end{definition}

Then, we can formulate the following bound.

\begin{proposition} [General bound on non-repeatability] \label{general_bound_non_repeatability_nonrelativistic}
Given a system in state $\ket{\psi} \in \mathcal H$ and a partition $P \in \mathcal P_\Gamma$, the probability to see a non-classical transition when repeating the experiment with a L\"uders implementation is bounded, for any $R>0$, by
\begin{equation} \label{reapitability_bound_nonrelativistic}
    \sum_{i=1}^N\braket{\psi|\hat{\Pi}_{C_i}-\hat{\Pi}_{C_i}^2|\psi} \le 8 \eta_R(\psi,P) + \tau(R),
\end{equation}
with the decreasing function 
\begin{equation}
    \tau(R) = e^{-R^2/2} (8 + 4R^2 + R^4).
\end{equation}
\end{proposition}
The proof is given in Appendix~\ref{general_bound_non_repeatability_nonrelativistic_proof}.

Eq.~\eqref{reapitability_bound_nonrelativistic} can be made small, in particular smaller than $\epsilon$, by choosing $R$ large enough so that $\tau(R) < \epsilon$ and $\ket \psi$ so that it has support sufficiently far away from $\partial_R P$, meaning $\eta_R(\psi,P) < \frac{1}{8} \left (\epsilon - \tau(R) \right)$. 

For a partition $P = \{\mathbb R^3 \times Q_i\}_{i=1}^N$, we can formulate a similar bound, related to Def.~\ref{def:weak-momentum-compatibility}, by replacing the phase-space distance with the pure momentum distance
\begin{equation}
    d_{\sigma}^{\mathbf{\hat P}}(\mathbf{k},\mathbf{p})
    =
    \frac{\sigma}{\hbar}
    |\mathbf{k}-\mathbf{p}|,
\end{equation}
defining the thickened momentum boundary 
\begin{equation}
    \partial_R^{\mathbf{\hat P}}P
    = \bigsqcup_{i=1}^N  \partial_R Q_i := 
    \bigsqcup_{i=1}^N \left\{
    \mathbf{k}\in Q_i:
    d_{\sigma}^{\mathbf{\hat P}}
    \left(
    \mathbf{k},
    \mathbb R^3\setminus Q_i
    \right)
    \le R
    \right\}.
\end{equation}

Then, the probability mass of $\psi$ near the boundary of $P$ is defined as
\begin{equation}
    \eta_R^{\mathbf{\hat P}}(\psi,P)
    :=
    \int_{\partial_R^{\mathbf{\hat P}}P}
    d^3\mathbf{k}\,
    |\braket{\mathbf{k}|\psi}|^2.
\end{equation}
This results in the following bound.

\begin{proposition} [General bound on the non-compatibility between classical and quantum momentum] \label{general_bound_non_compatibility_nonreativistic}

Given a system in state $\ket{\psi} \in \mathcal H$ and a partition $P = \{\mathbb R^3 \times Q_i\}_{i=1}^N \in \mathcal P_\Gamma$, the probability to see differences between the classical and quantum momentum is bounded, for any $R>0$, by

    \begin{equation} \label{equivalence_bound_nonrelativistic}
    \sum_i \sum_{j\neq i}
    \braket{\psi|
    \hat E_{\hat{\mathbf P}}(Q_i)
    \hat{\Pi}_{\mathbb R^3\times Q_j}
    \hat E_{\hat{\mathbf P}}(Q_i)
    |\psi}
    \le
    \eta_R^{\mathbf{\hat P}}(\psi,P)
    + \tau_{\mathbf{\hat P}}(R),
\end{equation}

where 

\begin{equation}
    \tau_{\mathbf{\hat P}}(R) = \operatorname{erfc}(R)
    +
    \frac{2R}{\sqrt{\pi}}e^{-R^2}.
\end{equation}
\end{proposition}

The proof is given in Appendix~\ref{general_bound_non_compatibility_nonreativistic_proof}.

\section{Proofs} \label{proofs}

\subsection{Proofs of Sect.~\ref{Minimal_requirements_for_Poincaré_group}}

\subsubsection{Proof of Proposition \ref{prop:operational-implies-microscopic-covariance}}\label{proof:operational-implies-microscopic-covariance}

\begin{proof}
By the density representation,
\begin{equation}
\frac{\operatorname{tr}\!\left(\hat\rho\,\hat\Pi_{B_\epsilon(\mathbf x',\mathbf p')}\right)}{\mu(B_\epsilon(\mathbf x',\mathbf p'))}
=
\frac{1}{\mu(B_\epsilon(\mathbf x',\mathbf p'))}\int_{B_\epsilon(\mathbf x',\mathbf p')}d\mu(\mathbf y,\mathbf q)\,\operatorname{tr}\!\left(\hat\rho\,\hat\pi_{(\mathbf y,\mathbf q)}\right).
\end{equation}
By weak continuity of the density, the right-hand side converges to
\begin{equation}
\operatorname{tr}\!\left(\hat\rho\,\hat\pi_{(\mathbf x',\mathbf p')}\right)
\end{equation}
as $\epsilon\to0$. Similarly,
\begin{equation}
\frac{\operatorname{tr}\!\left(\hat U^\dagger(g)\hat\rho\,\hat U(g)\hat\Pi_{B_\epsilon(\mathbf x,\mathbf p)}\right)}{\mu(B_\epsilon(\mathbf x,\mathbf p))}
\end{equation}
converges to
\begin{equation}
\operatorname{tr}\!\left(\hat U^\dagger(g)\hat\rho\,\hat U(g)\hat\pi_{(\mathbf x,\mathbf p)}\right)
=
\operatorname{tr}\!\left(\hat\rho\,\hat U(g)\hat\pi_{(\mathbf x,\mathbf p)}\hat U^\dagger(g)\right).
\end{equation}
Therefore Eq.~\eqref{operational_local_poincare_covariance} implies
\begin{equation}
\operatorname{tr}\!\left(\hat\rho\,\hat\pi_{(\mathbf x',\mathbf p')}\right)
=
\operatorname{tr}\!\left(\hat\rho\,\hat U(g)\hat\pi_{(\mathbf x,\mathbf p)}\hat U^\dagger(g)\right)
\end{equation}
for every state $\hat\rho\in S(\mathcal H)$. Hence
\begin{equation}
\hat\pi_{(\mathbf x',\mathbf p')}=\hat U(g)\hat\pi_{(\mathbf x,\mathbf p)}\hat U^\dagger(g),
\end{equation}
which proves the claim.
\end{proof}

\subsubsection{Proof of Proposition \ref{prop:rho-mu-depends-on-p-modulus}}\label{proof:rho-mu-depends-on-p-modulus}

\begin{proof}
Let $h=(\mathbf a,R)$ be an element of the Euclidean subgroup, acting on $\Gamma$ as in Eq.~\eqref{Euclidean_action}. Since this action preserves the Lebesgue measure $d\nu(\mathbf x,\mathbf p)=d^3\mathbf x\,d^3\mathbf p$, the density representation \eqref{Pi_C_density} gives, for every measurable region $C\subseteq\Gamma$,
\begin{equation}
\hat\Pi_{hC}
=
\int_C d\nu(\mathbf x,\mathbf p)\,\rho(h\cdot(\mathbf x,\mathbf p))\,\hat\pi_{h\cdot(\mathbf x,\mathbf p)}.
\end{equation}
On the other hand, Euclidean covariance of the integrated POVM, Eq.~\eqref{Euclidean_integrated_covariance}, together with microscopic covariance of the density, Eq.~\eqref{microscopic_covariance_coordinates}, gives
\begin{equation}
\hat\Pi_{hC}
=
\int_C d\nu(\mathbf x,\mathbf p)\,\rho(\mathbf x,\mathbf p)\,\hat\pi_{h\cdot(\mathbf x,\mathbf p)}.
\end{equation}
Therefore, for every measurable $C\subseteq\Gamma$,
\begin{equation}\label{rho_mu_weight_difference_integral}
\int_C d\nu(\mathbf x,\mathbf p)\,\Big[\rho(h\cdot(\mathbf x,\mathbf p))-\rho(\mathbf x,\mathbf p)\Big]\hat\pi_{h\cdot(\mathbf x,\mathbf p)}=0.
\end{equation}
Consider the measurable set on which the scalar factor in square brackets is positive. Restricting Eq.~\eqref{rho_mu_weight_difference_integral} to this set gives the integral of a positive operator-valued function with a nonnegative scalar weight. By the nondegeneracy assumption, this is possible only if that set has vanishing $\nu$-measure. The same argument applied to the set on which the scalar factor is negative shows that this set also has vanishing $\nu$-measure. Hence
\begin{equation}
\rho(h\cdot(\mathbf x,\mathbf p))=\rho(\mathbf x,\mathbf p)
\end{equation}
for $\nu$-almost every $(\mathbf x,\mathbf p)\in\Gamma$. Since $h=(\mathbf a,R)$ was arbitrary, this proves Eq.~\eqref{rho_mu_translation_rotation_invariance}.

Taking $R=\mathbb I$ shows that $\rho$ is independent of $\mathbf x$ up to $\nu$-null sets. Taking $\mathbf a=\mathbf0$ then shows that its remaining dependence on $\mathbf p$ is invariant under all rotations. Therefore it can depend on $\mathbf p$ only through $|\mathbf p|$, and there exists a measurable function $\rho:\mathbb R_+\to(0,\infty)$ such that Eq.~\eqref{rho_mu_mod_p} holds.
\end{proof}

\subsubsection{Proof of Example~\ref{ex:covariant-massless-povm}}\label{proof:covariant-povm-example}

Throughout this appendix we will work with the universal cover of the Poincar\'e group $\widetilde P^\uparrow_+=\mathbb R^{1,3}\rtimes SL(2,\mathbb C)$, whose elements we write as $(a,S)$. We use the bijection between forward null momenta and positive semidefinite rank-one Hermitian matrices, $k\mapsto K_k:=k^0\mathbb{\hat I}+\mathbf k\cdot\boldsymbol\sigma$ with $\det K_k=k^\mu k_\mu=0$. It relates the two group actions, as $S K_k S^\dagger=K_{\Lambda_S k}$. Since $K_k$ is positive semidefinite of rank one, it factorizes as $K_k=\kappa\kappa^\dagger$ with $\kappa\in\mathbb C^2$ unique up to a phase.

We now recall where the Wigner angle of the helicity representation comes from. The irreducible representation is induced from the little group of a fixed reference null momentum $k_R$: one fixes standard Lorentz transformations $L_k$ with $L_k k_R=k$, and to every pair $(S,k)$ associates the Wigner element $\tilde W(S,k):=\tilde L_{\Lambda_S k}^{-1}S\tilde L_k$, where $k\mapsto\tilde L_k\in SL(2,\mathbb C)$ is a fixed lift of $k\mapsto L_k$. The Wigner angle is usually defined by decomposing $\Lambda_{\tilde W(S,k)}$ into a rotation and a translation, but the spinor formalism extracts this angle directly. Fix a reference spinor $\kappa_R$ with $\kappa_R^{\phantom\dagger}\kappa_R^\dagger=K_{k_R}$ and define the spinor section $\kappa(k):=\tilde L_k\kappa_R$, so that $\kappa(k)\kappa(k)^\dagger=K_k$. Since $\tilde W(S,k)$ stabilizes $k_R$, $\tilde W\kappa_R$ equals $\kappa_R$ up to a unit-modulus phase. This phase identifies the rotation angle and we define the Wigner angle, mod $4\pi$, by
\begin{equation}\label{theta_W_normalization}
\tilde W(S,k)\,\kappa_R=e^{-\frac{i}{2}\theta_W(S,k)}\,\kappa_R,
\end{equation}
and the helicity representation is understood with this convention, $\hat U(S)\ket{\mathbf k}=e^{-ih\theta_W(S,k)}\ket{\Lambda_S\mathbf k}$. Finally, we write $[\kappa,\pi]:=\kappa_1\pi_2-\kappa_2\pi_1$ for the antisymmetric spinor product, which satisfies $[S\kappa,S\pi]=[\kappa,\pi]$ for all $S\in SL(2,\mathbb C)$.

\begin{lemma}[Relative Wigner phase]\label{lem:relative-wigner-phase}
For all $k,p \in V_0^+$ one has
\begin{equation}\label{spinor_modulus_identity}
\big|[\kappa(k),\kappa(p)]\big|^2=2\,k^\mu p_\mu,
\end{equation}
so that the relative Wigner phase
\begin{equation}\label{eta_definition}
e^{i\eta(k,p)}:=\frac{[\kappa(k),\kappa(p)]}{\big|[\kappa(k),\kappa(p)]\big|}
\end{equation}
is well defined on the set of pairs with $k^\mu p_\mu\neq0$. Moreover, for every $S\in SL(2,\mathbb C)$,
\begin{equation}\label{eta_cocycle}
\eta(\Lambda_S k,\Lambda_S p)=\eta(k,p)+\tfrac{1}{2}\theta_W(S,k)+\tfrac{1}{2}\theta_W(S,p)\quad \mathrm{mod}\ 2\pi.
\end{equation}
\end{lemma}

\begin{proof}
Writing $K_k=\kappa\kappa^\dagger$ in components gives $k^0\pm k^3=|\kappa_{1/2}|^2$ and $k^1-ik^2=\kappa_1\bar\kappa_2$, and similarly for $p$ with spinor $\pi$. A direct computation then yields $2k^\mu p_\mu=2k^0p^0-2\mathbf k\cdot\mathbf p=|\kappa_1|^2|\pi_2|^2+|\kappa_2|^2|\pi_1|^2-2\,\mathrm{Re}\big(\kappa_1\bar\kappa_2\bar\pi_1\pi_2\big)=\big|\kappa_1\pi_2-\kappa_2\pi_1\big|^2$, which is Eq.~\eqref{spinor_modulus_identity}. For the transformation law, the definition of $\tilde W(S,k)$ and the normalization \eqref{theta_W_normalization} give $S\kappa(k)=\tilde L_{\Lambda_S k}\tilde W(S,k)\kappa_R=e^{-\frac{i}{2}\theta_W(S,k)}\kappa(\Lambda_S k)$, hence, using the $SL(2,\mathbb C)$ invariance of the spinor product, $[\kappa(\Lambda_S k),\kappa(\Lambda_S p)]=e^{\frac{i}{2}\theta_W(S,k)+\frac{i}{2}\theta_W(S,p)}[S\kappa(k),S\kappa(p)]=e^{\frac{i}{2}\theta_W(S,k)+\frac{i}{2}\theta_W(S,p)}[\kappa(k),\kappa(p)]$. Dividing by the moduli, which are equal by Eq.~\eqref{spinor_modulus_identity} and the invariance of $k^\mu p_\mu$, gives Eq.~\eqref{eta_cocycle}.
\end{proof}

\begin{proof}[Proof of Example~\ref{ex:covariant-massless-povm}]
We verify the four conditions of Definition~\ref{def:regular-poincare-covariant-povm} and the normalization condition \eqref{identity}.

\emph{Nondegeneracy.} Since $N_{\mathbf q}>0$, the function $\phi$ is nonzero on a set of positive measure, so the generalised wave function $\mathbf k\mapsto e^{-\frac{i}{\hbar}\mathbf x\cdot\mathbf k}e^{-2ih\eta(k,p)}\phi(k^\mu p_\mu)$ is nonzero on a set of positive measure and $\hat\pi_{(\mathbf x,\mathbf p)}\neq0$.

\emph{Covariance at the origin.} Let $S\in SL(2,\mathbb C)$ and write $\Lambda:=\Lambda_S$ and $k':=\Lambda^{-1}k$. We compute
\begin{align}
\braket{\mathbf k|\hat U(S)|\mathbf 0,\mathbf p}&=e^{-ih\theta_W(S,k')}\braket{\mathbf k'|\mathbf 0,\mathbf p}=e^{-ih\theta_W(S,k')}e^{-2ih\eta(k',p)}\phi(k'^\mu p_\mu)\\
&=e^{-ih\theta_W(S,k')}e^{-2ih\eta(k,\Lambda p)}e^{ih\theta_W(S,k')}e^{ih\theta_W(S,p)}\phi\big(k^\mu(\Lambda p)_\mu\big)=e^{ih\theta_W(S,p)}\braket{\mathbf k|\mathbf 0,\Lambda\mathbf p}.
\end{align}
By inserting $\Lambda$, this directly ensures
\begin{equation}\label{origin_covariance_helicity}
\hat U(\Lambda)\,\hat\pi_{(\mathbf 0,\mathbf p)}\,\hat U^\dagger(\Lambda)=\hat\pi_{(\mathbf 0,\Lambda\mathbf p)}.
\end{equation}

\emph{Microscopic Poincar\'e covariance.} Let $g=(a,S)$ and $(\mathbf x,\mathbf p)\in\Gamma$ with $x=(0,\mathbf x)$ and $(\Lambda x+a)^0=0$, and write $(0,\mathbf x')=\Lambda x+a$. Since $\ket{\mathbf x,\mathbf p}=\hat U(x)\ket{\mathbf 0,\mathbf p}$, the group law gives $\hat U(g)\hat U(x)=\hat U(a+\Lambda x)\hat U(S)=\hat U\big((0,\mathbf x')\big)\hat U(S)$, and therefore, by Eq.~\eqref{origin_covariance_helicity}, $\hat U(g)\hat\pi_{(\mathbf x,\mathbf p)}\hat U^\dagger(g)=\hat\pi_{(\mathbf x',\Lambda\mathbf p)}$.

\emph{Euclidean covariance.} This follows directly from microscopic covariance and the invariance of $\rho(|\mathbf p|) d^3\mathbf x d^3 \mathbf p$ under Euclidean transformations.

\emph{Normalization.} For $\hat O:=\int_\Gamma d^3\mathbf x\,d^3\mathbf p\,\rho(|\mathbf p|)\,\hat\pi_{(\mathbf x,\mathbf p)}$, carrying out the $\mathbf x$-integration gives
\begin{equation}
\bra{\mathbf k}\hat O\ket{\mathbf k'}=(2\pi\hbar)^3\delta^3(\mathbf k-\mathbf k')\int d^3\mathbf p\,\rho(|\mathbf p|)\,|\phi(k^\mu p_\mu)|^2,
\end{equation}
where the helicity phases cancel on the diagonal $\mathbf k=\mathbf k'$. Defining $I(\mathbf k):=(2\pi\hbar)^3\int\frac{d^3\mathbf p}{2|\mathbf p|^4}|\phi(k^\mu p_\mu)|^2$, we claim $I(\mathbf k)=\frac{|\mathbf k|}{|\mathbf q|}N_{\mathbf q}$. Indeed, any forward null $k$ can be written as $k=\lambda Rq$ with $\lambda=|\mathbf k|/|\mathbf q|$ and $R$ a spatial rotation; substituting first $\mathbf p\to R\mathbf p$, under which $d^3\mathbf p$, $|\mathbf p|$ are invariant and $k^\mu(Rp)_\mu=\lambda\,q^\mu p_\mu$, and then $\mathbf p\to\mathbf p/\lambda$, under which $d^3\mathbf p/|\mathbf p|^4\to\lambda\,d^3\mathbf p/|\mathbf p|^4$ and $\phi(\lambda q^\mu p_\mu/\lambda)=\phi(q^\mu p_\mu)$, gives $I(\mathbf k)=\lambda I(\mathbf q)$. Therefore
\begin{equation}
\bra{\mathbf k}\hat O\ket{\mathbf k'}=\frac{|\mathbf q|}{N_{\mathbf q}}\,2I(\mathbf k)\,\delta^3(\mathbf k-\mathbf k')=2|\mathbf k|\,\delta^3(\mathbf k-\mathbf k')=\braket{\mathbf k|\mathbf k'},
\end{equation}
i.e., $\hat O=\hat{\mathbb I}$.
\end{proof}

\subsubsection{Proof of Proposition~\ref{prop:rigidity-shift}}\label{proof:rigidity-shift}

We first show that no measurable effect density of a regular covariant family is invariant under a nontrivial spacetime translation.

\begin{lemma}[No translation invariance of the density]\label{lem:no-translation-invariance}
Let $\hat\pi_{(\mathbf 0,\mathbf q)}\neq0$ be an effect density whose momentum kernel is given by a measurable function, and let $b\in\mathbb R^{1,3}$ satisfy $\hat U(b)\,\hat\pi_{(\mathbf 0,\mathbf q)}\,\hat U^\dagger(b)=\hat\pi_{(\mathbf 0,\mathbf q)}$. Then $b=0$.
\end{lemma}

\begin{proof}
Write $\pi(\mathbf k,\mathbf k'):=\braket{\mathbf k|\hat\pi_{(\mathbf 0,\mathbf q)}|\mathbf k'}$. The conjugated density has kernel $e^{\frac{i}{\hbar}b\cdot(k-k')}\pi(\mathbf k,\mathbf k')$, so the invariance condition reads $\big(e^{\frac{i}{\hbar}b\cdot(k-k')}-1\big)\pi(\mathbf k,\mathbf k')=0$ for almost every $(\mathbf k,\mathbf k')$. Hence $\pi$ vanishes almost everywhere outside the set $D_b:=\{(\mathbf k,\mathbf k'):b\cdot(k-k')\in2\pi\hbar\mathbb Z\}$, which is a measure zero set unless $b=0$.
\end{proof}

\begin{proof}[Proof of Proposition~\ref{prop:rigidity-shift}]
By nondegeneracy and translation covariance, $\hat\pi_{(\mathbf 0,\mathbf p)}\neq0$ and the kernel assumption holds for $\nu$-almost every $\mathbf p$; fix such a $\mathbf p$, chosen so that the same holds at the finitely many Lorentz-transformed momenta used below. We use microscopic covariance for two families of transformations. First, for a purely spatial translation $g=((0,\mathbf a),\mathbb I)$ the shift vanishes and Eq.~\eqref{microscopic_covariance_coordinates} gives
\begin{equation}\label{translation_generated_family}
\hat\pi_{(\mathbf x,\mathbf p)}=\hat U(\mathbf x)\,\hat\pi_{(\mathbf 0,\mathbf p)}\,\hat U(\mathbf x)^\dagger.
\end{equation}
Second, for a pure Lorentz transformation $g=(0,\Lambda)$ the condition $(\Lambda x+a)^0=0$ always holds at $\mathbf x=\mathbf 0$, and Eq.~\eqref{microscopic_covariance_coordinates} gives
\begin{equation}\label{origin_covariance_shifted}
\hat U(\Lambda)\,\hat\pi_{(\mathbf 0,\mathbf p)}\,\hat U^\dagger(\Lambda)=\hat\pi_{(\boldsymbol\Delta_\lambda(\Lambda,\mathbf p),\,\Lambda\mathbf p)}.
\end{equation}
Let now $\Lambda_1,\Lambda_2\in SO^+(1,3)$ and write $\boldsymbol\Delta_1:=\boldsymbol\Delta_\lambda(\Lambda_1,\mathbf p)$, $\boldsymbol\Delta_2':=\boldsymbol\Delta_\lambda(\Lambda_2,\Lambda_1\mathbf p)$, $\boldsymbol\Delta_{21}:=\boldsymbol\Delta_\lambda(\Lambda_2\Lambda_1,\mathbf p)$, and $\mathbf q:=\Lambda_2\Lambda_1\mathbf p$. We evaluate $\hat U(\Lambda_2\Lambda_1)\hat\pi_{(\mathbf 0,\mathbf p)}\hat U^\dagger(\Lambda_2\Lambda_1)$ in two ways. Directly, Eqs.~\eqref{origin_covariance_shifted} and \eqref{translation_generated_family} give $\hat U(\boldsymbol\Delta_{21})\hat\pi_{(\mathbf 0,\mathbf q)}\hat U(\boldsymbol\Delta_{21})^\dagger$. In two steps, using $\hat U(\Lambda_2\Lambda_1)=\hat U(\Lambda_2)\hat U(\Lambda_1)$, Eq.~\eqref{origin_covariance_shifted} for $\Lambda_1$, Eq.~\eqref{translation_generated_family}, the group law $\hat U(0,\Lambda_2)\,\hat U\big((0,\boldsymbol\Delta_1),\mathbb{I}\big)=\hat U\big(\Lambda_2(0,\boldsymbol\Delta_1),\mathbb I\big)\,\hat U(0,\Lambda_2)$, Eq.~\eqref{origin_covariance_shifted} for $\Lambda_2$ at momentum $\Lambda_1\mathbf p$, and the commutativity of translations, one finds $\hat U(c)\,\hat\pi_{(\mathbf 0,\mathbf q)}\,\hat U^\dagger(c)$ with $c:=\Lambda_2(0,\boldsymbol\Delta_1)+(0,\boldsymbol\Delta_2')$. Equating the two expressions yields
\begin{equation}\label{invariance_condition_d}
\hat U(d)\,\hat\pi_{(\mathbf 0,\mathbf q)}\,\hat U^\dagger(d)=\hat\pi_{(\mathbf 0,\mathbf q)},\qquad d:=\Lambda_2(0,\boldsymbol\Delta_1)+\big(0,\boldsymbol\Delta_2'-\boldsymbol\Delta_{21}\big).
\end{equation}
Since $\hat\pi_{(\mathbf 0,\mathbf q)}\neq0$, Lemma~\ref{lem:no-translation-invariance} gives $d=0$. The time component of $d$ receives a contribution only from the first term, so $\big(\Lambda_2(0,\boldsymbol\Delta_1)\big)^0=0$. As $\Lambda_2$ is arbitrary, we conclude $\boldsymbol\Delta_1=\boldsymbol\Delta_\lambda(\Lambda_1,\mathbf p)=\mathbf 0$. Since $\Lambda_1$ was arbitrary and the excluded set of $\mathbf p$ is $\nu$-null, the shift vanishes almost everywhere.
\end{proof}

\subsection{Proofs of Sect.~\ref{nogo_theorem}}\label{proof_nogo_theorem}

\subsubsection{Proof of Proposition~\ref{prop:detection-kernel-invariant-form}}\label{proof:detection-kernel-invariant-form}

\begin{proof}
Let $p=(|\mathbf p|,\mathbf p)$ and $k=(|\mathbf k|,\mathbf k)$ be the corresponding null four-momenta, and denote by $\mathrm{LG}(\mathbf p)$ the little group of $p$, i.e. the subgroup of Lorentz transformations $S$ such that $Sp=p$. Including the side-jump induced by helicity, microscopic Poincar\'e covariance gives
\begin{equation}
\hat U(S)\hat\pi_{(\mathbf0,\mathbf p)}\hat U^\dagger(S)=\hat\pi_{(\mathbf\Delta_\lambda(S,\mathbf p),\mathbf p)} = \hat U(\mathbf\Delta_\lambda(S,\mathbf p)) \hat\pi_{(\mathbf0,\mathbf p)} \hat U(\mathbf\Delta_\lambda(S,\mathbf p))^\dagger.
\end{equation}
Using the helicity $h$ transformation law of the generalized momentum eigenstates,
\begin{equation}
\hat U(S)\ket{\mathbf k}=e^{-ih\theta_W(S,\mathbf k)}\ket{S\mathbf k},
\end{equation}
we obtain
\begin{align}\label{kernel_little_group_invariance}
\mathcal K(\mathbf k,\mathbf p)
&=
\braket{\mathbf k|\hat\pi_{(\mathbf0,\mathbf p)}|\mathbf k}
=
\braket{\mathbf k | e^{i \mathbf k \cdot \mathbf\Delta_\lambda(S,\mathbf p)/\hbar} \hat\pi_{(\mathbf0,\mathbf p)}e^{-i \mathbf k \cdot \mathbf\Delta_\lambda(S,\mathbf p)/\hbar}|\mathbf k}\\
&= \braket{\mathbf k|\hat U(\mathbf\Delta_\lambda(S,\mathbf p))^\dagger \hat\pi_{(\mathbf0,\mathbf p)} \hat U(\mathbf\Delta_\lambda(S,\mathbf p))|\mathbf k} = \braket{\mathbf k|\hat U(S)^\dagger \hat\pi_{(\mathbf0,\mathbf p)} \hat U(S)|\mathbf k}\\
&= \braket{S\mathbf k|\hat\pi_{(\mathbf0,\mathbf p)}|S\mathbf k}
= 
\mathcal K(S\mathbf k,\mathbf p)
\end{align}
for every $S\in\mathrm{LG}(\mathbf p)$, where we used that phases appearing on both sides of the bra-ket cancel.

We now show explicitly why, for fixed $\mathbf p$, the kernel depends on $\mathbf k$ only through $k^\mu p_\mu/|\mathbf p|$. Let
\begin{equation}
\mathbf n_{\mathbf p}:=\frac{\mathbf p}{|\mathbf p|}.
\end{equation}
For a null momentum $k=(|\mathbf k|,\mathbf k)$, define the $\mathbf p$-adapted light-cone components
\begin{equation}
k_{\mathbf p}^-:=\frac{k^\mu p_\mu}{|\mathbf p|}=|\mathbf k|-\mathbf k\cdot\mathbf n_{\mathbf p},
\qquad
k_{\mathbf p}^+:=|\mathbf k|+\mathbf k\cdot\mathbf n_{\mathbf p},
\end{equation}
and the transverse component
\begin{equation}
\mathbf k_{\perp\mathbf p}:=\mathbf k-(\mathbf k\cdot\mathbf n_{\mathbf p})\mathbf n_{\mathbf p}.
\end{equation}
Since $k$ is null, these quantities satisfy
\begin{equation}
k_{\mathbf p}^+k_{\mathbf p}^-=|\mathbf k_{\perp\mathbf p}|^2.
\end{equation}
Thus, for $k_{\mathbf p}^->0$, the null momentum $k$ is completely determined by $k_{\mathbf p}^-$ and by the transverse vector $\mathbf k_{\perp\mathbf p}$.

The little group of $p=(|\mathbf p|,\mathbf p)$ contains null rotations. In the $\mathbf p$-adapted decomposition, a null rotation with transverse parameter $\mathbf b\perp\mathbf n_{\mathbf p}$ acts as
\begin{equation}
k_{\mathbf p}^-\mapsto k_{\mathbf p}^-,
\qquad
\mathbf k_{\perp\mathbf p}\mapsto \mathbf k_{\perp\mathbf p}+\mathbf b\,k_{\mathbf p}^-,
\qquad
k_{\mathbf p}^+\mapsto k_{\mathbf p}^+ +2\mathbf b\cdot\mathbf k_{\perp\mathbf p}+|\mathbf b|^2k_{\mathbf p}^-.
\end{equation}
This transformation leaves $p$ invariant and preserves the null condition. Let $k$ and $\ell$ be two future null momenta such that
\begin{equation}
k_{\mathbf p}^-=\ell_{\mathbf p}^-=:s>0.
\end{equation}
Choosing
\begin{equation}
\mathbf b:=\frac{\boldsymbol{\ell}_{\perp\mathbf p}-\mathbf k_{\perp\mathbf p}}{s}
\end{equation}
maps $\mathbf k_{\perp\mathbf p}$ to $\boldsymbol{\ell}_{\perp\mathbf p}$. Since the transformed momentum is null and has the same value of $k_{\mathbf p}^-=s$, its $k_{\mathbf p}^+$ component is fixed by the null condition and equals $\ell_{\mathbf p}^+$. Hence this null rotation maps $k$ to $\ell$.

Therefore any two future null momenta with the same value of $k_{\mathbf p}^-=k^\mu p_\mu/|\mathbf p|$ are related by an element of $\mathrm{LG}(\mathbf p)$, except for the collinear set $k^\mu p_\mu=0$, which has vanishing measure. Since the detection kernel is invariant under the little group, Eq.~\eqref{kernel_little_group_invariance}, it is constant on each level set of $k^\mu p_\mu/|\mathbf p|$. Hence, for fixed $\mathbf p$, there exists a nonnegative function $f_{\mathbf p}$ such that
\begin{equation}\label{kernel_fixed_p}
\mathcal K(\mathbf k,\mathbf p)
=
f_{\mathbf p}\!\left(\frac{k^\mu p_\mu}{|\mathbf p|}\right)
\end{equation}
for almost every $\mathbf k\in\mathbb R^3_*$.

It remains to show that the function can be chosen independently of $\mathbf p$. Fix an arbitrary reference momentum $\mathbf q\in\mathbb R^3_*$. For every $\mathbf p\in\mathbb R^3_*$, choose a Lorentz transformation $\Lambda_{\mathbf q\to\mathbf p}$ such that
\begin{equation}
\Lambda_{\mathbf q\to\mathbf p}q=p,
\qquad
q=(|\mathbf q|,\mathbf q).
\end{equation}
Microscopic covariance gives
\begin{equation}
\hat\pi_{(\mathbf \Delta_\lambda(\Lambda_{\mathbf q\to\mathbf p},\mathbf q),\mathbf p)}
=
\hat U(\Lambda_{\mathbf q\to\mathbf p})\hat\pi_{(\mathbf0,\mathbf q)}\hat U^\dagger(\Lambda_{\mathbf q\to\mathbf p}).
\end{equation}
As the translation phases cancel in the diagonal kernel, we can conclude
\begin{equation}
\mathcal K(\mathbf k,\mathbf p)
=
\mathcal K(\Lambda_{\mathbf q\to\mathbf p}^{-1}\mathbf k,\mathbf q),
\end{equation}
where $\Lambda_{\mathbf q\to\mathbf p}^{-1}\mathbf k$ denotes the spatial part of $\Lambda_{\mathbf q\to\mathbf p}^{-1}(|\mathbf k|,\mathbf k)$. Using Eq.~\eqref{kernel_fixed_p} with $\mathbf p$ replaced by the fixed momentum $\mathbf q$, the right-hand side depends only on
\begin{equation}
\frac{(\Lambda_{\mathbf q\to\mathbf p}^{-1}k)^\mu q_\mu}{|\mathbf q|}.
\end{equation}
By Lorentz invariance of the Minkowski product and by $\Lambda_{\mathbf q\to\mathbf p}q=p$, we have
\begin{equation}
(\Lambda_{\mathbf q\to\mathbf p}^{-1}k)^\mu q_\mu
=
k^\mu(\Lambda_{\mathbf q\to\mathbf p}q)_\mu
=
k^\mu p_\mu.
\end{equation}
Hence
\begin{equation}
\mathcal K(\mathbf k,\mathbf p)
=
f_{\mathbf q}\!\left(\frac{k^\mu p_\mu}{|\mathbf q|}\right).
\end{equation}
Since $\mathbf q$ is fixed, the constant factor $|\mathbf q|^{-1}$ can be absorbed into the definition of a single nonnegative function $\Phi$. Therefore,
\begin{equation}
\mathcal K(\mathbf k,\mathbf p)=\Phi(k^\mu p_\mu),
\end{equation}
that is,
\begin{equation}
\braket{\mathbf k|\hat\pi_{(\mathbf0,\mathbf p)}|\mathbf k}
=
\Phi\!\left(|\mathbf k|\,|\mathbf p|-\mathbf k\cdot\mathbf p\right),
\end{equation}
for almost every $\mathbf k,\mathbf p\in\mathbb R^3_*$.
\end{proof}

\subsubsection{Proof of Proposition~\ref{prop:diagonal-momentum-marginal}}\label{proof:diagonal-momentum-marginal}

\begin{proof}
Since the POVM is a regular Poincar\'e-covariant phase-space POVM, Proposition~\ref{prop:rho-mu-depends-on-p-modulus} gives $\rho(\mathbf x,\mathbf p)=\rho(|\mathbf p|)$. Let $T_{\mathbf x}$ denote the spatial translation that maps $(\mathbf0,\mathbf p)$ to $(\mathbf x,\mathbf p)$. By microscopic covariance,
\begin{equation}
\hat\pi_{(\mathbf x,\mathbf p)}=\hat U(T_{\mathbf x})\hat\pi_{(\mathbf0,\mathbf p)}\hat U^\dagger(T_{\mathbf x}).
\end{equation}
In the momentum representation, spatial translations act as
\begin{equation}
\hat U(T_{\mathbf x})\ket{\mathbf k}=e^{-\frac{i}{\hbar}\mathbf x\cdot\mathbf k}\ket{\mathbf k}.
\end{equation}
Therefore,
\begin{equation}
\braket{\mathbf k|\hat\pi_{(\mathbf x,\mathbf p)}|\mathbf k'}
=
e^{-\frac{i}{\hbar}\mathbf x\cdot(\mathbf k-\mathbf k')}
\braket{\mathbf k|\hat\pi_{(\mathbf0,\mathbf p)}|\mathbf k'}.
\end{equation}
Using the density representation for $\mathbb R^3\times Q$, we obtain
\begin{align}
\braket{\mathbf k|\hat\Pi_{\mathbb R^3\times Q}|\mathbf k'}
&=
\int_{\mathbb R^3}d^3\mathbf x\int_Qd^3\mathbf p\,\rho(|\mathbf p|)\braket{\mathbf k|\hat\pi_{(\mathbf x,\mathbf p)}|\mathbf k'} \nonumber\\
&=
\int_{\mathbb R^3}d^3\mathbf x\,e^{-\frac{i}{\hbar}\mathbf x\cdot(\mathbf k-\mathbf k')}
\int_Qd^3\mathbf p\,\rho(|\mathbf p|)\braket{\mathbf k|\hat\pi_{(\mathbf0,\mathbf p)}|\mathbf k'} \nonumber\\
&=
(2\pi\hbar)^3\delta^3(\mathbf k-\mathbf k')\int_Qd^3\mathbf p\,\rho(|\mathbf p|)\braket{\mathbf k|\hat\pi_{(\mathbf0,\mathbf p)}|\mathbf k'}.
\end{align}
Since the expression is multiplied by $\delta^3(\mathbf k-\mathbf k')$, only the diagonal part of the kernel contributes. Using Proposition~\ref{prop:detection-kernel-invariant-form}, we get
\begin{equation}
\braket{\mathbf k|\hat\Pi_{\mathbb R^3\times Q}|\mathbf k'}
=
(2\pi\hbar)^3\delta^3(\mathbf k-\mathbf k')
\int_Qd^3\mathbf p\,\rho(|\mathbf p|)\Phi(k^\mu p_\mu).
\end{equation}
Comparing this expression with the standard diagonal form
\begin{equation}
\braket{\mathbf k|\hat\Pi_{\mathbb R^3\times Q}|\mathbf k'}=2|\mathbf k|\,\delta^3(\mathbf k-\mathbf k')\,\Pi_{\mathbb R^3\times Q}(\mathbf k)
\end{equation}
gives Eq.~\eqref{Pi_Q_kernel}. The spectral form \eqref{Pi_Q_spectral_form} follows immediately from the normalization $\braket{\mathbf k|\mathbf k'}=2|\mathbf k|\delta^3(\mathbf k-\mathbf k')$.
\end{proof}

\subsubsection{Proof of Proposition~\ref{prop:angular-response-area-bounds}}\label{proof:angular-response-area-bounds}

\begin{proof}
By Proposition~\ref{prop:diagonal-momentum-marginal}, applied to the momentum cone $Q_\Omega$ of Definition~\ref{def:angular-region-momentum-cone}, the corresponding multiplier is
\begin{equation}
\Pi_{\mathbb R^3\times Q_\Omega}(\mathbf k)=\frac{(2\pi\hbar)^3}{2|\mathbf k|}\int_{Q_\Omega}d^3\mathbf p\,\rho(|\mathbf p|)\Phi(k^\mu p_\mu).
\end{equation}
Writing $\mathbf p=r\mathbf n$, with $r>0$ and $\mathbf n\in S^2$, gives
\begin{equation}\label{Pi_QOmega_before_J}
\Pi_{\mathbb R^3\times Q_\Omega}(\mathbf k)=\frac{(2\pi\hbar)^3}{2|\mathbf k|}\int_0^\infty dr\,r^2\rho(r)\int_\Omega d\Omega(\mathbf n)\,\Phi\!\left(|\mathbf k|r(1-\mathbf n_{\mathbf k}\cdot\mathbf n)\right).
\end{equation}
Applying the same formula to $Q=\mathbb R^3_*$ gives $\hat\Pi_\Gamma=\hat{\mathbb I}$, hence
\begin{equation}\label{rho_Phi_normalization}
\frac{(2\pi\hbar)^3}{2|\mathbf k|}\int_0^\infty dr\,r^2\rho(r)\int_{S^2}d\Omega(\mathbf n)\,\Phi\!\left(|\mathbf k|r(1-\mathbf n_0\cdot\mathbf n)\right)=1,
\end{equation}
where $\mathbf n_0\in S^2$ is arbitrary. Define
\begin{equation}
J(s):=(2\pi\hbar)^3\int_0^\infty dr\,r^2\rho(r)\Phi(sr).
\end{equation}
Then Eq.~\eqref{rho_Phi_normalization} becomes
\begin{equation}
\frac{1}{2|\mathbf k|}\int_{S^2}d\Omega(\mathbf n)\,J\!\left(|\mathbf k|(1-\mathbf n_0\cdot\mathbf n)\right)=1.
\end{equation}
Using spherical coordinates with $u=\mathbf n_0\cdot\mathbf n$, this reads
\begin{equation}
\frac{\pi}{|\mathbf k|}\int_{-1}^{1}du\,J\!\left(|\mathbf k|(1-u)\right)=1.
\end{equation}
With the change of variables $s=|\mathbf k|(1-u)$, we obtain
\begin{equation}\label{J_integral_identity}
\frac{\pi}{|\mathbf k|^2}\int_0^{2|\mathbf k|}ds\,J(s)=1.
\end{equation}
Since this holds for every $|\mathbf k|>0$, differentiating Eq.~\eqref{J_integral_identity} gives
\begin{equation}\label{J_s}
J(s)=\frac{s}{2\pi}
\end{equation}
for almost every $s>0$. Substituting this identity into Eq.~\eqref{Pi_QOmega_before_J} yields
\begin{align}
\Pi_{\mathbb R^3\times Q_\Omega}(\mathbf k)
&=
\frac{1}{2|\mathbf k|}\int_\Omega d\Omega(\mathbf n)\,J\!\left(|\mathbf k|(1-\mathbf n_{\mathbf k}\cdot\mathbf n)\right) \nonumber\\
&=
\frac{1}{4\pi}\int_\Omega d\Omega(\mathbf n)\,\left(1-\mathbf n_{\mathbf k}\cdot\mathbf n\right),
\end{align}
which proves Eq.~\eqref{Pi_QOmega_angular_response}.

It remains to prove the area bounds. Fix $\mathbf k\in\mathbb R^3_*$ and write $\mathbf n':=\mathbf n_{\mathbf k}$. Define
\begin{equation}
I_{\mathbf n'}(\Omega):=\int_\Omega d\Omega(\mathbf n)\,\left(1-\mathbf n\cdot\mathbf n'\right).
\end{equation}
Then $\Pi_{\mathbb R^3\times Q_\Omega}(\mathbf k)=I_{\mathbf n'}(\Omega)/(4\pi)$. We now bound $I_{\mathbf n'}(\Omega)$ only in terms of the normalized area $s:=\sigma(\Omega)$.

Choose spherical coordinates around the axis $\mathbf n'$, so that $\mathbf n\cdot\mathbf n'=\cos\theta$. For each polar angle $\theta$, define the angular length of the section of $\Omega$ at latitude $\theta$ by
\begin{equation}
L_\Omega(\theta):=\sin\theta\int_0^{2\pi}d\varphi\,\chi_\Omega(\theta,\varphi).
\end{equation}
Then
\begin{equation}
\operatorname{Area}(\Omega)=\int_0^\pi d\theta\,L_\Omega(\theta)=4\pi s,
\end{equation}
with
\begin{equation}
0\leq L_\Omega(\theta)\leq 2\pi\sin\theta,
\end{equation}
and
\begin{equation}
I_{\mathbf n'}(\Omega)=\int_0^\pi d\theta\,(1-\cos\theta)L_\Omega(\theta).
\end{equation}
The weight $1-\cos\theta$ is monotone increasing in $\theta$. Therefore, under the fixed-area constraint $\operatorname{Area}(\Omega)=\int_0^\pi d\theta\,L_\Omega(\theta)=4\pi s$, the integral is minimized by placing as much angular length as possible near $\theta=0$, and maximized by placing as much angular length as possible near $\theta=\pi$.

For the lower bound, let $R_s^+$ be the spherical cap centered at $\mathbf n'$ with normalized area $s$, namely
\begin{equation}
R_s^+:=\{(\theta,\varphi):0\leq\theta\leq\theta_0\},
\end{equation}
where $\theta_0$ is determined by
\begin{equation}
4\pi s=\operatorname{Area}(R_s^+)=2\pi\int_0^{\theta_0}d\theta\,\sin\theta=2\pi(1-\cos\theta_0).
\end{equation}
Hence $\cos\theta_0=1-2s$. Therefore,
\begin{align}
I_{\mathbf n'}(\Omega)
&\geq I_{\mathbf n'}(R_s^+) \nonumber\\
&=2\pi\int_0^{\theta_0}d\theta\,\sin\theta\,(1-\cos\theta) \nonumber\\
&=2\pi\int_{\cos\theta_0}^{1}du\,(1-u)
=4\pi s^2.
\end{align}
For the upper bound, the extremal region is the cap $R_s^-$ centered at $-\mathbf n'$, with the same normalized area $s$. The analogous computation gives
\begin{equation}
I_{\mathbf n'}(\Omega)\leq I_{\mathbf n'}(R_s^-)=4\pi(2s-s^2).
\end{equation}
Dividing by $4\pi$, we obtain
\begin{equation}
s^2\leq\Pi_{\mathbb R^3\times Q_\Omega}(\mathbf k)\leq 2s-s^2=1-(1-s)^2.
\end{equation}
Since $s=\sigma(\Omega)$, this proves Eq.~\eqref{Pi_QOmega_area_bounds}.
\end{proof}

\subsubsection{Proof of Proposition~\ref{prop:weak-repeatability-almost-projective}}\label{proof:weak-repeatability-almost-projective}

\begin{proof}
Since $\mathbb R^3\times Q\triangleleft\mathcal P_{\rm CG}$, there exists a coarse partition containing $\mathbb R^3\times Q$. By closure under coarsening, the dichotomic partition
\begin{equation}
P_Q=
\left\{
\mathbb R^3\times Q,\,
\mathbb R^3\times Q^{\rm c}
\right\}
\end{equation}
also belongs to $\mathcal P_{\rm CG}$.
Set
\begin{equation}
m:=\operatorname*{ess\,inf}_{\mathbf k}
\Pi_{\mathbb R^3\times Q}(\mathbf k),
\qquad
M:=\operatorname*{ess\,sup}_{\mathbf k}
\Pi_{\mathbb R^3\times Q}(\mathbf k).
\end{equation}
By Proposition~\ref{prop:diagonal-momentum-marginal},
\begin{equation}
m\hat{\mathbb I}\leq\hat\Pi_{\mathbb R^3\times Q}\leq M\hat{\mathbb I},
\end{equation}
and hence
\begin{equation}
\hat\Pi_{\mathbb R^3\times Q}\geq m\hat{\mathbb I},
\qquad
\hat{\mathbb I}- \hat\Pi_{\mathbb R^3\times Q}\geq(1-M)\hat{\mathbb I}.
\end{equation}

Let
$\mathcal I^{P_Q}=\{\mathcal I_Q,\mathcal I_{Q^{\rm c}}\}
\in\mathfrak I(P_Q)$
be any instrument implementing the dichotomic POVM
$\{\hat\Pi_{\mathbb R^3\times Q},\hat{\mathbb I}-\hat\Pi_{\mathbb R^3\times Q}\}$, and let $\hat\rho$ be any normalized state. Its probability of switching outcomes under an immediate repetition is
\begin{align}
\operatorname{tr}\!\left[
(\hat{\mathbb I}-\hat\Pi_{\mathbb R^3\times Q})\mathcal I_Q(\hat\rho)
\right]
+
\operatorname{tr}\!\left[
\hat\Pi_{\mathbb R^3\times Q}\,\mathcal I_{Q^{\rm c}}(\hat\rho)
\right]
\nonumber
&\geq
(1-M)\operatorname{tr}\!\left[\mathcal I_Q(\hat\rho)\right]
+
m\operatorname{tr}\!\left[\mathcal I_{Q^{\rm c}}(\hat\rho)\right]
\nonumber\\
&\geq
\min\{m,1-M\}.
\end{align}
We used that, for $\hat A \ge \hat B \ge 0$ and $\hat C \ge 0$, $\mathrm{tr}(\hat A \hat C) \ge \mathrm{tr}(\hat B \hat C)$, and that $\mathrm{tr} \left[\mathcal I_Q(\hat \rho)\right] + \mathrm{tr} \left[\mathcal I_{Q^c}(\hat \rho) \right]= 1$, as non-selective measurements are trace-preserving.
This bound holds for every compatible instrument and every normalized state. Therefore,
\begin{equation}
\inf_{\mathcal I^{P_Q}\in\mathfrak I(P_Q)}
\inf_{\hat\rho\in S(\mathcal H)}
\operatorname{tr}\!\left[
(\hat{\mathbb I}-\hat\Pi_{\mathbb R^3\times Q})\mathcal I_Q(\hat\rho)
\right]
+
\operatorname{tr}\!\left[
\hat\Pi_{\mathbb R^3\times Q}\,\mathcal I_{Q^{\rm c}}(\hat\rho)
\right]
\geq
\min\{m,1-M\}.
\end{equation}
Weak non-invasive repeatability requires the left-hand side to be smaller than $\epsilon$. Consequently,
\begin{equation}
m<\epsilon
\qquad\text{or}\qquad
1-M<\epsilon,
\end{equation}
which is precisely Eq.~\eqref{projective}.
\end{proof}

\subsubsection{Proof of Proposition~\ref{prop:weak-momentum-compatibility-quasi-localization}}\label{proof:weak-momentum-compatibility-quasi-localization}

\begin{proof}
We can proceed similarly as for the proof of Proposition~\ref{prop:weak-repeatability-almost-projective}. By closure under coarsening, the dichotomic partition
\begin{equation}
P_Q=
\left\{
\mathbb R^3\times Q,\,
\mathbb R^3\times Q^{\rm c}
\right\}
\end{equation} 
belongs to $\mathcal P_{\rm CG}$.
Set
\begin{equation}
\hat A:=\hat\Pi_{\mathbb R^3\times Q},
\qquad
\hat E:=\hat E_{\hat{\mathbf P}}(Q),
\qquad
\hat E^{\rm c}:=\hat{\mathbb I}-\hat E,
\end{equation}
and define
\begin{equation}
M_Q:=
\operatorname*{ess\,sup}_{\mathbf k\in Q}
\Pi_{\mathbb R^3\times Q}(\mathbf k),
\qquad
m_{Q^{\rm c}}:=
\operatorname*{ess\,inf}_{\mathbf k\in Q^{\rm c}}
\Pi_{\mathbb R^3\times Q}(\mathbf k).
\end{equation}
By Proposition~\ref{prop:diagonal-momentum-marginal}, $\hat A$ and $\hat E$ are diagonal in the momentum representation and satisfy
\begin{equation}
\hat E(\hat{\mathbb I}-\hat A)\hat E
\geq
(1-M_Q)\hat E,
\qquad
\hat E^{\rm c}\hat A\hat E^{\rm c}
\geq
m_{Q^{\rm c}}\hat E^{\rm c}.
\end{equation}

For any normalized state $\hat\rho$, the probability that the momentum PVM and the subsequent momentum marginal of the phase-space POVM give different outcomes is
\begin{equation}
\operatorname{tr}\!\left[
(\hat{\mathbb I}-\hat A) \hat E \hat\rho \hat E
\right]
+
\operatorname{tr}\!\left[
\hat A \hat E^c\hat\rho\hat E^c
\right]
\nonumber
\geq
(1-M_Q)\operatorname{tr}(\hat E\hat\rho)
+
m_{Q^{\rm c}}\operatorname{tr}(\hat E^{\rm c}\hat\rho)
\nonumber \geq
\min\{1-M_Q,m_{Q^{\rm c}}\}.
\end{equation}

Taking the infimum over normalized states and applying weak classical--quantum momentum compatibility gives
\begin{equation}
\min\{1-M_Q,m_{Q^{\rm c}}\}<\epsilon',
\end{equation}
and implies Eq.~\eqref{compatibility_double}.
\end{proof}

\subsubsection{Proof of Theorem~\ref{thm:no-go-constraints}}\label{proof:no-go-constraints}

\begin{proof}
By minimal directional information gain, Definition~\ref{def:minimal-directional-information-gain}, there exists a measurable angular region $\Omega\subseteq S^2$ such that
\begin{equation}\label{proof_delta_sigma}
\frac{1}{2}(1-\delta)\leq\sigma(\Omega)\leq\frac{1}{2},
\end{equation}
and
\begin{equation}
\left\{\mathbb R^3\times Q_\Omega,\Gamma\setminus(\mathbb R^3\times Q_\Omega)\right\}\in\mathcal P_{\rm CG}.
\end{equation}
Set
\begin{equation}
s:=\sigma(\Omega).
\end{equation}
Then
\begin{equation}\label{s_interval}
\frac{1}{2}(1-\delta)\leq s\leq\frac{1}{2}.
\end{equation}
Since $\mathbb R^3\times Q_\Omega$ is a cell of a coarse partition, we have $\mathbb R^3\times Q_\Omega\triangleleft\mathcal P_{\rm CG}$. By Proposition~\ref{prop:angular-response-area-bounds}, the diagonal multiplier of $\hat\Pi_{\mathbb R^3\times Q_\Omega}$ satisfies
\begin{equation}\label{area_bounds_s}
s^2\leq \Pi_{\mathbb R^3\times Q_\Omega}(\mathbf k)\leq 1-(1-s)^2
\end{equation}
for almost every $\mathbf k\in\mathbb R^3_*$.

Assume first that the scheme satisfies weak non-invasive repeatability with tolerance $\epsilon$. By Proposition~\ref{prop:weak-repeatability-almost-projective}, applied to $Q=Q_\Omega$, one has
\begin{equation}\label{repeatability_projective_proof}
\operatorname*{ess\,inf}_{\mathbf k\in\mathbb R^3_*}\Pi_{\mathbb R^3\times Q_\Omega}(\mathbf k)<\epsilon
\quad\text{or}\quad
\operatorname*{ess\,sup}_{\mathbf k\in\mathbb R^3_*}\Pi_{\mathbb R^3\times Q_\Omega}(\mathbf k)>1-\epsilon.
\end{equation}
Combining this with Eq.~\eqref{area_bounds_s}, we obtain
\begin{equation}
s^2<\epsilon
\quad\text{or}\quad
1-(1-s)^2>1-\epsilon.
\end{equation}
Equivalently,
\begin{equation}\label{s_repeatability_alternative}
s<\sqrt{\epsilon}
\quad\text{or}\quad
s>1-\sqrt{\epsilon}.
\end{equation}
For the range of tolerances considered in the theorem, $\sqrt{\epsilon} < 1/2$. Since minimal directional information gain gives $s\leq1/2$, the second alternative in Eq.~\eqref{s_repeatability_alternative} is impossible. Hence $s<\sqrt{\epsilon}$. Combining this with Eq.~\eqref{s_interval}, namely $s\geq(1-\delta)/2$, gives
\begin{equation}
\frac{1}{2}(1-\delta)<\sqrt{\epsilon}.
\end{equation}
Equivalently,
\begin{equation}\label{proof_repeatability_constraint}
\delta+2\sqrt{\epsilon}>1.
\end{equation}
This proves Eq.~\eqref{no_go_repeatability_constraint}. In particular, if $\epsilon<1/6$ and $\delta<1/6$, then
\begin{equation}
\delta+2\sqrt{\epsilon}
<
\frac{1}{6}+2\sqrt{\frac{1}{6}}
<1,
\end{equation}
which contradicts Eq.~\eqref{proof_repeatability_constraint}. Therefore weak non-invasive repeatability with $\epsilon<1/6$ is incompatible with minimal directional information gain with $\delta<1/6$. The exact same proof applies to a scheme that satisfies weak classical--quantum momentum compatibility with tolerance $\epsilon'$, as $\operatorname*{ess\,inf}_{\mathbf k\in\mathbb R^3_*}\Pi_{\mathbb R^3\times Q_\Omega}(\mathbf k) \le \operatorname*{ess\,inf}_{\mathbf k\in Q_\Omega^c}\Pi_{\mathbb R^3\times Q_\Omega}(\mathbf k)$ and $\operatorname*{ess\,sup}_{\mathbf k\in\mathbb R^3_*}\Pi_{\mathbb R^3\times Q_\Omega}(\mathbf k) \ge \operatorname*{ess\,sup}_{\mathbf k\in Q_\Omega}\Pi_{\mathbb R^3\times Q_\Omega}(\mathbf k)$.
\end{proof}

\subsection{Proofs of Sect.~\ref{sec:multiple-helicities}}\label{proof_multiple_helicities}

\subsubsection{Proof of Proposition~\ref{prop:helicity-structure-kernel}}\label{proof:helicity-structure-kernel}

\begin{proof}

\emph{Statement 1.} This holds just as in the irreducible case as Euclidean transformations are untouched by helicity.

\emph{Statement 2.} We use the spinor formalism and the conventions of Appendix~\ref{proof:covariant-povm-example}. Fix $\lambda\in\Lambda$ and $\mathbf p\in\mathbb R^3_*$, and let $\mathbf k\in\mathbb R^3_*$ with $k^\mu p_\mu\neq0$, which excludes a null set. By Eq.~\eqref{spinor_modulus_identity}, $[\kappa(k),\kappa(p)]\neq0$, so the spinors $\kappa(k)$ and $\kappa(p)$ form a basis of $\mathbb C^2$. For $\theta\in[0,4\pi)$ define $S_\theta\in GL(2,\mathbb C)$ by
\begin{equation}\label{S_theta_definition}
S_\theta\,\kappa(k)=e^{-\frac{i}{2}\theta}\kappa(k),\qquad S_\theta\,\kappa(p)=e^{\frac{i}{2}\theta}\kappa(p).
\end{equation}
Then $\det S_\theta=1$, so $S_\theta\in SL(2,\mathbb C)$, and $S_\theta K_kS_\theta^\dagger=S_\theta\kappa(k)\kappa(k)^\dagger S_\theta^\dagger=K_k$, and likewise $S_\theta K_pS_\theta^\dagger=K_p$. Hence, $S_\theta$ lies in the lift of the little group $\mathrm{LG}(\mathbf p)$. Its Wigner angle at $k$ follows from Eq.~\eqref{theta_W_normalization}: since $\Lambda_{S_\theta}k=k$, we have $\tilde W(S_\theta,k)\kappa_R=\tilde L_k^{-1}S_\theta\tilde L_k\kappa_R=\tilde L_k^{-1}S_\theta\kappa(k)=e^{-\frac{i}{2}\theta}\kappa_R$, and
\begin{equation}\label{theta_W_S_theta}
\theta_W(S_\theta,k)=\theta\quad\mathrm{mod}\ 4\pi .
\end{equation}
Now apply microscopic covariance to the point $(\mathbf 0,\mathbf p,\lambda)$ and the transformation $g=(0,\Lambda_{S_\theta})$. Since $\Lambda_{S_\theta}$ fixes $p$, Eq.~\eqref{g_x_p_prime} gives $g\cdot(\mathbf 0,\mathbf p,\lambda)=(\boldsymbol\Delta_\lambda(\Lambda_{S_\theta},\mathbf p),\mathbf p,\lambda)$, and therefore, exactly as in Eq.~\eqref{kernel_little_group_invariance},
\begin{equation}
\hat U(S_\theta)\hat\pi_{(\mathbf 0,\mathbf p,\lambda)}\hat U^\dagger(S_\theta)=\hat U(\boldsymbol\Delta_\lambda)\hat\pi_{(\mathbf 0,\mathbf p,\lambda)}\hat U^\dagger(\boldsymbol\Delta_\lambda).
\end{equation}
Taking matrix elements between $\ket{\mathbf k,h,\alpha}$ and $\ket{\mathbf k,h',\alpha'}$ with the same momentum $\mathbf k$ on both sides, the translation phases cancel, because they only depend on $\mathbf k$. Thus, we obtain
\begin{equation}\label{cross_helicity_phase}
\braket{\mathbf k,h,\alpha|\hat\pi_{(\mathbf 0,\mathbf p,\lambda)}|\mathbf k,h',\alpha'}=\braket{\mathbf k,h,\alpha|\hat U^\dagger(S_\theta)\hat\pi_{(\mathbf 0,\mathbf p,\lambda)}\hat U(S_\theta)|\mathbf k,h',\alpha'}=e^{i(h-h')\theta}\braket{\mathbf k,h,\alpha|\hat\pi_{(\mathbf 0,\mathbf p,\lambda)}|\mathbf k,h',\alpha'}.
\end{equation}
Since this holds for every $\theta\in[0,4\pi)$, the matrix element vanishes and this proves Eq.~\eqref{kernel_helicity_diagonal}.

\emph{Statement 3.} For $h=h'$ and an arbitrary $S$ in the lift of $\mathrm{LG}(\mathbf p)$, the same computation as in Eq.~\eqref{cross_helicity_phase}, now with $S\mathbf k\neq\mathbf k$ in general, gives
\begin{equation}
\braket{\mathbf k,h,\alpha|\hat\pi_{(\mathbf 0,\mathbf p,\lambda)}|\mathbf k,h,\alpha'}=\braket{S\mathbf k,h,\alpha|\hat\pi_{(\mathbf 0,\mathbf p,\lambda)}|S\mathbf k,h,\alpha'}.
\end{equation}
Thus each entry $(\alpha,\alpha')$ of the $n_h\times n_h$ block is a little-group invariant function of $\mathbf k$, and the null-rotation argument in the proof of Proposition~\ref{prop:detection-kernel-invariant-form} shows that it is a function of $k^\mu p_\mu/|\mathbf p|$ alone. The transport from a reference momentum $\mathbf q$ to $\mathbf p$ by $\hat U(\Lambda_{\mathbf q\to\mathbf p})$, carried out in the same proof, produces only Wigner phases that cancel within a helicity sector, so that the dependence on $\mathbf p$ enters through $k^\mu p_\mu$ only, which proves Eq.~\eqref{kernel_multiplicity_block} with a measurable matrix-valued function $\Phi^{(h)}_\lambda$. Finally, $\hat\pi_{(\mathbf 0,\mathbf p,\lambda)}\geq0$ implies that, for almost every $\mathbf k$, the matrix $[\braket{\mathbf k,h,\alpha|\hat\pi_{(\mathbf 0,\mathbf p,\lambda)}|\mathbf k,h,\alpha'}]_{\alpha\alpha'}$ is positive semidefinite.
\end{proof}

\subsubsection{Proof of Proposition~\ref{prop:helicity-blind-marginal}}\label{proof:helicity-blind-marginal}

\begin{proof}
Repeating the computation in the proof of Proposition~\ref{prop:diagonal-momentum-marginal}, the integration over $\mathbf x\in\mathbb R^3$ produces the factor $(2\pi\hbar)^3\delta^3(\mathbf k-\mathbf k')$, so that
\begin{equation}
\braket{\mathbf k,h,\alpha|\hat\Pi_{\mathbb R^3\times Q\times\Lambda}|\mathbf k',h',\alpha'}=(2\pi\hbar)^3\delta^3(\mathbf k-\mathbf k')\int_Qd^3\mathbf p\int_\Lambda d\kappa(\lambda)\,\rho_\lambda(|\mathbf p|)\braket{\mathbf k,h,\alpha|\hat\pi_{(\mathbf 0,\mathbf p,\lambda)}|\mathbf k,h',\alpha'},
\end{equation}
where statement 1 of Proposition~\ref{prop:helicity-structure-kernel} was used for the scalar density. By statement 2, the integrand vanishes for $h\neq h'$, and by statement 3 it equals $[\Phi^{(h)}_\lambda(k^\mu p_\mu)]_{\alpha\alpha'}$ for $h=h'$. Comparing with the normalization $\braket{\mathbf k,h,\alpha|\mathbf k',h',\alpha'}=2|\mathbf k|\delta^3(\mathbf k-\mathbf k')\delta_{hh'}\delta_{\alpha\alpha'}$ gives Eqs.~\eqref{Pi_Q_reducible} and \eqref{Pi_Q_reducible_multiplier}.

For the cone $Q_\Omega$, write $\mathbf p=r\mathbf n$ with $r>0$ and $\mathbf n\in S^2$, so that $k^\mu p_\mu=|\mathbf k|r(1-\mathbf n_{\mathbf k}\cdot\mathbf n)$, and define the matrix-valued function
\begin{equation}
J^{(h)}(s):=(2\pi\hbar)^3\int_\Lambda d\kappa(\lambda)\int_0^\infty dr\,r^2\rho_\lambda(r)\,\Phi^{(h)}_\lambda(sr),
\end{equation}
which is positive semidefinite for every $s>0$. Then
\begin{equation}\label{Pi_QOmega_reducible_J}
\Pi^{(h)}_{Q_\Omega}(\mathbf k)=\frac{1}{2|\mathbf k|}\int_\Omega d\Omega(\mathbf n)\,J^{(h)}\!\left(|\mathbf k|(1-\mathbf n_{\mathbf k}\cdot\mathbf n)\right).
\end{equation}
The normalization $\hat\Pi_{\Gamma_\Lambda}=\hat{\mathbb I}$, applied to Eq.~\eqref{Pi_Q_reducible} with $Q=\mathbb R^3_*$, gives $\Pi^{(h)}_{\mathbb R^3_*}(\mathbf k)=\mathbb I_{n_h}$ for every $h$ and almost every $\mathbf k$. Evaluating Eq.~\eqref{Pi_QOmega_reducible_J} for $\Omega=S^2$ in spherical coordinates around $\mathbf n_{\mathbf k}$ and substituting $s=|\mathbf k|(1-u)$, exactly as in Eqs.~\eqref{rho_Phi_normalization}--\eqref{J_integral_identity}, yields the matrix identity
\begin{equation}
\frac{\pi}{|\mathbf k|^2}\int_0^{2|\mathbf k|}ds\,J^{(h)}(s)=\mathbb I_{n_h}
\end{equation}
for every $|\mathbf k|>0$. Differentiating with respect to $|\mathbf k|$ entrywise gives $J^{(h)}(s)=\frac{s}{2\pi}\mathbb I_{n_h}$ for almost every $s>0$. Inserting this into Eq.~\eqref{Pi_QOmega_reducible_J} gives Eq.~\eqref{Pi_QOmega_reducible}. The area bounds \eqref{Pi_QOmega_area_bounds} for the scalar $\Pi_{\mathbb R^3\times Q_\Omega}(\mathbf k)$ were proven in Appendix~\ref{proof:angular-response-area-bounds} and depend only on its explicit form.
\end{proof}

\subsection{Proofs of Sect.~\ref{massive}}

\subsubsection{Proof of Proposition \ref{inf_repeatability_bound_nonrelativistic}} \label{inf_repeatability_bound_nonrelativistic_proof}

\begin{proof}
Let $ C \in P$ be the cell containing the ball $B_R^\sigma(\mathbf{x}_0,\mathbf{p}_0)$, and choose the trial state
\begin{equation}
    \ket{\psi}=\ket{\mathbf{x}_0,\mathbf{p}_0},
\end{equation}
as well as the L\"uders instrument $\mathcal I_i^P(\hat \rho) = \sqrt{\hat \Pi_i} \hat \rho \sqrt{\hat \Pi_i}$.
Using the resolution of the identity we obtain
\begin{align}
    \sum_{i=1}^N\braket{\psi|\hat{\Pi}_{C_i}-\hat{\Pi}_{ C_i}^2|\psi}
    &=1-\sum_{i=1}^N\braket{\psi|\hat{\Pi}_{ C_i}^2|\psi}.
\end{align}
Since every term in the sum is non-negative, we may keep only the contribution from the cell $ C$:
\begin{align}
    \sum_{i=1}^N\braket{\psi|\hat{\Pi}_{ C_i}-\hat{\Pi}_{ C_i}^2|\psi}
    &\leq 1-\braket{\psi|\hat{\Pi}_{ C}^2|\psi}.
\end{align}
Moreover, since $\hat{\Pi}_{ C}$ is self-adjoint, the Cauchy-Schwarz inequality gives
\begin{equation}
    \braket{\psi|\hat{\Pi}_{ C}|\psi}^2 \leq \braket{\psi|\psi}\braket{\psi|\hat{\Pi}_{ C}^2|\psi},
\end{equation}
with $\braket{\psi|\psi}=1$.
Therefore,
\begin{align} \label{estimate_1}
    \sum_{i=1}^N\braket{\psi|\hat{\Pi}_{ C_i}-\hat{\Pi}_{ C_i}^2|\psi}
    &\leq 1-\braket{\psi|\hat{\Pi}_{ C}|\psi}^2.
\end{align}
We now estimate the expectation value. By definition,
\begin{align}
    \braket{\psi|\hat{\Pi}_{ C}|\psi}
    &=\int_{ C}\frac{d^3\mathbf{x}d^3\mathbf{p}}{(2\pi\hbar)^3}|\braket{\mathbf{x},\mathbf{p}|\mathbf{x}_0,\mathbf{p}_0}|^2.
\end{align}
Using Eq.~\eqref{overlap} and $B_R^\sigma(\mathbf{x}_0,\mathbf{p}_0)\subseteq  C$, we get
\begin{align}
    \braket{\psi|\hat{\Pi}_{ C}|\psi}
    &\geq \int_{B_R^\sigma(\mathbf{x}_0,\mathbf{p}_0)}\frac{d^3\mathbf{x}d^3\mathbf{p}}{(2\pi\hbar)^3}\exp\left(-\frac{|\mathbf{x}-\mathbf{x}_0|^2}{2\sigma^2}-\frac{\sigma^2|\mathbf{p}-\mathbf{p}_0|^2}{2\hbar^2}\right).
\end{align}
For the dimensionless variables
\begin{equation}
    \mathbf{u}=\frac{\mathbf{x}-\mathbf{x}_0}{\sqrt{2}\sigma},\qquad \mathbf{v}=\frac{\sigma(\mathbf{p}-\mathbf{p}_0)}{\sqrt{2}\hbar},
\end{equation}

this becomes

\begin{align}
    \braket{\psi|\hat{\Pi}_{ C}|\psi}
    &\geq \frac{1}{\pi^3}\int_{|\mathbf{u}|^2+|\mathbf{v}|^2<R^2}d^3\mathbf{u}d^3\mathbf{v}\,e^{-|\mathbf{u}|^2-|\mathbf{v}|^2}.
\end{align}
Writing $r^2=|\mathbf{u}|^2+|\mathbf{v}|^2$ in $\mathbb R^6$, and using that the area of the unit sphere $S^5$ is $\pi^3$, we find
\begin{align}  
    \braket{\psi|\hat{\Pi}_{ C}|\psi} &\geq \int_0^R dr\,r^5e^{-r^2} 
    =1-e^{-R^2}\left(1+R^2+\frac{R^4}{2}\right). \label{estimate_2}
\end{align}

Combining the inequalities \eqref{estimate_1} and \eqref{estimate_2} yields
\begin{align}
    \sum_{i=1}^N\braket{\psi|\hat{\Pi}_{ C_i}-\hat{\Pi}_{ C_i}^2|\psi}
    &\leq 1-\left(1-e^{-R^2}\left(1+R^2+\frac{R^4}{2}\right)\right)^2.
\end{align}

This proves the claim, as the infimum is bounded by any trial state and instrument.
\end{proof}

\subsubsection{Proof of Proposition \ref{inf_momentum_compatibility_nonrelativistic}} \label{inf_momentum_compatibility_nonrelativistic_proof}

\begin{proof}
Let $ C=\mathbb R^3\times Q$ be the cell containing $B_{R'}^\sigma(\mathbf{x}_0,\mathbf{p}_0)$, and again choose the trial state
\begin{equation}
    \ket{\psi}=\ket{\mathbf{x}_0,\mathbf{p}_0}.
\end{equation}
Since $ C=\mathbb R^3\times Q$ contains $B^\sigma_{R'}(\mathbf{x}_0,\mathbf{p}_0)$, we have
\begin{equation}
    \left\{\mathbf{p}\in\mathbb R^3:\frac{\sigma^2|\mathbf{p}-\mathbf{p}_0|^2}{2\hbar^2}<R'^2\right\}\subseteq Q.
\end{equation}

Using $\sum_{i=1}^N\hat E_{\mathbf{\hat P}}(Q_i)=\hat{\mathbb I}$ and positivity, we find

\begin{align}
    &\sum_{i=1}^N\braket{\psi|\hat E_{\mathbf{\hat P}}(Q_i)-\hat E_{\mathbf{\hat P}}(Q_i)\hat{\Pi}_{\mathbb R^3 \times Q_i} \hat E_{\mathbf{\hat P}}(Q_i)|\psi} \notag\\
    &\leq 1-\braket{\psi|\hat E_{\hat{\mathbf P}}(Q)\hat{\Pi}_{\mathbb R^3\times Q}\hat E_{\hat{\mathbf P}}(Q)|\psi}. \label{estimate_3}
\end{align}

We compute the remaining expectation value in momentum space. Since $\hat{\Pi}_{\mathbb R^3\times Q}$ is obtained by integrating the coherent-state POVM over all positions, it acts in momentum space by multiplication, with
\begin{equation}
    F_{Q}(\mathbf{k})=\int_{Q}d^3\mathbf{p}\left(\frac{\sigma^2}{\pi\hbar^2}\right)^{3/2}\exp\left[-\frac{\sigma^2}{\hbar^2}|\mathbf{k}-\mathbf{p}|^2\right].
\end{equation}
Therefore,
\begin{align}
    \braket{\psi|\hat E_{\hat{\mathbf P}}(Q)\hat{\Pi}_{\mathbb R^3\times Q}\hat E_{\hat{\mathbf P}}(Q)|\psi}
    &=\int_{Q}d^3\mathbf{k}\,|\braket{\mathbf{k}|\mathbf{x}_0,\mathbf{p}_0}|^2F_{Q}(\mathbf{k}).
\end{align}
Using Eq.~\eqref{overlap_momentum_position} we get
\begin{align}
    &\braket{\psi|\hat E_{\hat{\mathbf P}}(Q)\hat{\Pi}_{\mathbb R^3\times Q}\hat E_{\hat{\mathbf P}}(Q)|\psi} \notag\\
    &=\int_{Q}d^3\mathbf{k}\int_{Q}d^3\mathbf{p}\left(\frac{\sigma^2}{\pi\hbar^2}\right)^3\exp\left[-\frac{\sigma^2}{\hbar^2}|\mathbf{k}-\mathbf{p}_0|^2-\frac{\sigma^2}{\hbar^2}|\mathbf{k}-\mathbf{p}|^2\right].
\end{align}
With dimensionless momentum variables
\begin{equation}
    \mathbf{u}=\frac{\sigma(\mathbf{k}-\mathbf{p}_0)}{\hbar},\qquad \mathbf{v}=\frac{\sigma(\mathbf{p}-\mathbf{p}_0)}{\hbar},
\end{equation}

this becomes

\begin{align}
    \braket{\psi|\hat E_{\hat{\mathbf P}}(Q)\hat{\Pi}_{\mathbb R^3\times Q}\hat E_{\hat{\mathbf P}}(Q)|\psi}
    &\geq \frac{1}{\pi^3}\int_{|\mathbf{u}|<\sqrt{2}R'}d^3\mathbf{u}\int_{|\mathbf{v}|<\sqrt{2}R'}d^3\mathbf{v}\,e^{-|\mathbf{u}|^2-|\mathbf{u}-\mathbf{v}|^2}.
\end{align}
We now restrict the integration domain further. If $|\mathbf{u}|<\frac{R'}{\sqrt{2}}$ and $|\mathbf{v}-\mathbf{u}|<\frac{R'}{\sqrt{2}}$, then
\begin{equation}
    |\mathbf{v}|\leq |\mathbf{u}|+|\mathbf{v}-\mathbf{u}|<\sqrt{2}R'.
\end{equation}
Thus this smaller domain is contained in the previous integration domain, and therefore
\begin{align}
    \braket{\psi|\hat E_{\hat{\mathbf P}}(Q)\hat{\Pi}_{\mathbb R^3\times Q}\hat E_{\hat{\mathbf P}}(Q)|\psi}
    &\geq \frac{1}{\pi^3}\int_{|\mathbf{u}|<R'/\sqrt{2}}d^3\mathbf{u}\int_{|\mathbf{v}-\mathbf{u}|<R'/\sqrt{2}}d^3\mathbf{v}\,e^{-|\mathbf{u}|^2-|\mathbf{u}-\mathbf{v}|^2}.
\end{align}
With the change of variables
\begin{equation}
    \mathbf{w}=\mathbf{v}-\mathbf{u},
\end{equation}
we obtain
\begin{align}
    \braket{\psi|\hat E_{\hat{\mathbf P}}(Q)\hat{\Pi}_{\mathbb R^3\times Q}\hat E_{\hat{\mathbf P}}(Q)|\psi}
    &\geq \frac{1}{\pi^3}\int_{|\mathbf{u}|<R'/\sqrt{2}}d^3\mathbf{u}\int_{|\mathbf{w}|<R'/\sqrt{2}}d^3\mathbf{w}\,e^{-|\mathbf{u}|^2-|\mathbf{w}|^2} \\
    &=\left(\frac{1}{\pi^{3/2}}\int_{|\mathbf{u}|<R'/\sqrt{2}}d^3\mathbf{u}\,e^{-|\mathbf{u}|^2}\right)^2\\
    &= \left(\operatorname{erf}\left(\frac{R'}{\sqrt{2}}\right)-\sqrt{\frac{2}{\pi}}R'e^{-R'^2/2}\right)^2.
\end{align}

Combining this with \eqref{estimate_3} yields the claim

\begin{equation}
    \inf_{\braket{\psi|\psi} = 1} \sum_{i=1}^N\braket{\psi|\hat E_{\mathbf{\hat P}}(Q_i)-\hat E_{\mathbf{\hat P}}(Q_i)\hat{\Pi}_{\mathbb R^3 \times Q_i} \hat E_{\mathbf{\hat P}}(Q_i)|\psi} \le 1 - \left(\operatorname{erf}\left(\frac{R'}{\sqrt{2}}\right)-\sqrt{\frac{2}{\pi}}R'e^{-R'^2/2}\right)^2.
\end{equation}
\end{proof}

\subsubsection{Proof of Proposition \ref{inferential_completeness_rectangles}} \label{inferential_completeness_rectangles_proof}

\begin{proof}
We first show that rectangles of arbitrary size belong to $\mathcal D(\mathcal R_\ell)$. Consider an arbitrary
\begin{equation}
    R=(a_1,b_1]\times\cdots\times(a_6,b_6].
\end{equation}
Choose numbers $c_j>b_j$ such that  $c_j-b_j\geq\ell \, \forall j$. Then, we also have $c_j-a_j\geq\ell \, \forall j$. For each subset $J\subseteq\{1,\dots,6\}$, define
\begin{equation}
    R_J:=\prod_{j=1}^6 I_j^{(J)},
    \qquad
    I_j^{(J)}:=
    \begin{cases}
        (b_j,c_j], & j\in J,\\
        (a_j,c_j], & j\notin J.
    \end{cases}
\end{equation}
Every $R_J$ belongs to $\mathcal R_\ell$, because each factor has length either $c_j-b_j\geq\ell$ or $c_j-a_j\geq\ell$. Moreover,
\begin{equation}
    R=R_\varnothing\setminus \bigcup_{\emptyset\neq J\subseteq\{1,\dots,6\}} R_J.
\end{equation}
To make this a disjoint Dynkin-system construction, set
\begin{equation}
    A_J:=\left(\prod_{j\in J}(b_j,c_j]\right)\times
    \left(\prod_{j\notin J}(a_j,b_j]\right).
\end{equation}
Then the rectangles $A_J$ are pairwise disjoint and
\begin{equation}
    R_\varnothing
    =
    \bigsqcup_{J\subseteq\{1,\dots,6\}} A_J.
\end{equation}
For $J\neq\varnothing$, the set $A_J$ can be written as a finite alternating difference of the coarse rectangles $R_K$ with $K\supseteq J$. Explicitly,
\begin{equation}
    A_J
    =
    R_J\setminus \bigsqcup_{K\supsetneq J} A_K.
\end{equation}

Starting from the maximal sets $J=\{1,\dots,6\}$ and proceeding by downwards induction in $|J|$, this shows that each $A_J$ with $J\neq\varnothing$ belongs to $\mathcal D(\mathcal R_\ell)$. Since $\mathcal D(\mathcal R_\ell)$ is closed under disjoint unions, it follows that
\begin{equation}
    \bigsqcup_{\emptyset\neq J\subseteq\{1,\dots,6\}} A_J
    \in
    \mathcal D(\mathcal R_\ell).
\end{equation}
This set is contained in $R_\varnothing$, and $R_\varnothing\in\mathcal R_\ell\subseteq\mathcal D(\mathcal R_\ell)$. Since a Dynkin system is closed under relative complements of nested sets, we obtain
\begin{equation}
    R
    =
    R_\varnothing
    \setminus
    \bigsqcup_{\emptyset\neq J\subseteq\{1,\dots,6\}} A_J
    \in
    \mathcal D(\mathcal R_\ell).
\end{equation}
Thus every bounded half-open rectangle in $\Gamma$ belongs to $\mathcal D(\mathcal R_\ell)$.

The bounded half-open rectangles form a $\pi$-system generating the Borel $\sigma$-algebra $\mathcal B(\Gamma)$. Therefore, by the $\pi$-$\lambda$ theorem (Theorem 3.2 in \cite{Billingsley1995}),
\begin{equation}
    \mathcal B(\Gamma)\subseteq \mathcal D(\mathcal R_\ell).
\end{equation}
Conversely, since every element of $\mathcal R_\ell$ is Borel, and $\mathcal B(\Gamma)$ is itself a Dynkin system, minimality implies
\begin{equation}
    \mathcal D(\mathcal R_\ell)=\mathcal B(\Gamma).
\end{equation}
By definition, $\sigma_\Gamma$ is the completion of the Borel algebra $B(\Gamma)$. Thus,
\begin{equation}
    \overline{\mathcal D(\mathcal R_\ell)}^\nu= \overline{\mathcal B(\Gamma)}^\nu = \sigma_\Gamma.
\end{equation}
Finally, suppose that
\begin{equation}
    \mathcal R_\ell\subseteq \mathcal G(\mathcal P_{\mathrm{CG}}).
\end{equation}
Then, by monotonicity of Dynkin closure and $\nu$-completion,
\begin{equation}
    \overline{\mathcal D(\mathcal G(\mathcal P_{\mathrm{CG}}))}^\nu
    =
    \sigma_\Gamma.
\end{equation}
\end{proof}

\subsubsection{Proof of Eq.~\eqref{res_identity_massive_rel}} \label{res_identity_massive_rel_proof}

\begin{proof}
Writing
\begin{equation}
    C_\alpha^2:=\frac{\alpha}{\pi mK_1(2\alpha m)},
\end{equation}
the coherent states satisfy
\begin{equation}
    \braket{\mathbf{k}|\mathbf{x},\mathbf{p}}
    =
    C_\alpha
    \exp\!\left(-\frac{i}{\hbar}\mathbf{k}\cdot\mathbf{x}\right)
    \exp\!\left(-\alpha\frac{p_\mu k^\mu}{m}\right),
\end{equation}
where \(p=(p^0,\mathbf{p})\), \(p^0=\sqrt{m^2+|\mathbf{p}|^2}\). Consider the operator
\begin{equation}
    \hat A
    :=
    \int_{\Gamma}
    \frac{d^3\mathbf{x}\,d^3\mathbf{p}}{(2\pi\hbar)^3}
    \ket{\mathbf{x},\mathbf{p}}\bra{\mathbf{x},\mathbf{p}}.
\end{equation}
Its momentum-space kernel is
\begin{align}
    A(\mathbf{k},\mathbf{k}') := \braket{\mathbf{k}|\hat A| \mathbf{k'}} 
    &=
    C_\alpha^2
    \int
    \frac{d^3\mathbf{x}\,d^3\mathbf{p}}{(2\pi\hbar)^3}
    \exp\!\left(-\frac{i}{\hbar}(\mathbf{k}-\mathbf{k}')\cdot\mathbf{x}\right)
    \exp\!\left(-\alpha\frac{p\cdot k+p\cdot k'}{m}\right) \\
    &=
    C_\alpha^2
    \delta^{(3)}(\mathbf{k}-\mathbf{k}')
    \int d^3\mathbf{p}\,
    \exp\!\left(-2\alpha\frac{p\cdot k}{m}\right).
\end{align}
It remains to evaluate
\begin{equation}
    I(k)
    :=
    \int \frac{d^3\mathbf{p}}{p^0}\,
    p^0 \exp\!\left(-2\alpha\frac{p\cdot k}{m}\right).
\end{equation}

The measure \(d^3\mathbf{p}/p^0\) is Lorentz invariant. Thus, applying a Lorentz transformation $\Lambda_{\mathbf k}$ so that $\Lambda_{\mathbf k} k = (m,\mathbf 0)$, yields

\begin{equation}
    I(k) = \int \frac{d^3\mathbf{p}}{p^0}\,
    (\Lambda_{\mathbf k}^{-1} p)^0 \exp\!\left(-2\alpha p^0\right) = \frac{k^0}{m} \int d^3\mathbf{p} \,
    e^{-2\alpha p^0},
\end{equation} 

as $(\Lambda_{\mathbf k}^{-1} p)^0 = \frac{k^0}{m} p^0 + \text{ terms odd in } \mathbf{p}$, and the latter terms vanish when integrating against the rotationally symmetric kernel. Using the substitution
\begin{equation}
    p^0=m\cosh t,
\end{equation}
this can be evaluated to
\begin{equation}
    I(k) = \frac{2\pi m}{\alpha}K_2(2\alpha m)\,k^0.
\end{equation}
Therefore
\begin{align}
    A(\mathbf{k},\mathbf{k}')
    &=
    \frac{K_2(2\alpha m)}{K_1(2\alpha m)}
    2k^0\delta^{(3)}(\mathbf{k}-\mathbf{k}').
\end{align}
With respect to the measure \(d^3\mathbf{k}/(2k^0)\), the identity operator has kernel
\begin{equation}
    \mathbb I(\mathbf{k},\mathbf{k}')
    =
    2k^0\delta^{(3)}(\mathbf{k}-\mathbf{k}').
\end{equation}
Hence
\begin{equation}
    \hat A
    =
    \frac{K_2(2\alpha m)}{K_1(2\alpha m)} \mathbb{\hat I}.
\end{equation}
\end{proof}

\subsubsection{Proof of Eq.~\eqref{covariance_massive_relativistic}} \label{covariance_massive_relativistic_proof}

\begin{proof}
The coherent states have wavefunction
\begin{equation}
    \psi_{\mathbf x,\mathbf p}(\mathbf k) := \braket{\mathbf{k}|\mathbf{x},\mathbf{p}}
    =
    C_\alpha
    \exp\!\left(-\frac{i}{\hbar}\mathbf{k}\cdot\mathbf{x}\right)
    \exp\!\left(-\alpha\frac{p \cdot k}{m}\right).
\end{equation}

Thus, for $x = (0,\mathbf x)$,

\begin{align}
\braket{\mathbf k|\hat U(T_a)\hat U(\Lambda)|\mathbf x,\mathbf p}
&=e^{ik\cdot a/\hbar}\psi_{\mathbf x,\mathbf p}(\Lambda^{-1}\mathbf k)\\
&=C_\alpha
\exp\!\left(\frac{i}{\hbar}\bigl[(\Lambda^{-1}k)\cdot x+k\cdot a\bigr]\right)
\exp\!\left(-\alpha\frac{p\cdot(\Lambda^{-1}k)}{m}\right)\\
&=C_\alpha
\exp\!\left(\frac{i}{\hbar}k\cdot(\Lambda x+a)\right)
\exp\!\left(-\alpha\frac{(\Lambda p)\cdot k}{m}\right)\\
&=C_\alpha
\exp\!\left(-\frac{i}{\hbar}\mathbf k\cdot\mathbf x'\right)
\exp\!\left(-\alpha\frac{p'\cdot k}{m}\right)\\
&=\braket{\mathbf k|\mathbf x',\mathbf p'}.
\end{align}

where we used the assumptions $\Lambda x + a = (0,\mathbf{x'})$ and $\Lambda \cdot (p^0,\mathbf p) = ((p')^0,\mathbf {p'})$. As this holds for all $\mathbf k$, this concludes the proof.
\end{proof}

\subsubsection{Proof of Eq.~\eqref{overlap_relativistic}} \label{overlap_relativistic_proof}

\begin{proof}
The coherent-state wavefunction is
\begin{equation}
    \braket{\mathbf{k}|\mathbf{x},\mathbf{p}}=\left(\frac{\alpha}{\pi mK_1(2\alpha m)}\right)^{1/2}\exp\!\left(\frac{i}{\hbar}k\cdot x\right)\exp\!\left(-\alpha\frac{p\cdot k}{m}\right),
\end{equation}
Therefore
\begin{align}
    \braket{\mathbf{x},\mathbf{p}|\mathbf{x}',\mathbf{p}'} &= \frac{\alpha}{\pi mK_1(2\alpha m)}\int \frac{d^3\mathbf{k}}{2k^0}\exp\!\left(-\frac{\alpha}{m}(p+p')\cdot k-\frac{i}{\hbar}(x-x')\cdot k\right) \\
    &= \frac{\alpha}{\pi mK_1(2\alpha m)}\int \frac{d^3\mathbf{k}}{2k^0}\exp\!\left(-\xi\cdot k\right),
\end{align}
with
\begin{equation} \label{xi}
    \xi^\mu=\frac{\alpha}{m}(p+p')^\mu+\frac{i}{\hbar}(x-x')^\mu.
\end{equation}

To compute this integral, first assume that \(\xi\) is real-valued. Then, by
Lorentz invariance of \(d^3\mathbf{k}/(2k^0)\), the integral depends only
on the invariant \(\xi^2\), and we may evaluate it in a frame where
\(\xi^\mu=(\sqrt{\xi^2},\mathbf 0)\). Using the substitution $k^0 = m \cosh t$, this gives
\begin{equation} \label{int_e_xi}
    \int \frac{d^3\mathbf{k}}{2k^0}e^{-\xi\cdot k} =\frac{2\pi m}{\sqrt{\xi^2}}K_1\!\left(m\sqrt{\xi^2}\right).
\end{equation}
To generalise this to complex $\xi$, note that both sides of Eq.~\eqref{int_e_xi} are analytic functions of \(\xi\) as long as $\text{Re} \, \xi$ is future time-like, and hence, the identity extends by analytic continuation to all complex \(\xi\) in this domain. Since, by Eq.~\eqref{xi}, $\text{Re} \, \xi = \frac{\alpha}{m}(p+p')$ is always future time-like, this concludes the proof.

\end{proof}

\subsubsection{Proof of Theorem \ref{coarse_graining_relativistic}} \label{coarse_graining_relativistic_proof}

\begin{proof}
    As discussed in section \ref{massive_relativistic_particle}, the map $C \mapsto \hat \Pi_C$ is a regular, Poincar\'e covariant phase-space POVM. Clearly, the condition defining $\mathcal P_{\text{CG}}$ in \eqref{P_CG_relativistic} is preserved under coarsening, ensuring Eq.~\eqref{closure_under_coarsening_definition} holds. Furthermore, by definition of $\mathcal P_{\mathrm{CG}}$, for any $\mathbf{p_0} \in \mathbb R^3$, the implication

    \begin{equation}
        C \subseteq \mathbb{R}^3 \times (\mathbb R^3 \setminus H_{R_{\mathrm{min}}}(\mathbf p_0)) \, \Rightarrow C \, \triangleleft \, \mathcal P_{\mathrm{CG}}(R_{\mathrm{min}}),
    \end{equation}

    holds, as the complement of the hyperbolic ball can be partitioned freely. As massive hyperbolic balls are compact, this ensures that arbitrarily small rectangles are coarse. By applying disjoint countable union and completing by $\nu$-null sets, these small rectangles generate the entire Lebesgue $\sigma$-algebra. Thus, Eq.~\eqref{dynkin_inferential_completeness} holds and we can conclude that $(\Gamma, \hat \Pi, \mathcal P_{\mathrm{CG}}(R_{\mathrm{min}}))$ is indeed a coarse-graining scheme.

As in the nonrelativistic estimates above, both classicality bounds follow by concentrating a preparation inside a coarse enough cell. Let $P \in\mathcal P_{\mathrm{CG}}$. By definition, there exist a cell $C_0\in P$ and a momentum $\mathbf p_0$ such that
\begin{equation}
    \mathbb R^3\times H_R(\mathbf p_0)\subseteq C_0,
\end{equation}
where $R=R_{\mathrm{min}}(\epsilon,\epsilon')$. Position-integrated POVM elements are diagonal in momentum space,
\begin{equation}
    \hat\Pi_{\mathbb R^3\times Q}
    =
    F_Q(\hat{\mathbf P}), \qquad
F_Q(\mathbf k)
=
\frac{\int_Q d^3\mathbf p\,e^{-2\alpha p\cdot k/m}}
{\int_{\mathbb R^3}d^3\mathbf p\,e^{-2\alpha p\cdot k/m}}.
\end{equation}

To evaluate $F_{H_R(\mathbf p_0)}(\mathbf p_0)$, let $\Lambda$ be a Lorentz transformation satisfying $\Lambda(m,\mathbf0)=p_0$, and write $p=\Lambda q$. The invariant mass-shell measure gives
\begin{equation}
d^3\mathbf p=\frac{(\Lambda q)^0}{q^0}\,d^3\mathbf q.
\end{equation}
Moreover, $p\cdot p_0=mq^0$, and $H_R(\mathbf p_0)$ becomes the ball $|\mathbf q|\le m\sinh R$. Writing $(\Lambda q)^0=\Lambda ^0{}_0q^0+\sum_{j=1}^3\Lambda^0{}_jq^j$, the terms odd in $\mathbf q$ integrate to zero over both this ball and all of momentum space. The remaining factor $\Lambda^0{}_0=p_0^0/m$ cancels between numerator and denominator. Hence
\begin{equation}
F_{H_R(\mathbf p_0)}(\mathbf p_0)
= \eta_{\alpha,m}(R) :=
\frac{\int_{|\mathbf q|\le m\sinh R}d^3\mathbf q\,e^{-2\alpha q^0}}
{\int_{\mathbb R^3}d^3\mathbf q\,e^{-2\alpha q^0}},
\end{equation}
is independent of $\mathbf p_0$. Additionally, as the numerator converges monotonically to the denominator, we have
\begin{equation}
    \eta_{\alpha,m}(R)\rightarrow 1 \quad (R \rightarrow \infty).
\end{equation} 

We first use this to prove weak non-invasive repeatability, applying a L\"uders instrument. Let $\ket{\psi_n} \in \mathcal H$ be a sequence of normalized states whose square momentum-space wavefunctions converge to a delta distribution at $\mathbf p_0$. Then
\begin{equation}
\bra{\psi_n}\hat\Pi_{C_0}\ket{\psi_n} \ge \bra{\psi_n}\hat\Pi_{\mathbb R^3 \times H_R(\mathbf{p}_0)}\ket{\psi_n}
    \rightarrow
    F_{H_R(\mathbf{p}_0)}(\mathbf p_0)
    =
    \eta_{\alpha,m}(R).
\end{equation}
Using $\sum_i\hat\Pi_{C_i}=\mathbb{\hat I}$, dropping all terms except $i=0$ and using
$\langle \hat A^2\rangle\geq \langle \hat A\rangle^2$ for positive operators implies
\begin{equation}
    \sum_i
    \bra{\psi_n}
    (\hat\Pi_{C_i}-\hat\Pi_{C_i}^2)
    \ket{\psi_n}
    \leq
    1-\bra{\psi_n}\hat\Pi_{C_0}\ket{\psi_n}^2.
\end{equation}
Taking the limit $n\to\infty$ yields
\begin{equation}
    \inf_{\|\psi\|=1}
    \sum_i
    \bra{\psi}
    (\hat\Pi_{C_i}-\hat\Pi_{C_i}^2)
    \ket{\psi}
    \leq
    1-\eta_{\alpha,m}(R)^2.
\end{equation}
Therefore, weak non-invasive repeatability with tolerance $\epsilon$ is guaranteed whenever
\begin{equation}
    1-\eta_{\alpha,m}(R)^2<\epsilon.
\end{equation}

We next prove classical-quantum momentum compatibility. For a partition $P = \{\mathbb R^3 \times Q_i\}_{i\in I}$, let
\begin{equation}
    G(\mathbf k)
    :=
    \sum_i
    \chi_{Q_i}(\mathbf k)
    \bigl(1-F_{Q_i}(\mathbf k)\bigr).
\end{equation}
Then,
\begin{equation}
    \sum_i
    \hat E_{\hat{\mathbf P}}(Q_i)
    -
    \hat E_{\hat{\mathbf P}}(Q_i)
    \hat\Pi_{\mathbb R^3\times Q_i}
    \hat E_{\hat{\mathbf P}}(Q_i)
    =
    G(\hat{\mathbf P}).
\end{equation}
Evaluating again on the concentrating sequence $\ket{\psi_n}$,
\begin{equation}
    \lim_{n\to\infty}
    \bra{\psi_n}G(\hat{\mathbf P})\ket{\psi_n}
    =
    1-F_{Q_0}(\mathbf p_0)
    \leq
    1-\eta_{\alpha,m}(R).
\end{equation}
Hence,
\begin{equation}
    \inf_{\|\psi\|=1}
    \sum_i
    \bra{\psi}
    \hat E_{\hat{\mathbf P}}(Q_i)
    -
    \hat E_{\hat{\mathbf P}}(Q_i)
    \hat\Pi_{\mathbb R^3\times Q_i}
    \hat E_{\hat{\mathbf P}}(Q_i)
    \ket{\psi}
    \leq
    1-\eta_{\alpha,m}(R).
\end{equation}
Thus momentum compatibility with tolerance $\epsilon'$ is guaranteed whenever
\begin{equation}
    1-\eta_{\alpha,m}(R)\leq\epsilon'.
\end{equation}
We can thus choose $R_{\mathrm{min}}(\epsilon,\epsilon')$ large enough so that both
\begin{equation}
    1-\eta_{\alpha,m}(R)^2\leq\epsilon,
    \qquad
    1-\eta_{\alpha,m}(R)\leq\epsilon'
\end{equation}
hold, and be guaranteed that the coarse-graining scheme fulfills Defs. \ref{def:weak-noninvasive-repeatability} and \ref{def:weak-momentum-compatibility}.

Finally, consider the binary partition
\begin{equation}
    C^\pm
    =
    \{(\mathbf x,\mathbf p)\in\Gamma:\pm p^1>0\}.
\end{equation}
Choose $\mathbf p_0=(P,0,0)$ with
\begin{equation}
    P>m\sinh R_{\mathrm{min}}.
\end{equation}
Then, the hyperbolic ball $H_{R_{\mathrm{min}}}(\mathbf p_0)$ lies entirely in the positive hemisphere, so
\begin{equation}
    \mathbb R^3\times H_{R_{\mathrm{min}}}(\mathbf p_0)
    \subseteq
    C^+.
\end{equation}
Hence, $\{C^+,C^-\}\in\mathcal P_{\mathrm{CG}}(R_{\mathrm{min}})$ and Def.~\ref{def:minimal-directional-information-gain} is satisfied for arbitrary tolerance $\delta$. This concludes the proof.
\end{proof}

\subsection{Proofs of Appendix~\ref{appendix_minimal_information}}

\subsubsection{Proof of Proposition~\ref{M_delta_proposition}}\label{M_delta_proof}

\begin{proof}
Let $\Omega$ be the angular region supplied by Definition~\ref{def:minimal-directional-information-gain}, and set $s:=\sigma(\Omega)$. The associated binary partition is coarse and satisfies $(1-\delta)/2\le s\le1/2$. Since $h_2$ is increasing on $[0,1/2]$, its normalized entropy obeys
\begin{equation}
\frac{H(\{\Omega,S^2\setminus\Omega\})}{\log2}
=
h_2(s)
\ge
h_2\!\left(\frac{1-\delta}{2}\right).
\end{equation}
This proves the claim.
\end{proof}

\subsubsection{Proof of Proposition~\ref{M_sigma_Omega}}\label{M_sigma_Omega_proof}

\begin{proof}
Write $C_i:=\mathbb R^3\times Q_{\Omega_i}$. Proposition~\ref{prop:angular-response-area-bounds} gives
\begin{equation}
\hat\Pi_{C_i}\le\bigl[1-(1-s_i)^2\bigr]\hat{\mathbb I}.
\end{equation}
Since $1-s_i\ge r$, every outcome therefore satisfies
\begin{equation}\label{angular_failure_lower_bound}
\hat{\mathbb I}-\hat\Pi_{C_i}\ge r^2\hat{\mathbb I}.
\end{equation}
Thus no state can make the probability of this outcome exceed $1-r^2$.

For any instrument $\mathcal I^P$ implementing the angular POVM and any normalized state $\hat\rho$, the probability of obtaining different outcomes in an immediate repetition satisfies
\begin{align}
\sum_i\sum_{j\ne i}\operatorname{tr}\!\left[\hat\Pi_{C_j}\mathcal I_i^P(\hat\rho)\right]
&=
\sum_i\operatorname{tr}\!\left[(\hat{\mathbb I}-\hat\Pi_{C_i})\mathcal I_i^P(\hat\rho)\right]\\
&\ge
r^2\sum_i\operatorname{tr}\!\left[\mathcal I_i^P(\hat\rho)\right]
=r^2.
\end{align}
The last equality follows because the nonselective instrument is trace preserving. Taking the infima in Definition~\ref{def:weak-noninvasive-repeatability} shows that weak repeatability requires $r^2<\epsilon$.

For momentum compatibility, applying Eq.~\eqref{angular_failure_lower_bound}  gives
\begin{align}
\sum_i\sum_{j\ne i}\operatorname{tr}\!\left[ \hat\Pi_{C_j}\hat E_{\hat{\mathbf P}}(Q_{\Omega_i})\hat\rho\hat E_{\hat{\mathbf P}}(Q_{\Omega_i})\right]
&=
\sum_i\operatorname{tr}\!\left[(\hat{\mathbb I}-\hat \Pi_{C_i})\hat E_{\hat{\mathbf P}}(Q_{\Omega_i})\hat\rho\hat E_{\hat{\mathbf P}}(Q_{\Omega_i})\right]\\
&\ge
r^2\sum_i\operatorname{tr}\!\left[\hat E_{\hat{\mathbf P}}(Q_{\Omega_i})\hat\rho\hat E_{\hat{\mathbf P}}(Q_{\Omega_i})\right]
=r^2.
\end{align}
Here we used that the $\hat E_{\hat{\mathbf P}}(Q_{\Omega_i})$ form a normalized PVM. Definition~\ref{def:weak-momentum-compatibility} consequently requires $r^2<\epsilon'$.

It remains to bound the entropy. Relabel the largest cell so that $s_1=1-r$. If $r=0$, the entropy vanishes and the claim is immediate. Otherwise, writing $s_i=r t_i$ for $i\ge2$ gives a probability distribution $\{t_i\}_{i=2}^N$ and
\begin{equation}
H(s_1,\dots,s_N)
=
(\log2)h_2(r)+rH(t_2,\dots,t_N)
\le
(\log2)h_2(r)+r\log(N-1).
\end{equation}
For $0\le r\le1/2$, concavity of the binary entropy gives $h_2(r)\ge2r$. Moreover, $2\sqrt{N-1}\le N$ for every $N\ge2$. Therefore
\begin{align}
H(s_1,\dots,s_N)
&\le
h_2(r)\left[\log2+\frac12\log(N-1)\right]\\
&=
h_2(r)\log\!\left(2\sqrt{N-1}\right)
\le
h_2(r)\log N.
\end{align}
Dividing by $\log N$ proves Eq.~\eqref{entropy_dominant_cell}.
\end{proof}

\subsubsection{Proof of Proposition~\ref{prop:entropy-no-go}}\label{proof:entropy-no-go}

\begin{proof}
Choose the angular partition supplied by Definition~\ref{def:weak-minimal-directional-information-gain}, and let $r$ be as in Proposition~\ref{M_sigma_Omega}. If weak repeatability holds with $0<\epsilon<1/4$, that proposition gives $r<\sqrt{\epsilon}<1/2$. Hence
\begin{equation}
M
\le
\frac{H(\{\Omega_i\}_{i=1}^N)}{\log N}
\le
h_2(r)
<
h_2(\sqrt{\epsilon}),
\end{equation}
where the strict inequality follows from the strict monotonicity of $h_2$ on $[0,1/2]$. The argument for momentum compatibility is identical, using $r<\sqrt{\epsilon'}<1/2$. This proves both assertions.
\end{proof}

\subsection{Proofs of Appendix \ref{stronger_classicality_nonrelativistic}}

\subsubsection{Proof of Proposition \ref{general_bound_non_repeatability_nonrelativistic}}\label{general_bound_non_repeatability_nonrelativistic_proof}

\begin{proof}
For simplicity, we write 
\begin{equation}
    \ket{z}=\ket{\mathbf{x},\mathbf{p}}, \qquad d\mu(z)=\frac{d^3\mathbf{x}\,d^3\mathbf{p}}{(2\pi\hbar)^3}, \qquad Q_\psi(z) = |\braket{z|\psi}|^2.
\end{equation}

We want to bound

\begin{equation}
    S := \sum_{i=1}^N \braket{\psi|\hat{\Pi}_{C_i} - \hat{\Pi}_{C_i}^2|\psi}= \sum_{i=1}^N \sum_{j\neq i}\braket{\psi|\hat{\Pi}_{C_i}\hat{\Pi}_{C_j}|\psi}.
\end{equation}

By the triangle inequality and the identity $2ab \le a^2 + b^2$, we have
\begin{align}
    \left|\braket{\psi|\hat{\Pi}_{C_i}\hat{\Pi}_{C_j}|\psi}\right|
    &= \left |\int_{C_i} d\mu(z)\int_{C_j} d\mu(z')\,\braket{\psi|z}\braket{z|z'}\braket{z'|\psi} \right | \notag\\
    &\le \int_{C_i} d\mu(z) \int_{C_j} d\mu(z')\,|\braket{\psi|z}|\,|\braket{z|z'}|\,|\braket{z'|\psi}|\\
    &\le \frac{1}{2}\int_{C_i} d\mu(z)\int_{C_j} d\mu(z')\,|\braket{z|z'}|\left(Q_\psi(z)+Q_\psi(z')\right),
\end{align}

where, by Eq.~\eqref{overlap},

\begin{equation}
    |\braket{z|z'}|=e^{-d_{\sigma}(z,z')^2/2}.
\end{equation}

Thus,
\begin{align}
    S &\le\frac{1}{2}\sum_{i=1}^N\int_{C_i} d\mu(z)\int_{\Gamma\setminus C_i} d\mu(z')\,e^{-d_{\sigma}(z,z')^2/2}\left(Q_\psi(z)+Q_\psi(z')\right).
\end{align}

We now split the right-hand side into a near-boundary and a far-from-boundary contribution. For fixed $R>0$, write
\begin{align}
    S &\le S_{\mathrm{near}}+S_{\mathrm{far}},
\end{align}
with

\begin{align}
S_{\mathrm{near}} &:= \frac{1}{2}\sum_{i=1}^N \int_{C_i} d\mu(z) \int_{\Gamma \setminus C_i} d\mu(z')\, \theta(R - d_{\sigma}(z,z')) \notag \\
&\hspace{2cm}\times e^{-d_{\sigma}(z,z')^2/2}\left(Q_\psi(z)+Q_\psi(z')\right),\\
S_{\mathrm{far}}&:= \frac{1}{2}\sum_{i=1}^N\int_{C_i} d\mu(z)\int_{\Gamma \setminus C_i} d\mu(z')\,\theta(d_{\sigma}(z,z')-R) \notag \\
&\hspace{2cm}\times e^{-d_{\sigma}(z,z')^2/2}\left(Q_\psi(z)+Q_\psi(z')\right),
\end{align}

and $\theta(x) =
\begin{cases}
0, & x < 0, \\
1, & x \geq 0 .
\end{cases}$

We first estimate $S_{\mathrm{near}}$. If $z\in C_i$, $z'\in\Gamma\setminus C_i$, and $d_{\sigma}(z,z')\le R$,
then by definition $z \in \partial_R P \cap C_i$.

Therefore, the term with $Q_\psi(z)$ is bounded by
\begin{align}
    &\frac{1}{2} \sum_{i=1}^N \int_{C_i} d\mu(z) \int_{\Gamma \setminus C_i} d\mu(z')\, \theta(R - d_{\sigma}(z,z')) \, e^{-d_{\sigma}(z,z')^2/2} Q_\psi(z)\\
    &\le \frac{1}{2}\sum_{i=1}^N \int_{(\partial_R P) \,  \cap \, C_i} d\mu(z)\, Q_\psi(z) \int_{\Gamma} d\mu(z')\,e^{-d_{\sigma}(z,z')^2/2} = 4 \eta_R(\psi,P),
\end{align}

where we used the integral
\begin{equation}
    \int_{\Gamma} d\mu(z')\,e^{-d_{\sigma}(z,z')^2/2} = 8,
\end{equation}
and the definition of $\eta_R$ in Eq.~\eqref{eta_R}.

Since everything is invariant under exchanging $z \leftrightarrow z'$, the term containing $Q_\psi(z')$ has the same bound and we find
\begin{equation}
    S_{\mathrm{near}}\le 8\eta_R(\psi,P).
\end{equation}

It remains to estimate $S_{\mathrm{far}}$. We again start with the part containing $Q_\psi(z)$.

We have
\begin{align}
    &\frac{1}{2}\sum_{i=1}^N \int_{C_i} d\mu(z)\int_{\Gamma \setminus C_i} d\mu(z')\, \theta(d_{\sigma}(z,z')-R) \, e^{-d_{\sigma}(z,z')^2/2}Q_\psi(z) \notag\\
    &\le \frac{1}{2} \sum_{i=1}^N \int_{C_i} d\mu(z)\, Q_\psi(z) \int_{\Gamma} d\mu(z')\, \theta(d_{\sigma}(z,z')-R) \, e^{-d_{\sigma}(z,z')^2/2}.
\end{align}
By translation invariance, the inner integral is independent of $z$. We denote it by
\begin{equation} \label{tau_R}
    \tau(R) := \int_{\Gamma} d\mu(z')\, \theta(d_{\sigma}(z,z')-R) \, e^{-d_{\sigma}(z,z')^2/2}.
\end{equation}

Thus, 

\begin{align}
    &\frac{1}{2} \sum_{i=1}^N \int_{C_i} d\mu(z) \int_{\Gamma} d\mu(z')\,\theta(d_{\sigma}(z,z') - R) \, e^{-d_{\sigma}(z,z')^2/2}Q_\psi(z) \notag\\
    &\le \frac{1}{2}\tau(R) \sum_{i=1}^N \int_{C_i} d\mu(z)\,Q_\psi(z) = \frac{1}{2}\tau(R),
\end{align}

by the resolution of identity.

Applying the same argument to the term with $Q_\psi(z')$ yields
\begin{equation}
    S_{\mathrm{far}} \le \tau(R).
\end{equation}

To conclude the proof, it is possible to compute the integral in Eq.~\eqref{tau_R} analytically, giving

\begin{equation}
    \tau(R) = e^{-R^2/2}(R^4 + 4R^2 + 8).
\end{equation}
\end{proof}

\subsubsection{Proof of Proposition ~\ref{general_bound_non_compatibility_nonreativistic}}\label{general_bound_non_compatibility_nonreativistic_proof}

\begin{proof}
We write
\begin{equation}
    \ket{\psi_i} := \hat E_{\hat{\mathbf P}}(Q_i) \ket{\psi}
\end{equation}
so that the quantity we want to bound becomes
\begin{equation}
    \sum_{i=1}^N S_i := \sum_{i=1}^N \braket{\psi_i|\hat{\Pi}_{\mathbb R^3\times(\mathbb R^3\setminus Q_i)}|\psi_i},
\end{equation}
where $\braket{\mathbf{k}|\psi_i} = \chi_{Q_i}(\mathbf{k}) \braket{\mathbf{k}|\psi}$.

Using that the $\mathbf{x}$-dependence of coherent states is a plane-wave phase, and Eq.~\eqref{overlap_momentum_position}, we find the matrix-elements
\begin{align}
    \braket{\mathbf{k}|\hat{\Pi}_{\mathbb R^3\times Q}|\mathbf{k}'} &=\delta^{(3)}(\mathbf{k}-\mathbf{k}')\int_{Q} d^3\mathbf{p}\,\left(\frac{\sigma^2}{\pi\hbar^2}\right)^{3/2}\exp\left[-\frac{\sigma^2}{\hbar^2}|\mathbf{k}-\mathbf{p}|^2 \right].
\end{align}
Therefore $\hat{\Pi}_{\mathbb R^3\times Q}$ acts as multiplication in momentum space by $\int_{Q} d^3\mathbf{p}\, G_{\sigma}(\mathbf{k}-\mathbf{p})$ with the normalized Gaussian
\begin{equation}
    G_{\sigma}(\mathbf{k}-\mathbf{p}) =\left(\frac{\sigma^2}{\pi\hbar^2}\right)^{3/2}\exp\left[-\frac{\sigma^2}{\hbar^2}|\mathbf{k}-\mathbf{p}|^2 \right].
\end{equation}

Applying this to $Q=\mathbb R^3\setminus Q_i$, we get
\begin{align}
    S_i&=\int_{\mathbb R^3} d^3\mathbf{k}\,|\braket{\mathbf{k}|\psi_i}|^2\int_{\mathbb R^3\setminus Q_i}d^3\mathbf{p}\,G_{\sigma}(\mathbf{k}-\mathbf{p}) \notag \\
    &=\int_{Q_i} d^3\mathbf{k}\,|\braket{\mathbf{k}|\psi}|^2\int_{\mathbb R^3\setminus Q_i}d^3\mathbf{p}\,G_{\sigma}(\mathbf{k}-\mathbf{p}).
\end{align}
We proceed similarly as for repeatability. The errors can be split into contributions near and far from the boundary,
\begin{equation}
    S_i=S_i^{\mathrm{near}}+ S_i^{\mathrm{far}},
\end{equation}
where
\begin{align}
    S_i^{\mathrm{near}} &:=\int_{\partial_R Q_i}d^3\mathbf{k}\,|\braket{\mathbf{k}|\psi}|^2\int_{\mathbb R^3\setminus Q_i}d^3\mathbf{p}\,G_{\sigma}(\mathbf{k}-\mathbf{p}),\\
    S_i^{\mathrm{far}}
    &:=
    \int_{Q_i \setminus \partial_R Q_i}
    d^3\mathbf{k}\,
    |\braket{\mathbf{k}|\psi}|^2
    \int_{\mathbb R^3\setminus Q_i}
    d^3\mathbf{p}\,
    G_{\sigma}(\mathbf{k}-\mathbf{p}).
\end{align}

We first estimate the near-boundary term. Since $G_\sigma$ is normalized and non-negative,
\begin{align}
    \sum_{i=1}^N S_i^{\mathrm{near}} &\le \sum_{i=1}^N \int_{\partial_RQ_i}d^3\mathbf{k}\,|\braket{\mathbf{k}|\psi}|^2 \notag\\
    &= \eta_R^{\mathbf{\hat P}}(\psi,P).
\end{align}

We now estimate the far-from-boundary term. By definition, for $\mathbf{k} \in Q_i \setminus \partial_R Q_i$,
\begin{align}
    \int_{\mathbb R^3\setminus Q_i} d^3\mathbf{p}\,G_{\sigma}(\mathbf{k}-\mathbf{p}) &\le\int_{\left\{\mathbf{p}\in\mathbb R^3: d_{\sigma}^{\mathbf{\hat P}}(\mathbf{k},\mathbf{p})>R\right\}} d^3\mathbf{p}\, G_{\sigma}(\mathbf{k}-\mathbf{p}).
\end{align}
The right-hand side is independent of $\mathbf{k}$ by translation invariance. We denote it by
\begin{equation}
    \tau_{\mathbf {\hat P}}(R):=\int_{\left\{\mathbf{p}\in\mathbb R^3: d_{\sigma}^{\mathbf{\hat P}}(\mathbf{k},\mathbf{p})>R\right\}}d^3\mathbf{p}\,G_{\sigma}(\mathbf{k}-\mathbf{p}).
\end{equation}
It follows that
\begin{align}
    \sum_{i=1}^N S_i^{\mathrm{far}}
    &\le \tau_{\mathbf {\hat P}}(R)\int_{\mathbb R^3}d^3\mathbf{k}\,|\braket{\mathbf{k}|\psi}|^2 \notag\\
    &= \tau_{\mathbf {\hat P}}(R).
\end{align}

Using the definition of the error function, it can be shown that
\begin{align}
    \tau_{\mathbf {\hat P}}(R) =
    \frac{2R}{\sqrt{\pi}}e^{-R^2}+\operatorname{erfc}(R).
\end{align}
This proves the claim.
\end{proof}

\setlength{\bibsep}{-0.3pt}
\bibliography{bibliography}

@article{RevModPhys.21.400,
  title = {Localized States for Elementary Systems},
  author = {Newton, T. D. and Wigner, E. P.},
  journal = {Rev. Mod. Phys.},
  volume = {21},
  issue = {3},
  pages = {400--406},
  year = {1949},
  month = {Jul},
  publisher = {American Physical Society},
  doi = {10.1103/RevModPhys.21.400},
  url = {https://link.aps.org/doi/10.1103/RevModPhys.21.400}
}

@incollection{fulling_1989,
place={Cambridge},
series={London Mathematical Society Student Texts},
title={Quantization of a static, scalar field theory},
DOI={10.1017/CBO9781139172073.004},
booktitle={Aspects of Quantum Field Theory in Curved Spacetime},
publisher={Cambridge University Press},
author={Fulling, Stephen A.},
year={1989},
pages={48--73},
collection={London Mathematical Society Student Texts}
}

@article{PhysRev.139.B963,
  title = {Nonlocal Properties of Stable Particles},
  author = {Fleming, Gordon N.},
  journal = {Phys. Rev.},
  volume = {139},
  issue = {4B},
  pages = {B963--B968},
  numpages = {0},
  year = {1965},
  month = {Aug},
  publisher = {American Physical Society},
  doi = {10.1103/PhysRev.139.B963},
  url = {https://link.aps.org/doi/10.1103/PhysRev.139.B963}
}

@article{PhysRevD.10.3320,
  title = {Remark on causality and particle localization},
  author = {Hegerfeldt, Gerhard C.},
  journal = {Phys. Rev. D},
  volume = {10},
  issue = {10},
  pages = {3320--3321},
  numpages = {0},
  year = {1974},
  month = {Nov},
  publisher = {American Physical Society},
  doi = {10.1103/PhysRevD.10.3320},
  url = {https://link.aps.org/doi/10.1103/PhysRevD.10.3320}
}

@article{PhysRevD.22.377,
  title = {Remarks on causality, localization, and spreading of wave packets},
  author = {Hegerfeldt, Gerhard C. and Ruijsenaars, Simon N. M.},
  journal = {Phys. Rev. D},
  volume = {22},
  issue = {2},
  pages = {377--384},
  numpages = {0},
  year = {1980},
  month = {Jul},
  publisher = {American Physical Society},
  doi = {10.1103/PhysRevD.22.377},
  url = {https://link.aps.org/doi/10.1103/PhysRevD.22.377}
}

@incollection{Malament1996-MALIDO,
  author = {Malament, David B.},
  title = {In Defense of Dogma: Why There Cannot Be a Relativistic Quantum Mechanical Theory of (Localizable) Particles},
  booktitle = {Perspectives on Quantum Reality},
  editor = {Clifton, Rob},
  publisher = {Kluwer Academic Publishers},
  pages = {1--10},
  year = {1996},
  doi = {10.1007/978-94-015-8656-6_1}
}

@article{87c29b34-f4a4-31da-ad1b-ae81d42452e1,
 ISSN = {00318248, 1539767X},
 URL = {https://www.jstor.org/stable/10.1086/338939},
 author = {Hans Halvorson and Rob Clifton},
 journal = {Philosophy of Science},
 number = {1},
 pages = {1--28},
 publisher = {The University of Chicago Press, Philosophy of Science Association},
 title = {No Place for Particles in Relativistic Quantum Theories?},
 urldate = {2024-03-07},
 volume = {69},
 year = {2002}
}

@book{schrodinger1928collected,
  title={Collected Papers on Wave Mechanics},
  author={Schr{\"o}dinger, E.},
  note={Translated by hearer, J.F. and Deans, W.M.},
  isbn={9780598628084},
  lccn={29010168},
  year={1928},
  publisher={Blackie \& Son limited}
}

@Article{Zeh1970,
author={Zeh, H. D.},
title={On the interpretation of measurement in quantum theory},
journal={Foundations of Physics},
year={1970},
month={Mar},
day={01},
volume={1},
number={1},
pages={69-76},
issn={1572-9516},
doi={10.1007/BF00708656},
url={https://doi.org/10.1007/BF00708656}
}

@article{PhysRevD.24.1516,
  title = {Pointer basis of quantum apparatus: Into what mixture does the wave packet collapse?},
  author = {Zurek, W. H.},
  journal = {Phys. Rev. D},
  volume = {24},
  issue = {6},
  pages = {1516--1525},
  numpages = {0},
  year = {1981},
  month = {Sep},
  publisher = {American Physical Society},
  doi = {10.1103/PhysRevD.24.1516},
  url = {https://link.aps.org/doi/10.1103/PhysRevD.24.1516}
}

@article{PhysRevD.26.1862,
  title = {Environment-induced superselection rules},
  author = {Zurek, W. H.},
  journal = {Phys. Rev. D},
  volume = {26},
  issue = {8},
  pages = {1862--1880},
  numpages = {0},
  year = {1982},
  month = {Oct},
  publisher = {American Physical Society},
  doi = {10.1103/PhysRevD.26.1862},
  url = {https://link.aps.org/doi/10.1103/PhysRevD.26.1862}
}

@Article{Joos1985,
author={Joos, E.
and Zeh, H. D.},
title={The emergence of classical properties through interaction with the environment},
journal={Zeitschrift f{\"u}r Physik B Condensed Matter},
year={1985},
month={Jun},
day={01},
volume={59},
number={2},
pages={223-243},
issn={1431-584X},
doi={10.1007/BF01725541},
url={https://doi.org/10.1007/BF01725541}
}

@article{RevModPhys.75.715,
  title = {Decoherence, einselection, and the quantum origins of the classical},
  author = {Zurek, Wojciech Hubert},
  journal = {Rev. Mod. Phys.},
  volume = {75},
  issue = {3},
  pages = {715--775},
  numpages = {0},
  year = {2003},
  month = {May},
  publisher = {American Physical Society},
  doi = {10.1103/RevModPhys.75.715},
  url = {https://link.aps.org/doi/10.1103/RevModPhys.75.715}
}

@article{PhysRevLett.99.180403,
  title = {Classical World Arising out of Quantum Physics under the Restriction of Coarse-Grained Measurements},
  author = {Kofler, Johannes and Brukner, {\v C}aslav},
  journal = {Phys. Rev. Lett.},
  volume = {99},
  issue = {18},
  pages = {180403},
  numpages = {4},
  year = {2007},
  month = {Nov},
  publisher = {American Physical Society},
  doi = {10.1103/PhysRevLett.99.180403},
  url = {https://link.aps.org/doi/10.1103/PhysRevLett.99.180403}
}

@book{duncan2012conceptual,
  author = {Duncan, Anthony},
  title = {The Conceptual Framework of Quantum Field Theory},
  publisher = {Oxford University Press},
  year = {2012},
  month = {aug},
  isbn = {9780199573264},
  doi = {10.1093/acprof:oso/9780199573264.001.0001},
  url = {https://doi.org/10.1093/acprof:oso/9780199573264.001.0001}
}

@article{PhysRevD.53.7327,
  title = {Decoherence of quantum fields: Pointer states and predictability},
  author = {Anglin, J. R. and Zurek, W. H.},
  journal = {Phys. Rev. D},
  volume = {53},
  issue = {12},
  pages = {7327--7335},
  numpages = {0},
  year = {1996},
  month = {Jun},
  publisher = {American Physical Society},
  doi = {10.1103/PhysRevD.53.7327},
  url = {https://link.aps.org/doi/10.1103/PhysRevD.53.7327}
}

@article{KUBLER1973405,
title = {Dynamics of quantum correlations},
journal = {Annals of Physics},
volume = {76},
number = {2},
pages = {405-418},
year = {1973},
issn = {0003-4916},
doi = {https://doi.org/10.1016/0003-4916(73)90040-7},
url = {https://www.sciencedirect.com/science/article/pii/0003491673900407},
author = {O K{\"u}bler and H.D Zeh}
}

@article{PhysRev.131.2766,
  title = {Coherent and Incoherent States of the Radiation Field},
  author = {Glauber, Roy J.},
  journal = {Phys. Rev.},
  volume = {131},
  issue = {6},
  pages = {2766--2788},
  numpages = {0},
  year = {1963},
  month = {Sep},
  publisher = {American Physical Society},
  doi = {10.1103/PhysRev.131.2766},
  url = {https://link.aps.org/doi/10.1103/PhysRev.131.2766}
}

@misc{kofler2010macroscopic,
  author = {Kofler, Johannes and Buri{\'c}, Nikola and Brukner, {\v C}aslav},
  title = {Macroscopic realism and spatiotemporal continuity},
  year = {2010},
  eprint = {0906.4465},
  archivePrefix = {arXiv},
  primaryClass = {quant-ph},
  url = {https://arxiv.org/abs/0906.4465}
}

@misc{bibak2026classicallimitquantummechanics,
      title={The classical limit of quantum mechanics through coarse-grained measurements}, 
      author={Fatemeh Bibak and Carlo Cepollaro and Nicol{\'a}s Medina S{\'a}nchez and Borivoje Daki{\'c} and {\v C}aslav Brukner},
      year={2026},
      eprint={2503.15642},
      archivePrefix={arXiv},
      primaryClass={quant-ph},
      url={https://arxiv.org/abs/2503.15642}, 
}

@Article{Omnes1988,
author={Omn{\`e}s, Roland},
title={{Logical reformulation of quantum mechanics. III. Classical limit and irreversibility}},
journal={Journal of Statistical Physics},
year={1988},
month={Nov},
day={01},
volume={53},
number={3--4},
pages={957-975},
issn={1572-9613},
doi={10.1007/BF01014232},
url={https://doi.org/10.1007/BF01014232}
}

@book{peres1995quantum,
  title={Quantum Theory: Concepts and Methods},
  author={Peres, A.},
  isbn={9780792336327},
  lccn={93032994},
  series={Fundamental Theories of Physics},
  year={1995},
  publisher={Springer Dordrecht}
}

@article{perelomov1977generalized,
  author = {Perelomov, Askold M.},
  title = {Generalized coherent states and some of their applications},
  journal = {Soviet Physics Uspekhi},
  volume = {20},
  number = {9},
  pages = {703--720},
  year = {1977},
  doi = {0.1070/PU1977v020n09ABEH005459}
}

@article{PhysRevLett.112.010402,
  title = {Coarsening Measurement References and the Quantum-to-Classical Transition},
  author = {Jeong, Hyunseok and Lim, Youngrong and Kim, M. S.},
  journal = {Phys. Rev. Lett.},
  volume = {112},
  issue = {1},
  pages = {010402},
  numpages = {5},
  year = {2014},
  month = {Jan},
  publisher = {American Physical Society},
  doi = {10.1103/PhysRevLett.112.010402},
  url = {https://link.aps.org/doi/10.1103/PhysRevLett.112.010402}
}

@book{weinberg1995quantum,
  author    = {Weinberg, Steven},
  title     = {The Quantum Theory of Fields. Volume I: Foundations},
  publisher = {Cambridge University Press},
  year      = {1995}
}

@article{PhysRevLett.70.1187,
  title = {Coherent states via decoherence},
  author = {Zurek, Wojciech H. and Habib, Salman and Paz, Juan Pablo},
  journal = {Phys. Rev. Lett.},
  volume = {70},
  issue = {9},
  pages = {1187--1190},
  year = {1993},
  month = {Mar},
  publisher = {American Physical Society},
  doi = {10.1103/PhysRevLett.70.1187},
  url = {https://link.aps.org/doi/10.1103/PhysRevLett.70.1187}
}

@article{PhysRevE.50.2538,
  title = {Decoherence produces coherent states: An explicit proof for harmonic chains},
  author = {Tegmark, Max and Shapiro, Harold S.},
  journal = {Phys. Rev. E},
  volume = {50},
  issue = {4},
  pages = {2538--2547},
  numpages = {0},
  year = {1994},
  month = {Oct},
  publisher = {American Physical Society},
  doi = {10.1103/PhysRevE.50.2538},
  url = {https://link.aps.org/doi/10.1103/PhysRevE.50.2538}
}

@article{PhysRevLett.85.3552,
  title = {Robustness and Diffusion of Pointer States},
  author = {Di\'osi, Lajos and Kiefer, Claus},
  journal = {Phys. Rev. Lett.},
  volume = {85},
  issue = {17},
  pages = {3552--3555},
  year = {2000},
  month = {Oct},
  publisher = {American Physical Society},
  doi = {10.1103/PhysRevLett.85.3552},
  url = {https://link.aps.org/doi/10.1103/PhysRevLett.85.3552}
}

@incollection{landsman2007between,
  author       = {Landsman, Nicolaas P.},
  title        = {Between Classical and Quantum},
  booktitle    = {Philosophy of Physics},
  editor       = {Butterfield, Jeremy and Earman, John},
  series       = {Handbook of the Philosophy of Science},
  volume       = {2},
  pages        = {417--553},
  publisher    = {North-Holland},
  address      = {Amsterdam},
  year         = {2007},
  doi          = {10.1016/B978-044451560-5/50008-7},
  eprint       = {quant-ph/0506082},
  archivePrefix = {arXiv},
  primaryClass = {quant-ph}
}

@book{cohn2013measure,
  title={Measure theory},
  author={Cohn, Donald L},
  volume={2},
  year={2013},
  publisher={Springer}
}

@book{goldstein2002classical,
  author    = {Goldstein, Herbert and Poole, Charles P. and Safko, John L.},
  title     = {Classical Mechanics},
  edition   = {3},
  publisher = {Addison-Wesley},
  year      = {2002},
  isbn      = {978-0-201-65702-9}
}

@book{landau1975classical,
  author    = {Landau, L. D. and Lifshitz, E. M.},
  title     = {The Classical Theory of Fields},
  series    = {Course of Theoretical Physics},
  volume    = {2},
  edition   = {4},
  publisher = {Pergamon Press},
  year      = {1975},
  isbn      = {978-0-08-018176-9}
}

@book{fetter1971quantum,
  author    = {Fetter, Alexander L. and Walecka, John Dirk},
  title     = {Quantum Theory of Many-Particle Systems},
  publisher = {McGraw-Hill},
  address   = {New York},
  year      = {1971}
}

@Article{Hepp1974,
author={Hepp, Klaus},
title={The classical limit for quantum mechanical correlation functions},
journal={Communications in Mathematical Physics},
year={1974},
month={Dec},
day={01},
volume={35},
number={4},
pages={265-277},
issn={1432-0916},
doi={10.1007/BF01646348},
url={https://doi.org/10.1007/BF01646348}
}

@book{Billingsley1995,
  author    = {Patrick Billingsley},
  title     = {Probability and Measure},
  edition   = {3},
  publisher = {John Wiley \& Sons},
  year      = {1995}
}

@book{Bogachev2007,
  author    = {Vladimir I. Bogachev},
  title     = {Measure Theory},
  volume    = {1},
  publisher = {Springer},
  year      = {2007}
}

@book{Kallenberg2021,
  author    = {Olav Kallenberg},
  title     = {Foundations of Modern Probability},
  edition   = {3},
  publisher = {Springer},
  year      = {2021}
}

@book{Kolmogorov1956,
  author    = {A. N. Kolmogorov},
  title     = {Foundations of the Theory of Probability},
  publisher = {Chelsea Publishing Company},
  year      = {1956},
  note      = {Second English edition}
}

@book{holevo2011probabilistic,
  title={Probabilistic and statistical aspects of quantum theory},
  author={Holevo, Alexander S},
  volume={1},
  year={2011},
  publisher={Springer Science \& Business Media}
}

@article{wigner1939unitary,
  author = {Wigner, Eugene P.},
  title = {{On unitary representations of the inhomogeneous Lorentz group}},
  journal = {Annals of Mathematics},
  volume = {40},
  number = {1},
  pages = {149--204},
  year = {1939},
  doi = {10.2307/1968551}
}

@article{luders1950zustandsanderung,
  title={{\"U}ber die Zustands{\"a}nderung durch den Me{\ss}proze{\ss}},
  author={L{\"u}ders, Gerhart},
  journal={Annalen der Physik},
  volume={443},
  number={5-8},
  pages={322--328},
  year={1950},
  doi = {10.1002/andp.19504430510}
}

@article{shannon1948mathematical,
  title={A mathematical theory of communication},
  author={Shannon, Claude Elwood},
  journal={The Bell system technical journal},
  volume={27},
  number={3},
  pages={379--423},
  year={1948},
  publisher={Nokia Bell Labs}
}

@book{cover1999elements,
  title={Elements of information theory},
  author={Cover, Thomas M and Thomas Joy A.},
  year={1999},
  publisher={John Wiley \& Sons}
}

@article{Emary_2014,
doi = {10.1088/0034-4885/77/1/016001},
url = {https://doi.org/10.1088/0034-4885/77/1/016001},
year = {2013},
month = {dec},
publisher = {IOP Publishing},
volume = {77},
number = {1},
pages = {016001},
author = {Emary, Clive and Lambert, Neill and Nori, Franco},
title = {{Leggett--Garg inequalities}},
journal = {Reports on Progress in Physics}
}

@article{PhysRevLett.54.857,
  title = {Quantum mechanics versus macroscopic realism: Is the flux there when nobody looks?},
  author = {Leggett, A. J. and Garg, Anupam},
  journal = {Phys. Rev. Lett.},
  volume = {54},
  issue = {9},
  pages = {857--860},
  year = {1985},
  month = {Mar},
  publisher = {American Physical Society},
  doi = {10.1103/PhysRevLett.54.857},
  url = {https://link.aps.org/doi/10.1103/PhysRevLett.54.857}
}

@article{DaviesLewis1970,
  author = {Davies, E. B. and Lewis, J. T.},
  title = {An operational approach to quantum probability},
  journal = {Communications in Mathematical Physics},
  volume = {17},
  pages = {239--260},
  year = {1970},
  doi = {10.1007/BF01647093}
}

@article{BuschLahti1996,
  author = {Busch, Paul and Lahti, Pekka J.},
  title = {The Standard Model of Quantum Measurement Theory: History and Applications},
  journal = {Foundations of Physics},
  volume = {26},
  number = {7},
  pages = {875--893},
  year = {1996},
  doi = {10.1007/BF02148831}
}

@article{schrodinger1926stetige,
  title={Der stetige {\"U}bergang von der Mikro - zur Makromechanik},
  author={Schr{\"o}dinger, Erwin},
  journal={naturwissenschaften},
  volume={14},
  number={28},
  pages={664--666},
  year={1926},
  publisher={Springer-Verlag Berlin/Heidelberg}
}

@article{zhang1990coherent,
  author = {Zhang, Wei-Min and Feng, Da Hsuan and Gilmore, Robert},
  title = {Coherent states: Theory and some applications},
  journal = {Rev. Mod. Phys.},
  volume = {62},
  number = {4},
  pages = {867--927},
  year = {1990},
  doi = {10.1103/RevModPhys.62.867},
  url = {https://link.aps.org/doi/10.1103/RevModPhys.62.867}
}

@article{KochenSpecker1967,
 ISSN = {00959057, 19435274},
 URL = {http://www.jstor.org/stable/24902153},
 author = {Simon Kochen and E. P. Specker},
 journal = {Journal of Mathematics and Mechanics},
 number = {1},
 pages = {59--87},
 publisher = {Indiana University Mathematics Department},
 title = {The Problem of Hidden Variables in Quantum Mechanics},
 urldate = {2026-06-08},
 volume = {17},
 year = {1967}
}

@article{Spekkens2005,
  title = {Contextuality for preparations, transformations, and unsharp measurements},
  author = {Spekkens, R. W.},
  journal = {Phys. Rev. A},
  volume = {71},
  issue = {5},
  pages = {052108},
  numpages = {17},
  year = {2005},
  month = {May},
  publisher = {American Physical Society},
  doi = {10.1103/PhysRevA.71.052108},
  url = {https://link.aps.org/doi/10.1103/PhysRevA.71.052108}
}

@article{AbramskyBrandenburger2011,
doi = {10.1088/1367-2630/13/11/113036},
url = {https://doi.org/10.1088/1367-2630/13/11/113036},
year = {2011},
month = {nov},
publisher = {IOP Publishing},
volume = {13},
number = {11},
pages = {113036},
author = {Abramsky, Samson and Brandenburger, Adam},
title = {The sheaf-theoretic structure of non-locality and contextuality},
journal = {New Journal of Physics}
}

@misc{FedidaGlowacki2026,
      title={Foundations of Relational Quantum Field Theory I: Scalars}, 
      author={Samuel Fedida and Jan G{\l}owacki},
      year={2026},
      eprint={2507.21601},
      archivePrefix={arXiv},
      primaryClass={quant-ph},
      url={https://arxiv.org/abs/2507.21601}, 
}

@article{kowalski2018coherent,
  author = {Kowalski, Krzysztof and Rembieli{\'n}ski, Jakub and Gazeau, Jean-Pierre},
  title = {On the coherent states for a relativistic scalar particle},
  journal = {Annals of Physics},
  volume = {399},
  pages = {204--223},
  year = {2018},
  doi = {10.1016/j.aop.2018.10.014}
}

@article{andrzejewski2015chiral,
  title={Chiral fermions, massless particles and Poincar{\'e} covariance},
  author={Andrzejewski, Krzysztof and Kijanka-Dec, Agnieszka and Kosi{\'n}ski, Piotr and Ma{\'s}lanka, Pawe{\l}},
  journal={Physics Letters B},
  volume={746},
  pages={417--423},
  year={2015},
  publisher={Elsevier}
}

@article{PhysRevLett.101.090403,
  title = {Conditions for Quantum Violation of Macroscopic Realism},
  author = {Kofler, Johannes and Brukner, {\v C}aslav},
  journal = {Phys. Rev. Lett.},
  volume = {101},
  issue = {9},
  pages = {090403},
  numpages = {4},
  year = {2008},
  month = {Aug},
  publisher = {American Physical Society},
  doi = {10.1103/PhysRevLett.101.090403},
  url = {https://link.aps.org/doi/10.1103/PhysRevLett.101.090403}
}

@article{wightman1962localizability,
  author = {Wightman, A. S.},
  title = {On the Localizability of Quantum Mechanical Systems},
  journal = {Rev. Mod. Phys.},
  volume = {34},
  number = {4},
  pages = {845--872},
  year = {1962},
  doi = {10.1103/RevModPhys.34.845}
}

@article{Castrigiano1981,
  author = {Castrigiano, D. P. L.},
  title = {On euclidean systems of covariance for massless particles},
  journal = {Letters in Mathematical Physics},
  volume = {5},
  number = {4},
  pages = {303--309},
  year = {1981},
  doi = {10.1007/BF00401478}
}

\end{document}